%% file: body.tex
\documentclass[12pt,svgnames,oneside]{book}

\usepackage{packages}
\graphicspath{{./figures/}}

\begin{document}

\frontmatter
\input{./Boundary/First_Pages.tex}

\tableofcontents
\listoffigures
\input{./apendices/nomenclature.tex}\chaptermark{}
\mainmatter

\chapter{Introduction and brief historical background}\chaptermark{Introduction}\label{chap:intro}
\input{./Boundary/Introduction.tex}

\chapter{One-dimensional blood flow modelling with non-rigid artery}\label{chap:model}
\chaptermark{1D modelling}
\begin{figure}[htb]
    \centering
    \includegraphics[trim={0 20cm 0 0},clip,width=.7\linewidth]{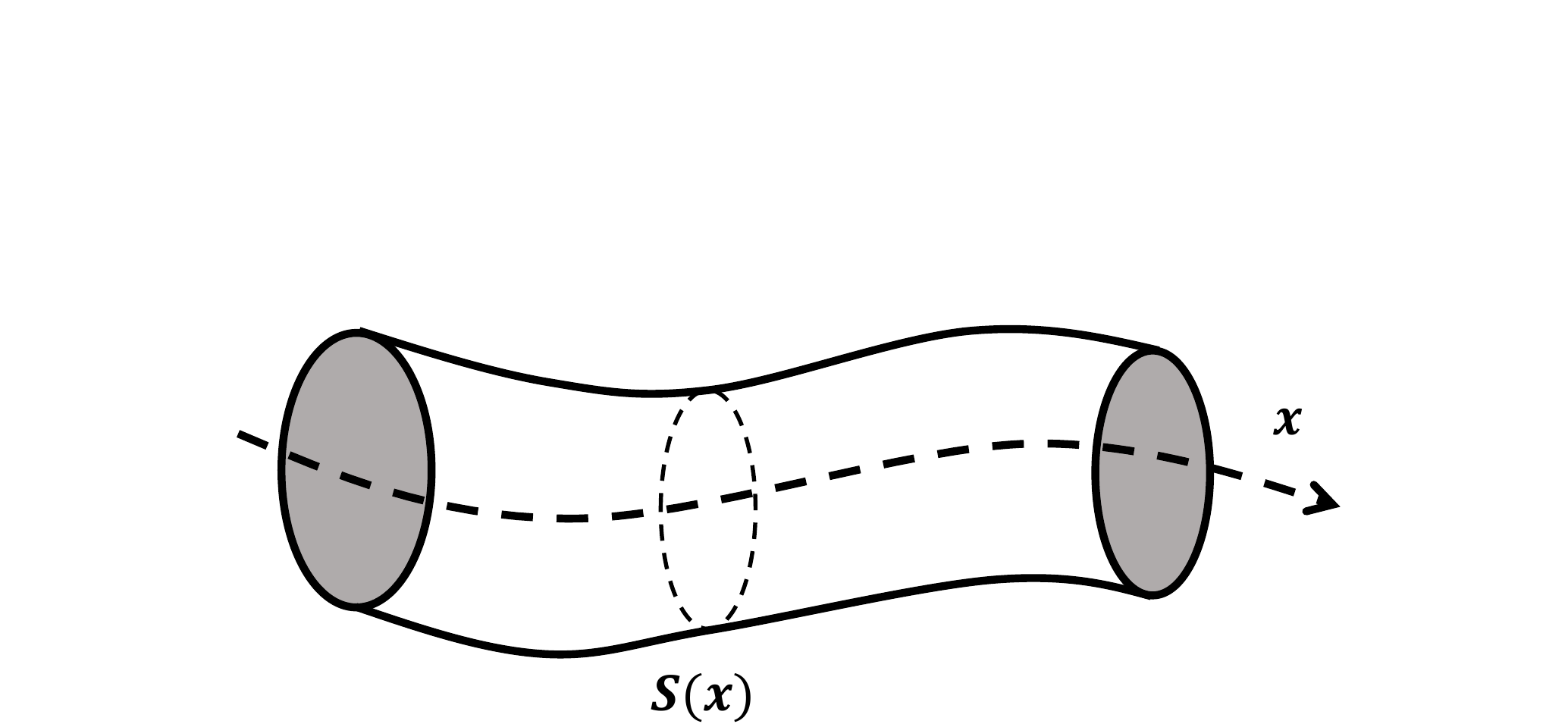}
    \caption{Artery as a compliant tube.}
    \label{fig:compliant}
\end{figure}

The object to model is a simplification of an artery as a compliant cylinder, illustrated in figure~\ref{fig:compliant}. We will start from a 3D reasoning and we will be making assumptions and simplifications until we arrive to the one-dimensional version of the equations. This way of modelling non-rigid arteries was presented in \(2003\) by Sherwin \textit{et al.} in~\cite{sherwin03} and~\cite{sherwin03b}.

The rest of the chapter will consist on motivating and detailing the main assumptions of the model. As we will see, a system of partial differential equations will appear, and each one of the equations will be detailed in different subsections. Before presenting them, we start with the notation and the main variables.
\section{Governing equations}\label{section:equations}
\input{./Model/Equations.tex}

\chapter{Theoretical analysis applied to the Navier-Stokes problem}\label{chap:theo}
\chaptermark{Theoretical analysis}
In this chapter we will use some mathematical tools to analyse the system of partial differential equations presented in the previous chapter. This analysis will be mainly based on a widely known (specially in engineering areas) method, the so-called method of characteristics. Although it is sometimes used as a numerical method, we will use the theoretical background of it (explained in section~\ref{sec:charac}) to prove some useful theorems in the section~\ref{sec:PDE}.
\section{Method of characteristics}\label{sec:charac}
\input{./Theory/Characteristics.tex}
\section{Sufficient conditions for smooth flow}\label{sec:PDE}
\input{./Theory/PDE_Stuff.tex}

\chapter{Discontinuous Galerkin Method}\label{chap:galerkin}
\input{./Numeric/Why_Galerkin.tex}
\section{Notation}\label{sec:galerkin_notation}
\input{./Numeric/Notation_Galerkin.tex}
\section{DG space (semi)discretisation}\label{sec:DG}
\input{./Numeric/Galerkin.tex}

\chapter{Implementation and results}\label{chap:results}
\input{./Computational/Integrations.tex}
\input{./Computational/Experiments.tex}

\chapter{Conclusions and future work}
\input{./Boundary/Conclusions.tex}

\appendix
\input{./apendices/biological_parameters.tex}
\backmatter

\bibliographystyle{abbrv}
\bibliography{./apendices/database.bib}
\end{document}

%% file: Boundary/First_Pages.tex
\begin{titlepage}

\begin{center}

\Huge
\textbf{Analysis of blood flow in one dimensional elastic artery using Navier-Stokes conservation laws} \\

\vspace*{2cm}
\begin{figure}[H]
\begin{center}
\includegraphics[width=4cm]{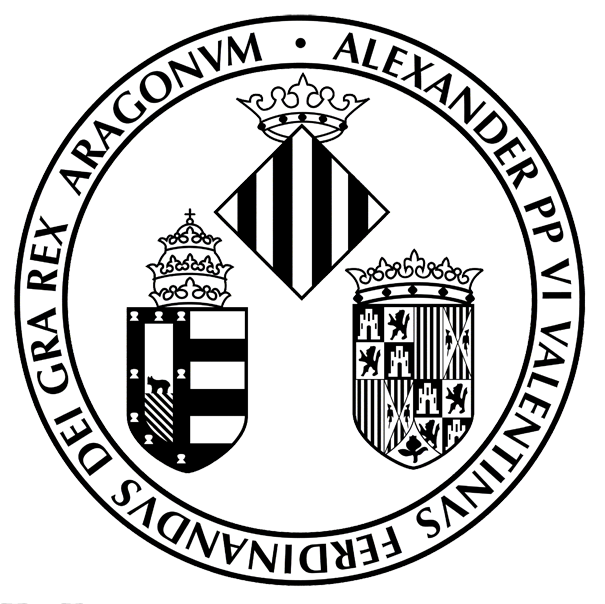}\qquad\qquad\qquad
\includegraphics[width=4cm]{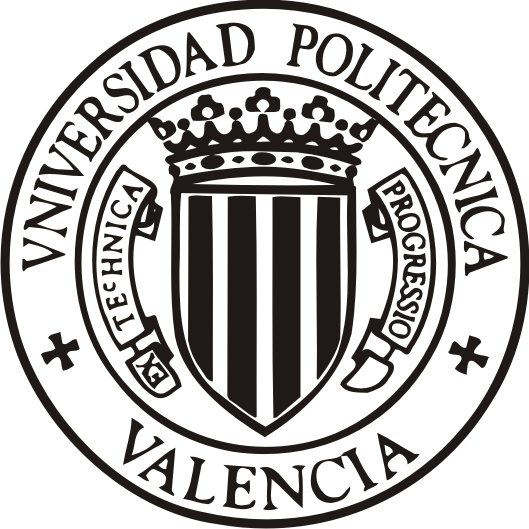}
  \end{center}
\end{figure}
\vspace*{1cm}
\begin{figure}[H]
    \centering
\includegraphics[width=\linewidth]{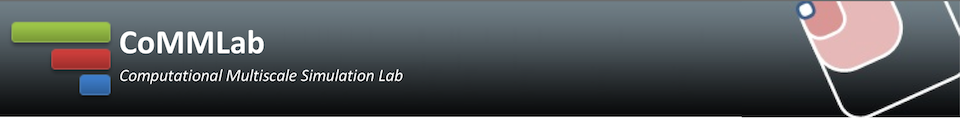}
\end{figure}
\vspace*{1cm}

\begin{Large}
Student: \textbf{Crist\'{o}bal Rodero G\'{o}mez}\textsuperscript{1,3}\\

\vspace*{0.3cm}
\hspace*{-1.6cm} Tutors: \textbf{J. Alberto Conejero}\textsuperscript{2}\\
\qquad\qquad \textbf{Ignacio Garc\'ia-Fern\'andez}\textsuperscript{3}

\end{Large}

\vfill
\end{center}

\large
\raggedright
\textsuperscript{1}Facultat de Matemàtiques, Universitat de València\\
\textsuperscript{2}IUMPA, Universitat Politècnica de València\\
\textsuperscript{3}CoMMLab, Departament d'Informàtica, Universitat de València
\end{titlepage}

\null\vfill
\textit{``Remember that all models are wrong; the practical question is how wrong do they have to be to not be useful.''}

\begin{flushright}
George Box, statistician.
\end{flushright}

\vfill\vfill\vfill\vfill\vfill\vfill\null
\textit{``Applied math pattern: In principle you could just ... but here's why that won't work in practice, and what you need to do instead.''}

\begin{flushright}
@AnalysisFact, tweet from Aug. 14th, 2017.
\end{flushright}

\vfill\vfill\vfill\vfill\vfill\vfill\null

\clearpage

\chapter*{Abstract}
\addtocontents{toc}{\vspace{1em}}

In the last years, medical computer simulation has seen a great growth in several scientific branches, from modelling to numerical methods, going through computer science. The main goals of this incipient discipline are testing hypotheses before an intervention, or see what effect could have a drug in the system before actually taking it, among others. \\
In this work we deduce from the most basic physical principles a one dimensional model for the simulation of blood flow in elastic arteries. We will provide some historical background, as well as a brief state of the art of these models. We will also study from a calculus point of view the equations of the model obtained, achieving an original result for the formation of shock waves in compliant vessels. Afterwards we will make some numerical simulations using Galerkin Discontinuous Finite Element Method. Since this is actually a family of methods, we will motivate and detail the elections and the implementation strategies.

\clearpage  

%% file: apendices/nomenclature.tex
\nomenclature{\(x\)}{Spatial coordinate.}
\nomenclature{\(t\)}{Temporal coordinate.}
\nomenclature{\(A\)}{Amplitude of the artery.}
\nomenclature{\(u\)}{Velocity of the blood flow.}
\nomenclature{\(p\)}{Pressure of the flow.}
\nomenclature{\(Q\)}{Flow flux.}
\nomenclature{\(\rho\)}{Constant blood density.}
\nomenclature{\(\mu\)}{Constant blood viscosity.}
\nomenclature{\(f\)}{Friction of the blood with the artery walls.}
\nomenclature{\(\alpha\)}{Coriolis coefficient.}
\nomenclature{\(p_0,p_{\textrm{ext}},p_{\textrm{ref}}\)}{Initial, external and reference pressure.}
\nomenclature{\(\sigma\)}{Stress tensor.}
\nomenclature{\(\varepsilon\)}{Strain tensor.}
\nomenclature{\(\nu\)}{Poisson's ratio.}
\nomenclature{\(E\)}{Young's modulus.}
\nomenclature{\(R(\cdot,\cdot)\)}{Radius of the artery.}
\nomenclature{\(R_0(\cdot,\cdot)\)}{Rest radius of the artery.}
\nomenclature{\(h_0\)}{Arterial wall thickness.}
\nomenclature{\(\beta\)}{Stiffness coefficient.}

\printnomenclature

%% file: Boundary/Introduction.tex
In the last years, medical computer simulation has seen a great growth in several scientific branches, from modelling to numerical methods, going through computer science. The main goals of this incipient discipline are testing hypotheses before an intervention, or see what effect could have a drug in the system before actually taking it, among others. In this chapter we will name some of the most important contributors to this discipline. We will focus on the physics' (specially mechanics') point of view, and for historical reasons, up to \(20\)\textsuperscript{th} century. In subsequent chapters the most recent computational scientific progress will be presented.

\section{Historical review}

For obvious reasons, one of the most important systems to be simulated is the cardiovascular system. Here, medicine works together with physics: electrophysiology (the nerve impulses that stimulates the heart), elasticity theory (in the movement of the heart, the arteries\ldots) or fluid dynamics (the blood's behaviour). The first models are often simplified versions of reality, neglecting some effects or reducing dimensionality. Taking this into account, the modelling of human arterial system can be traced back to Euler in \(1775\), who submitted an essay as an entry in a prize competition set by the Academy of Sciences in Dijon~\cite{euler75}. He derived a one-dimensional simplification using partial differential equations, arriving to the equations of conservation of mass and momentum --- they will be explained later --- in a distensible tube. Euler posited some rather unrealistic constitutive laws (tube laws) for arteries and unsuccessfully tried to solve the equations. He tried to solve the problem as he had done for rigid tubes: by reducing them to a single equation that could be solved by integration. Quoting, ``\textit{In motu igitur sanguinis explicando easdem offendimus insuperabiles difficultates[\ldots]}'', loosely translated as ``On the explanation of the blood motion we stumble upon the same insuperable difficulties''. This letter was lost over a century, being discovered and published by the Euler Opera postuma project in \(1862\).

Now, according to~\cite{parker09}, the next major event in the history of quantitative haemodynamics is the lecture delivered to the Royal Society in \(1808\) by Young~\cite{young09}. In the lecture, he stated the correct formula for the wave speed in an artery but gave no derivation of it. In an associated paper, he does give a derivation based on an analogy to Newton's derivation of the speed of sound in a compressible gas, altogether with some numerical guesses~\cite{young08}.

The development by Poiseuille (\(1799\)–-\(1869\)) of his law of flow in tubes is the next landmark in arterial mechanics. Because of its simplicity, this law has become the benchmark against which all other flows in tubes are
compared, although in arteries is quite difficult to observe it. Despite its shortcomings, it is cited by many medical and physiological textbooks as the law that governs flow in the whole of vasculature. Poiseuille, who trained as a physician, conducted a very thorough investigation of flow in capillary tubes motivated by his studies of the mesenteric microcirculation of the frog. All his experiments and conclusions were finally approved for publication in \(1846\)~\cite{poiseuille44}. It seems that Stokes also derived Poiseuille's law from the Navier-Stokes equation as early as \(1845\) but did not publish the work because he was unsure about the validity of the no-slip condition at the tube walls~\cite{stokes45}.

The question of the speed of travel of waves in elastic tubes was studied theoretically and experimentally by Weber and Weber and published in 1866~\cite{weber25,weber66}.

Riemann (\(1826\)--\(1866\)) did not work on arterial mechanics or waves in elastic tubes, but he did make an important contribution to the subject when he published a general solution for hyperbolic systems of partial differential equations in \(1853\)~\cite{riemann53}. Briefly, his work provides a general solution for a whole class of linear and nonlinear partial differential equations by observing that along directions defined by the eigenvalues of the matrix of coefficients of the differential terms, the partial differential equations reduce to ordinary differential equations. This method (the method of characteristics) will be explained in a more detailed way in subsequent chapters. The first application of the theory to arterial flows is probably the work of Lambert, who applied the theory to arteries using experimental measurements of the radius of the artery as a function of pressure~\cite{lambert58}. The approach was developed by Skalak~\cite{skalak72} and most completely by Anliker and his colleagues who mounted a systematic study of the different elements of the vascular system with the goal of synthesising a complete description of the arterial system using the method of characteristics~\cite{anliker71,hisland73}~\cite{stettler81,stettler81b}.

In \(1877\)--\(1878\), two more important works on the wave speed in elastic tubes were published. Moens (\(1846\)--\(1891\))~\cite{moens77} published a very careful experimental paper on wave speed in arteries and Korteweg (\(1848\)--\(1941\))~\cite{korteweg78} published a theoretical study of the wave speed. Korteweg’s analysis showed that the wave speed was determined both by the elasticity of the tube wall and the compressibility of the fluid.

Frank (\(1865\)--\(1944\)) was another important contributor in quantitative physiology. One of his major contributions is a series of three papers. In~\cite{frank91}, he introduces the theory of waves in arteries. In~\cite{frank20}, he correctly derives the wave speed in terms of the elasticity. Finally, in~\cite{frank26}, he considers the effect of viscosity, the motion of the wall and the energy of the pulse wave before turning to a number of examples of special cases. These examples include the use of Fourier analysis and probably the first treatment of the reflections of the pulse wave, including the reflection and transmission coefficients due to a bifurcation.

Many clinical cardiologists in the early twentieth century contributed to our understanding of the form and function of the cardiovascular system, but relatively few contributed significantly to our understanding of arterial mechanics (see~\cite{parker09} and the references therein).

\section{Objectives and structure of the work}

Now that a brief historical background has been introduced we notice that we are dealing with a non trivial and relevant problem. The main aim of this work is, therefor, to achieve an efficient simulation of blood flow in compliant arteries. In order to do that, the physical model will be presented for a further mathematical analysis. Once the main features of this analysis have been extracted, via an appropriate numerical method, some simulations will be performed. A review of some of the most important models, both numerical and theoretical has been made along the work.

Hence, this document is structured as follow. In chapter~\ref{chap:model} we present the main equations for the modelling of arterial flow. We will present the widely accepted conservation equations in section~\ref{section:equations}. In the same section, the main problem of no scientific agreement will be presented, along with a brief state of the art. When the model is stated, we will try to use some theoretical results before numerical implementation. This will be done in chapter~\ref{chap:theo}. In this chapter we present a useful theoretical analysis, based on the method of characteristics (detailed in section~\ref{sec:charac}). The reader can find in this chapter a theorem regarding shock waves in compliant arteries, genuine of this work. This result was presented in \(2017\) in the \textit{4º Congreso de jóvenes investigadores} (IV Conference for young researchers) (see~\cite{rodero17}). Afterwards, we will motivate and present the numerical scheme used for the simulations: Galerkin Discontinuous Finite Element Method. All the necessary details for the understanding of this method will be presented in chapter~\ref{chap:galerkin}. Thereafter, we can find in chapter~\ref{chap:results} and subsequents some numerical results validating the model and studying its stability, convergence and sensitivity to the parameters. In the appendix~\ref{appendix:bio}, the reader can find a list of reasonable values for the parameters of the model. Some of these values are used through this work, but for completeness all the possible parameters have been listed and referenced. 

%% file: Model/Equations.tex
 The first simplification will be to assume that the local curvature of the artery is small enough so we can indeed reduce the problem to one dimension. If we denote by \(S(x)\) a cross section (we can think of it as a slice of the artery) we define
\begin{equation}
    A(x,t)=\int_{S(x)}\dd\sigma,
\end{equation}

as the area of the cross section \(S\);
\begin{equation}
    u(x,t)=\frac{1}{A(x,t)}\int_{S(x)}\hat{u}(x,t)\dd\sigma,
\end{equation}

as the average velocity over the cross section where \(\hat{u}(x,t)\) denote the value of velocity within a constant \(x\)-section; and
\begin{equation}
    p(x,t)=\frac{1}{A(x,t)}\int_{S(x)}\hat{p}(x,t)\dd\sigma,
\end{equation}

as the internal pressure over the cross section where \(\hat{p}(x,t)\) denote the value of pressure within a constant \(x\)-section. As it is usual in the literature, \(x\) will denote the spatial coordinate and \(t\) the temporal one. We remark that in order to avoid a cumbersome notation, if no confusion is added, some arguments of the functions could be avoided. Hence, in some places of this work we could write \(A,u\) and \(p\) for the functions previously presented. This will allow us to treat them as variables. We will also assume that the blood is an incompressible and Newtonian fluid and so the density \(\rho\) and dynamic viscosity \(\mu\) are constant. 

Finally, for the derivation of the dynamics equations we introduce the dependent variable
\begin{equation}
Q(x,t) = A(x,t)u(x,t),
\end{equation}

that will represent the volume flux at a given section. Therefor, since we have three variables, \(A,u\) and \(p\), or equivalently \(A,Q\) and \(p\), we need three equations to relate them. The first two will be conservative equations, \textit{i.e.}, equations that express the conservation of some quantity. The third equation will be the responsible for modelling the artery as an elastic material.

\subsection{Continuity equation}

Here we use the fluid dynamics' continuum hypothesis what is an idealisation of continuum mechanics under which fluids can be treated as continuous, even though, on a microscopic scale, they are composed of molecules. Under the continuum assumption, macroscopic properties such as density, pressure, and velocity are taken to be well-defined at infinitesimal volume elements --- small in comparison to the characteristic length scale of the system, but large in comparison to molecular length scale. Fluid properties can vary continuously from one volume element to another and they are averaged values of the molecular properties. 

With this, if we take one portion of the artery as our control volume, conservation of mass requires that, if there are nor sources neither sinks, nothing disappears spontaneously, \textit{i.e.}, the rate of change of mass within the control volume is only due to what comes into the artery portion minus what comes out of this portion (assuming impermeable walls too). If this control volume has length \(l\), we can write it as
\begin{equation}
V(t)=\int_0^lA(x,t)\dd x,
\end{equation}

and hence, due to the reasoning made before, we can write the rate of change of mass (or volume, actually) as
\begin{equation}
\rho\frac{\partial V(t)}{\partial t}=\rho(Q(0,t)-Q(l,t)).
\end{equation}

We have corrected the equation with the density of the blood \(\rho\) for completeness, although in our case it will have no effect. If \(\rho\) depended on time, for example, this could not be simplified (we would be dealing with a compressible fluid). Now, we can rewrite this expression using the definition of \(V(t)\) and the fundamental theorem of Calculus as
\begin{equation}
\rho\frac{\partial}{\partial t}\int_0^lA(x,t)\dd x+\rho\int_0^l\frac{\partial Q(x,t)}{\partial x}\dd x = 0.
\end{equation}

We have no issues in doing this since all the functions we are considering (amplitude \(A\), velocity \(u\) and flux \(Q\)) are smooth enough due to its physical meaning to integrate them and take derivatives. If we assume \(l\) does not depend on time, we can take the derivative with respect to \(t\), inside the integral to arrive to
\begin{equation}
\int_0^l\frac{\partial A}{\partial t}+\frac{\partial Q}{\partial x}\dd x = 0.
\end{equation}

We have no problems in deriving inside the integral symbol due to the smoothness of the functions involved. Since the control volume is arbitrary, the integrand of the above equation must be zero. We therefor obtain the differential one-dimensional mass conservation equation
\begin{equation}\label{eq:mass-conservation}
    \frac{\partial A}{\partial t}+\frac{\partial Q}{\partial x}\equiv \frac{\partial A}{\partial t}+\frac{\partial (uA)}{\partial x}=0.
\end{equation}

\subsection{Momentum equation}

The second equation comes from the concept of momentum of Newtonian dynamics. Analogously as we have done with the mass-conservation equation (or continuity equation), the momentum equation states that the rate of change of momentum within the control volume plus the net flux of momentum out the control volume is equal to the applied forces on the control volume. Again, we have to weight with the blood density since the flux is involved. This is,
\begin{equation}
\frac{\partial}{\partial t}\int_0^l\rho Q\dd x+\rho\left(Q(l,t)u(l,t)-Q(0,t)u(0,t)\right) = F,
\end{equation}

where \(F\) is the set of forces acting in the control volume. In this term we have to take into account the friction with the walls, the pressure of the flow against the walls and the force at the inlet minus the force at the outlet. Since the pressure is force per unit area, we can write this term as
\begin{equation}
F(t)=\underbrace{p(0,t)A(0,t)}_{\textrm{Force at the inlet}}-\underbrace{p(l,t)A(l,t)}_{\textrm{Force at the outlet}}+\underbrace{\int_0^l\int_{\partial S}\hat{p}n_x\dd s\dd x}_{\textrm{Pressure against the walls}}+\underbrace{\int_0^l f\dd x}_{\textrm{Friction}},
\end{equation}

where \(\partial S\) is the boundary of section \(S\), \(n_x\) is the \(x\)-component of the surface normal and \(f\) represents the friction force per unit length. The side wall pressure force given by the double integral can be simplified if we assume constant sectional pressure and we treat the tube as axisymmetric:
\begin{equation}
\int_0^l\int_{\partial S}\hat{p}n_x\dd s\dd x=\int_0^lp\frac{\partial A}{\partial x}\dd x.
\end{equation}

Putting everything altogether we arrive to the control-volume statement of momentum conservation
\begin{align}
\begin{split}
&\frac{\partial}{\partial t}\int_0^l\rho Q\dd x+\rho(Q(l,t)u(l,t)-Q(0,t)u(0,t))\\
=&p(0,t)A(0,t)-p(l,t)A(l,t)+\int_0^l\left(p\frac{\partial A}{\partial x}+f\right)\dd x.
\end{split}
\end{align}

As we have done before, using the fundamental theorem of Calculus and assuming \(l\) is independent of time and \(\rho\) is constant, we obtain
\begin{equation}
\int_0^l\left[\rho\left(\frac{\partial Q}{\partial t}+\frac{\partial(Qu)}{\partial x}\right)+\frac{\partial(pA)}{\partial x}-p\frac{\partial A}{\partial x}-f\right]\dd x=0.
\end{equation}

Since this equation is satisfied for an arbitrary control volume, the integrand must be zero, so the equation has the form
\begin{equation}\label{eq:premomentum-equation}
    \frac{\partial Q}{\partial t}+\frac{\partial (Qu)}{\partial x}=-\frac{A}{\rho}\frac{\partial p}{\partial x}+\frac{f}{\rho}.
\end{equation}

Finally, due to the appearance of \(u^2\) in~\eqref{eq:premomentum-equation} (\(Qu=Au^2\)), it is convenient to introduce the \emph{Coriolis coefficient} \(\alpha\) as a correction factor for the non-linearity of the momentum, so as to satisfy
\begin{equation}
    \alpha(x,t)=\frac{\int_S\hat{u}\dd\sigma}{Au^2},
\end{equation}

where \(\hat{u}\) stands for velocity within a constant \(x\)-section. Although in some papers such as~\cite{sherwin03} this coefficient is simplified to \(\alpha=1\), which means a flat profile, some others as~\cite{toro16} set them to \(\alpha=4/3\) for a parabolic profile. With this, the final conservation of momentum equation reads
\begin{equation}\label{eq:momentum-equation}
    \frac{\partial Q}{\partial t}+\frac{\partial (\alpha Qu)}{\partial x}=-\frac{A}{\rho}\frac{\partial p}{\partial x}+\frac{f}{\rho}.
\end{equation}

\subsection{The tube laws}\label{subsec:tube_laws}
We recall that we have three variables, \(A,u\) and \(p\) (or the combinations using the volume flux \(Q\)) and up to this moment only two equations: mass and momentum conservation. Consequently, to close the system given by equations~\eqref{eq:mass-conservation} and~\eqref{eq:momentum-equation} either we need one equation more or we have to remove one variable. This is commonly done defining a relationship, either differential or algebraic, between pressure and amplitude (known as the \emph{local tube law}). One could expect one well-established equation, as is the case of the conservation equations but, to the best of our knowledge there is no scientific agreement at this point. The reason of this, is because they are generally simplifications of the physical reality, so depends on the author which assumptions to make. Other ones leave free parameters in order to fit them with experimental data and other group of authors uses expressions purely mathematical. In this part we will make a brief review of some of the most important ones. 

\medskip

Before exposing them one clarification must be done. Depending on the reference, we will be using an initial pressure \(p_0(x)=p(x,0)\), a reference pressure \(p_{\textrm{ref}}(x,t)\), an external pressure \(p_{\textrm{ext}}(x,t)\) or none of them. These usually appear as a difference with the current pressure so the idea is that we are in a non-equilibrium situation. Some authors assume that at the beginning of the simulation we are in equilibrium so they use \(p_0(x)\); other models embrace other pressures such as atmospheric pressure (to distinguish if the artery is in vertical position or not) or they treat it as a different term (\(p_{\textrm{ext}}(x,t)\)) and most of them assume that the initial/reference pressure is zero in order to simplify. This is not a great issue since the reasoning is the same and for further simulations we will simplify these terms.

\begin{enumerate}

\item[M1)] Historically, one of the first approaches has been to assume that the cross-sectional area is a linear function of pressure and also that changes in area are relatively small. That is,
\begin{equation}
A(x,t) = A_0(x)+(p(x,t)-p_0(x))C(x,t)
\end{equation}

with
\begin{equation}
|(p(x,t)-p_0(x))C(x,t)|\ll A_0(x)
\end{equation}

where  \(C(x,t)=\left.\frac{\partial A(x,t)}{\partial p(x,t)}\right|_{A_0(x)}\) is the vessel compliance per unit length, \\\(p_0(x)=p(x,0)\) and \(A_0(x) = A(x,0)\). These equations, together with experimental values of \(C(x,t)\) (assumed constant) can be found in~\cite{raines74}.

\item[M2)] We can think of the previous expression as a Taylor expansion neglecting terms higher than first order. The natural question is if it is possible to get a higher order and, indeed, this has been studied back in \(1986\) in~\cite{gporenta86}. They presented the equation
\begin{equation}
A(x,t) = A_0(x)\left[1+C_0(p(x,t)-p_0(x))+C_1(p(x,t)-p_0(x))^2\right],
\end{equation}

which is particularly convenient for numerical manipulation. 

\item[M3)] In~\cite{sherwin03} and~\cite{sherwin03b}, Sherwin \textit{et al.} assumed a thin wall tube where each section is independent of the others. This model is based on linear elasticity where, using Hooke's law (first formulated in~\cite{hooke74}) for continuous media we have that
\begin{equation}
\sigma=E\varepsilon
\end{equation}

being \(\sigma\) the stress, \(\varepsilon\) the strain and \(E\) the Young's modulus. We recall that, despite they are actually tensors, the stress can be defined as a physical quantity that expresses the internal forces that neighbouring particles of a continuous material exert on each other, while strain is the measure of the deformation of the material. Young's modulus characterise the stiffness of an elastic material.

Let us denote the radius of the artery by \(R(x,t)\) and \(R_0(x)=R(x,0)\). Here \(h_0(x)\) will be used to denote the vessel-wall thickness and sectional area at the equilibrium state \((p,u)=(p_{\textrm{ref}},0)\), where \(p_{\textrm{ref}}\) is the reference pressure. We assume a cross section of a vessel with a thin wall (\(h\ll R\)), an isotropic, homogeneous, incompressible arterial wall that it deforms axisymmetrically with each circular cross-section independently of the others. Making these assumptions, we can express the strain as
\begin{equation}\label{eq:strain}
    \varepsilon=\frac{R-R_0}{(1-\nu^2)R_0},
\end{equation}

where \(\nu(x)\) is the other elasticity parameter, Poisson's ratio. It is defined as the ratio of transverse contraction strain to longitudinal extension strain in the direction of stretching force, so along with Young's modulus one can univocally determine the properties of a (linear) elastic material. By Young-Laplace's law\footnote{We are nor explaining neither deducing this law since it would take us apart from the objective of this work. We refer the interested reader to~\cite{young05} and~\cite{laplace05}.}, assuming there is not external pressure we can relate the pressure with the stress as
\begin{equation}\label{eq:stress}
    p=\frac{h_0\sigma}{\pi R}.
\end{equation}

Combining the previous equations we arrive to the tube law
\begin{equation}\label{eq:pressure}
    p(x,t)=p_{\textrm{ext}}+\beta(x)\left(\sqrt{A(x,t)}-\sqrt{A_0(x)}\right),
\end{equation}

where
\begin{equation}
\beta(x) = \frac{\sqrt{\pi}h_0(x)E(x)}{\left(1-\nu(x)^2\right)A_0(x)}
\end{equation}

is the parameter embracing the material properties and \(p_{\textrm{ext}}\) is the external pressure. There are no problems with the denominator since in most of the materials \(\nu(x)\in[0,0.5]\) and, although in capillaries the amplitude \(A_0\) is very small, the wall thickness \(h_0\) is too. Nevertheless, if this is the case, we must pay attention because of computational issues.

\item[M4)]Another approach, similar to the previous model, is to assume a linear pressure-area constitutive relation, as in~\cite{sochi13}, and hence the pressure is proportional to arterial amplitude difference, that is
\begin{equation}
p(x,t)=\gamma(x,t)(A(x,t)-A_0(x))
\end{equation}

where \(\gamma\) is a proportionality coefficient. To the best of our knowledge this coefficient has no explicit physical meaning.

\item[M5)]All of these models can be found in a more general way in a recent review (see~\cite{toro16}). Here, Toro wrote the tube law as
\begin{equation}
p(x,t)=\psi(A;K,A_0),
\end{equation}

where
\begin{equation}
\psi(A;K,A_0)=K(x)\phi(A,A_0)
\end{equation}

and
\begin{equation}
\phi(A,A_0)=\left[\left(\frac{A(x,t)}{A_0(x)}\right)^m-\left(\frac{A(x,t)}{A_0(x)}\right)^n\right],
\end{equation}

where \(m,n\) and \(K(x)\) are free parameters and a function. We can recover equation~\eqref{eq:pressure} by setting \(m=1/2\), \(n=0\) and
\begin{equation}
K(x)=\sqrt{\frac{\pi}{A_0(x)}}\left(\frac{h_0(x)E(x)}{1-\nu^2(x)}\right).
\end{equation}

In this case, we would have that \(\beta(x)=K(x)/\sqrt{A_0(x)}\). If it is not specified, we can always replace \(p(x,t)\) by \(p(x,t)-p_{\textrm{ref}}(x,t)\) because in most of the models the reference pressure is set to \(0\) for simplicity.

\bigskip

Up to this point, the tube laws presented have been based on a purely elastic behaviour of the artery wall. Another parallel line of work has been to consider the artery wall  as a viscoelastic material. For the arterial wall (or viscoelastic solids in general), when a fixed stress is loaded, the wall keeps extending gradually (creeping) after an instantaneous extension. The main issue for the simulation of viscoelastic materials is the significant increase of theoretical and computational complexity (both in running time as in code development) but, for completeness, we present some of the most used models. We will do it in a more enumerating way, since the deduction of all the equations is out of the scope of this work.

\item[M6)]We can find a complex formulation of viscoelasticity for blood vessels back in \(1970\) in~\cite{westerhof70}. The following tube law of the generalised viscoelastic model can be derived using knowledge of solid mechanics:
\begin{equation}
p(x,t) = \frac{1}{C(x,t)}\left[A(x,t)-A_0(x)+\int_0^t\sum_{i=1}^nf_ie^{-(t-u)/\tau_i}\frac{\partial A(x,u)}{\partial t}\dd u\right]
\end{equation}

where \(C(x,t)=\left.\frac{\partial A(x,t)}{\partial p(x,t)}\right|_{A_0(x)}\). The viscoelastic property of the tube wall is reflected in the dynamic viscoelasticity parameters \(f_i\) and relaxation time parameter \(\tau_i\). It is possible to determine the term number \(n\) from the viscoelastic characteristics of the material.

\item[M7)]A widely used viscoelastic model is Kelvin-Voigt model which writes
\begin{equation}
\sigma(t)=E\varepsilon(t)+\phi\frac{\partial\varepsilon(t)}{\partial t},
\end{equation}

being \(\phi\) a coefficient for the viscosity of the material. We note that if \(\phi=0\) we have a linear relationship and hence we recover the elastic model. If we suddenly apply some constant stress \(\sigma _{0}\) to a Kelvin–Voigt material, then the deformations would approach the deformation for the pure elastic material \(\sigma_0/E\) with the difference decaying exponentially:
\begin{equation}
\varepsilon(t)=\frac{\sigma_0}{E}\left(1-e^{-\lambda t}\right),
\end{equation}

where \(\lambda\) can be interpreted as the rate of relaxation 
\begin{equation}
\lambda=\frac{E}{\phi}.
\end{equation}

Hence, it is the description of a elastic material but with some delay. This is the reason why it is the description of a viscoelastic material.

Using the same reasoning as in the third model we get the tube law
\begin{equation}
p=\beta\left(\sqrt{A}-\sqrt{A_0}\right)+\nu_s\frac{\partial A}{\partial t},
\end{equation}

with stiffness coefficient
\begin{equation}
\beta=\frac{\sqrt{\pi}Eh}{(1-\nu^2)A_0}
\end{equation}

and viscosity coefficient
\begin{equation}
\nu_s=\frac{\sqrt{\pi}\phi h}{2(1-\nu^2)\sqrt{A_0}A}.
\end{equation}

This model was first used in by Čanić \textit{et al.} in~\cite{canic06} (although in a more general way).

\item[M8)]Finally, one of the most novel approaches has been to use the concept of fractional derivative. This branch of calculus consider not only the first, second... derivatives, but ``intermediate'' ones. One of the most used definitions (specially oriented to computational purposes) starts with the basic definition of derivative:
\begin{align}
\begin{split}
    f'(x)&=\lim_{h\to0}\frac{f(x+h)-f(x)}{h},\\
    f''(x)&=\lim_{h\to0}\frac{f'(x+h)-f'(x)}{h}=\lim_{h\to0}\frac{f(x+2h)-2f(x+h)+f(x)}{h^2},\\
    &\vdots\\
    f^{(n)}(x)&=\lim_{h\to0}\frac{\sum_{0\leq m\leq n}(-1)^m\binom{n}{m}f(x+(m-n)h)}{h^n}.
\end{split}
\end{align}

In order to achieve real (even complex) values of the derivative's degree we replace the factorial with the Euler's gamma function (which returns the factorial for integer values). With some algebra we arrive to:
\begin{align}
\begin{split}
    \DD^\alpha f(x)&=\lim_{h\to0}h^{-\alpha}\sum_{m=0}^{\frac{x-a}{h}}\frac{(-1)^m\Gamma(\alpha+1)}{\Gamma(m+1)\Gamma(\alpha-m+1)}f(x-mh)\\
    &=\lim_{n\to\infty}\left(\frac{n}{x-a}\right)^\alpha\sum_{m=0}^n\frac{(-1)^m\Gamma(\alpha
    +1)}{\Gamma(m+1)\Gamma(\alpha-m+1)}f\left(x-m\left(\frac{x-a}{n}\right)\right),
\end{split}
\end{align}

where now \(\alpha\in\CC\) and \(a<x\) is the point from where the derivative corresponding to \(x\) will be calculated. This definition is called the Grünwald-Letnikov formula, first used in~\cite{letnikov68}. It is not the unique definition, since there are more than thirty (see~\cite{capelasdeoliveira14}) and some of them are not even equivalent. Nevertheless, what all of them have in common is that they have some ``memory'', in the sense that previous values of the derivative's degree, affect the current value. Using this idea, Perdikaris \textit{et al.} modelled arterial viscoelasticity with fractional calculus and they run its interaction with blood flow in~\cite{perdikaris14}. For this, they use the Grünwald-Letnikov formula in a recursive way:
\begin{gather}\label{eq:grunwald-letnikov}
    D_t^\alpha f(t)=\lim_{\Delta t\to0}\Delta t^{-\alpha}\sum_{k=0}^\infty GL_k^\alpha f(t-k\Delta t),\\
GL_k^\alpha=\frac{k-\alpha-1}{k}GL_{k-1}^\alpha
\end{gather}

with \(GL_0^\alpha=1\) and \(t-k\Delta t\) must be in the domain of \(f\). The fractional stress-strain relation reads as
\begin{equation}\label{eq:visco_frac}
\sigma(t)+\tau_\sigma^\alpha D_t^\alpha\sigma(t)=E\left[\varepsilon(t)+\tau_\varepsilon^\alpha D_t^\alpha\varepsilon(t)\right].
\end{equation}

And, as before
\begin{equation}\label{eq:stress-strain}
\sigma=\frac{R(p-p_{\textrm{ext}})}{h_0}\qquad\textrm{and}\qquad\varepsilon=\frac{R-R_0}{(1-\nu^2)R_0}.
\end{equation}

So, replacing the equations of~\eqref{eq:stress-strain} and the expression of fractional derivative of~\eqref{eq:grunwald-letnikov} in the stress-strain relationship~\eqref{eq:visco_frac} we arrive to
\begin{gather}
\begin{split}
p(x,t)=\frac{1+\tau_\varepsilon^\alpha\Delta t^{-\alpha}}{1+\tau_\sigma^\alpha\Delta t^{-\alpha}}p^E(x,t)+\frac{\Delta t^{-\alpha}}{1+\tau_\sigma^\alpha\Delta t^{-\alpha}}\\
\times\sum_{k=0}^\infty GL_k^\alpha\left[\tau_\varepsilon^\alpha p^E(t-k\Delta t)-\tau_\sigma^\alpha p(t-k\Delta t)\right]
\end{split}
\end{gather}

where \(p^E\) correspond to the elastic pressure contribution
\begin{equation}
p^E(x,t)=\left(\frac{\tau_\varepsilon}{\tau_\sigma}\right)^\alpha\beta\left(\sqrt{A}-\sqrt{A_0}\right).
\end{equation}

The parameters \(\tau_\varepsilon,\tau_\sigma\) and the fractional exponent \(\alpha\) have to be properly chosen each case. We notice that for the case \(\tau_\sigma^\alpha=0\), \(\tau_\varepsilon^\alpha=1/E\) and \(\alpha=1\) we recover the Kelvin-Voigt model. 
\end{enumerate}

We have provided a briefly overview of several ways of modelling the elastic or viscoelastic behaviour of the arterial wall, since the seventies up to the most recent models. We can visually check the appearing of these ideas in figure~\ref{fig:chronology}.
\begin{figure}[htb]
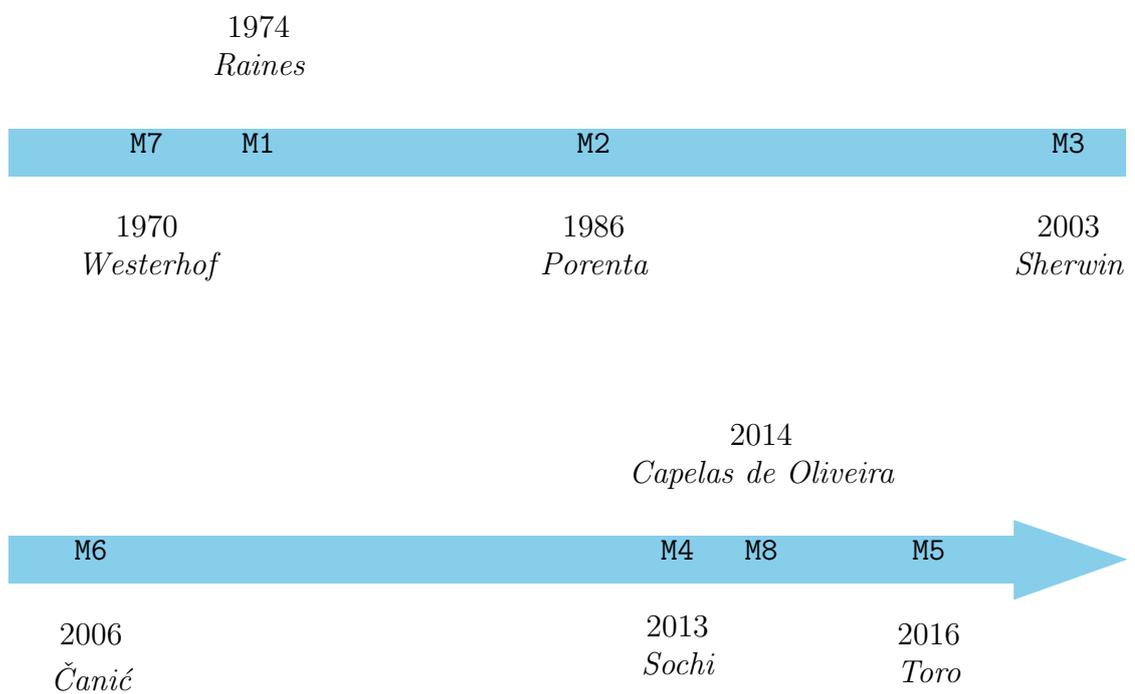

\startchronology[startyear=1965,stopyear=2005,color=SkyBlue,width=\hsize,height=1.5pc,dateselevation=45pt,dates=false,arrow=false]
\chronoevent[icon=\texttt{M7},barre=false,textstyle=\it,mark=false]{1970}{Westerhof}
\chronoevent[icon=\texttt{M1},barre=false,textstyle=\it,ifcolorbox=false,mark=false,markdepth=-40pt]{1974}{Raines}
\chronoevent[icon=\texttt{M2},barre=false,textstyle=\it,mark=false]{1986}{Porenta}
\chronoevent[icon=\texttt{M3},barre=false,textstyle=\it,mark=false]{2003}{Sherwin}
\stopchronology
\startchronology[startyear=2005,stopyear=2017,color=SkyBlue,width=\hsize,height=1.5pc,dateselevation=40pt,dates=false,arrowheight=2.5pc]
\chronoevent[icon=\texttt{M5},barre=false,textstyle=\it,mark=false]{2016}{Toro}
\chronoevent[icon=\texttt{M6},barre=false,textstyle=\it,mark=false]{2006}{Čanić}
\chronoevent[icon=\texttt{M8},barre=false,textstyle=\it,mark=false,markdepth=-40pt]{2014}{Capelas de Oliveira}
\chronoevent[icon=\texttt{M4},barre=false,textstyle=\it,ifcolorbox=false,mark=false]{2013}{Sochi}
\stopchronology
\caption{Chronology of the models presented on the tube law subsection.}
\label{fig:chronology}
\end{figure}

With this, we have shown in this chapter the statement of our problem with some of the most important options for the tube law. Although we will remark it in the following sections, due to its usefulness and because it is not purely phenomenological, we will use the third model, \textit{i.e.}, the purely elastic wall model where the proportionality constant has a clear physical sense. With this, in the next chapter we will use the aforementioned Navier-Stokes problem and with theory of partial differential equations (specifically, hyperbolic problems) we will prove some statements that will be useful for practical cases.

%% file: Theory/Characteristics.tex
Since the rest of the chapter is based on the method of characteristics, we will use this section to explain the main ideas of this method and apply it to our problem.

The method of Riemann characteristics has been used for more than a century to describe linear and nonlinear waves propagating in a medium~\cite{riemann60}. The main useful properties of this method, as are discussed in~\cite{paynter88}, are that it may be applied to both linear and nonlinear systems; that it may be applied equally well to solids, liquids, and gases; and that it provides a clear picture of the causal relations governing system behaviour. The method of characteristics was first developed by Riemann in an article published in \(1860\). Riemann limits himself to consideration of gases and begins with a discussion of the ideal gas law. He then progresses to a development of the method of characteristics and a discussion of applications of the method.

Riemann's work is built upon previous work on finite amplitude waves in air. An important contribution was made by Poisson~\cite{poisson08} whose \(1808\) article showed that the wave velocity is the sum of the sound speed and the mean flow velocity. Properties of waves of finite amplitude were discussed in \(1860\) by Earnshaw~\cite{earnshaw60}, but only for progressive waves.

Very roughly, a characteristic is a propagation path: a path followed by some entity, such a geometrical form or a physical disturbance when this entity is propagated. Thus, a ``gridiron'' of roads could be considered as propagation paths, as is shown in figure~\ref{fig:char_rectas}. With this very intuitive notion of characteristic (extracted from~\cite{abbott66}), the two families of lines correspond to two families of characteristics, usually named \emph{forward} characteristic and \emph{backward} characteristic.

\begin{figure}
    \centering
    \begin{subfigure}[t]{0.3\textwidth}
         \includegraphics[trim={3cm 18cm 7cm 2cm},clip,width=\linewidth]{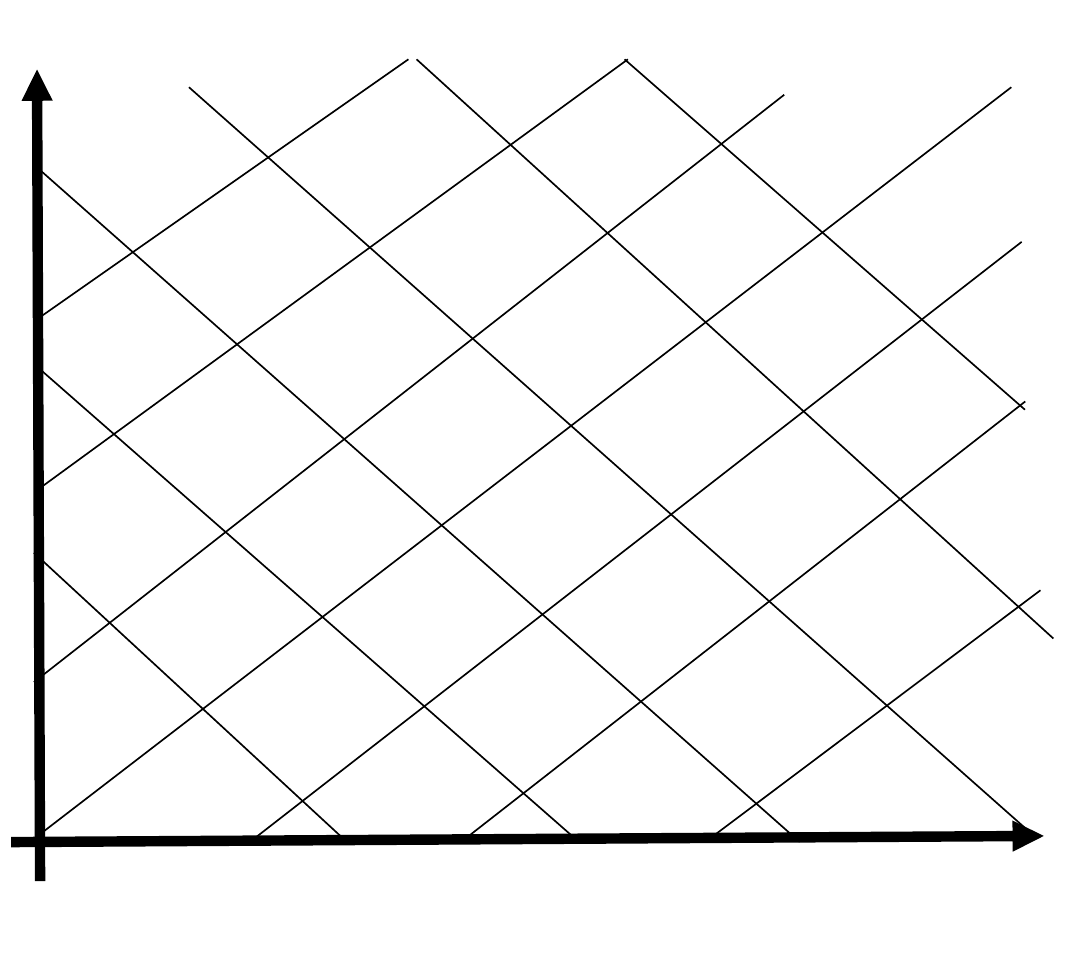}
        \caption{Characteristic curves with constant slope.}
        \label{fig:char_rectas}
    \end{subfigure}
    ~
    \begin{subfigure}[t]{0.3\textwidth}
         \includegraphics[trim={3cm 18cm 7cm 2cm},clip,width=\linewidth]{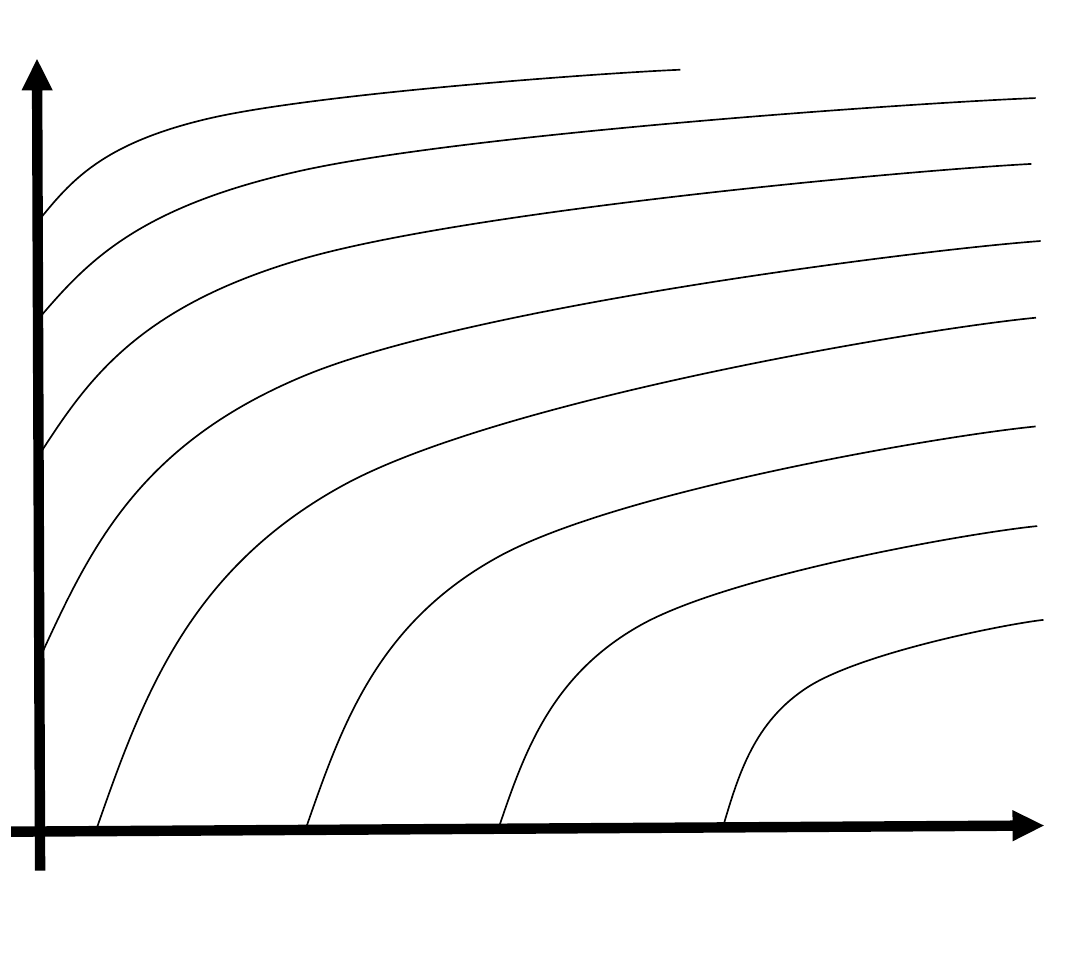}
        \caption{(Forward) characteristic curves with variable slope.}
        \label{fig:char_curvas}
    \end{subfigure}
    ~
    \begin{subfigure}[t]{0.3\textwidth}
        \includegraphics[trim={3cm 18cm 7cm 2cm},clip,width=\linewidth]{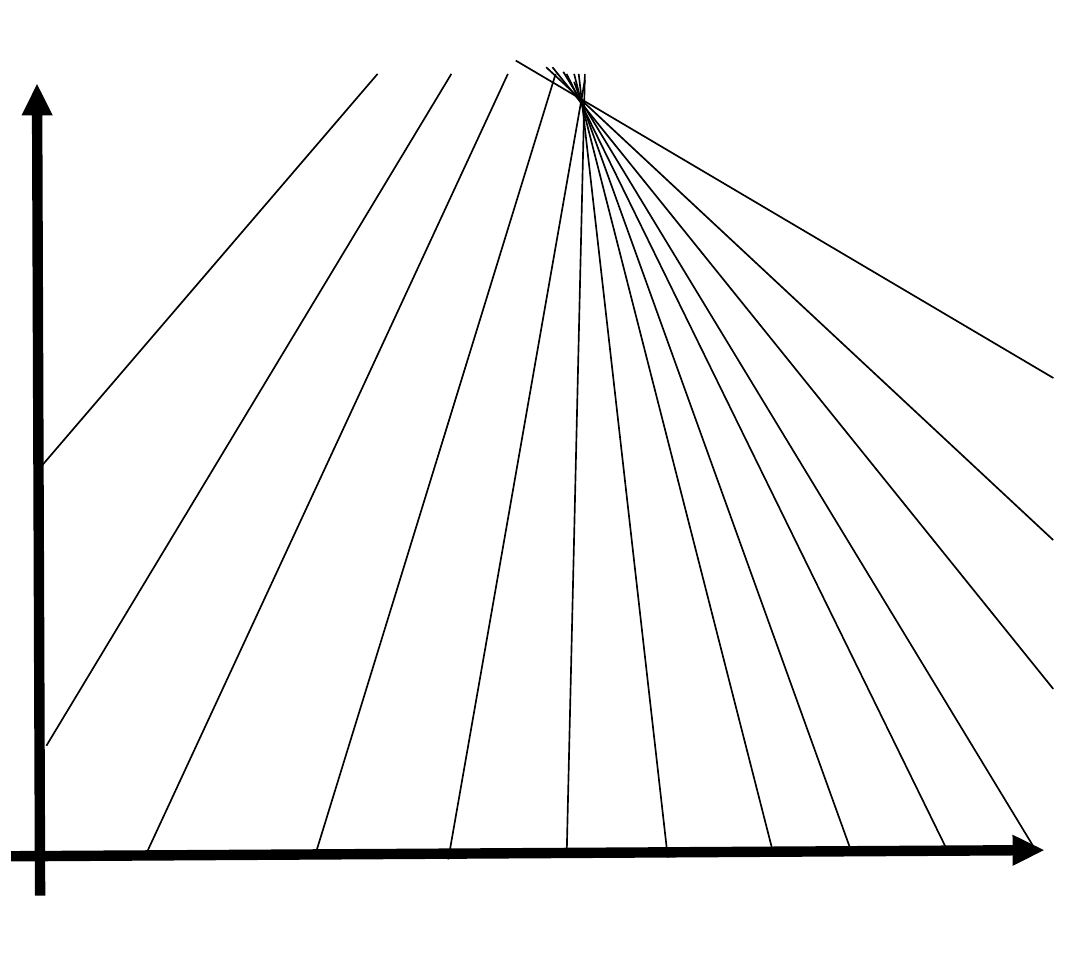}
        \caption{Formation of a shock wave.}
        \label{fig:char_shock}
    \end{subfigure}
    \caption{Scheme of different types of characteristic curves.}\label{fig:charac}
\end{figure}

With this method we want to transform a partial differential equation (or a system of equations) into an ordinary differential equation (or a system of equations). The usual way of achieving this is making linear combinations of certain terms choosing properly the multipliers. In the following subsection we will explain the general procedure for the case of a quasi-linear system with two variables. The problem with this derivation of the method is that it is based on calculus of variations where infinitesimal quantities are treated as manipulable entities. This approach is very common in engineering and physics but, regarding its rigour, it is not well appreciated in some mathematical fields. Nevertheless, for completeness and for historical reasons we will show this procedure.

\subsection{Derivation of the method via calculus of variations}
In this subsection we will follow the notation and discussion of the \(15\)\textsuperscript{th} chapter of~\cite{ralston65} written by Lister.

The general form of a quasi-linear system of equations for the case of two independent variables \(x,y\) and two dependent variables \(u,v\) can be written as a system of two equations, \(L_1\) and \(L_2\):
\begin{align}
    L_1:A_1u_x+B_1u_y+C_1v_x+D_1v_y+E_1=0\label{eq:L1}\\
    L_2:A_2u_x+B_2u_y+C_2v_x+D_2v_y+E_2=0\label{eq:L2}
\end{align}

where \(A_1,A_2,\dots,E_2\) are known functions of \(x,y,u,v\).

In the following considerations, it is assumed that all the functions introduced above are continuous and possess as many continuous derivatives as may be required. Consider a linear combination of \(L_1\) and \(L_2\):
\begin{align}
\begin{split}
L=\lambda_1L_1+\lambda_2L_2=&(\lambda_1A_1+\lambda_2A_2)u_x+(\lambda_1B_1+\lambda_2B_2)u_y+(\lambda_1C_1+\lambda_2C_2)v_x\\
&+(\lambda_1D_1+\lambda_2D_2)v_y+(\lambda_1E_1+\lambda_2E_2).
\end{split}
\end{align}

Now, if \(u=u(x,y)\) and \(v=v(x,y)\) are solutions to~\eqref{eq:L1} and~\eqref{eq:L2} then
\begin{equation}
\dd u=\frac{\partial u}{\partial x}\dd x+\frac{\partial u}{\partial y}\dd y,\qquad\dd v=\frac{\partial v}{\partial x}\dd x+\frac{\partial v}{\partial y}\dd y.
\end{equation}

The differential expression \(L\) can be written\footnote{Supposing that the differentials are ``virtual'' entities, following the standards of this field.} in the form
\begin{equation}\label{eq:dxL}
\dd xL=(\lambda_1A_1+\lambda_2A_2)\dd u+(\lambda_1C_1+\lambda_2C_2)\dd v+(\lambda_1E_1+\lambda_2E_2)\dd x
\end{equation}

if the constants \(\lambda_1\) and \(\lambda_2\) are chosen so that
\begin{equation}\label{eq:dx/dy}
\frac{\dd x}{\dd y}=\frac{\lambda_1A_1+\lambda_2A_2}{\lambda_1B_1+\lambda_2B_2}=\frac{\lambda_1C_1+\lambda_2C_2}{\lambda_1D_1+\lambda_2C_2}.
\end{equation}

In this case, in the differential expression \(L\), the derivatives of \(u\) and those of \(v\) are combined so that their derivatives are in the same direction, namely, \(\dd y/\dd x\). This direction is called a \emph{characteristic direction}.

From equation~\eqref{eq:dx/dy}, the ratio \(\lambda_1/\lambda_2\) can be obtained:
\begin{equation}\label{eq:quot_lamb}
-\frac{\lambda_1}{\lambda_2}=\frac{A_2\dd y-B_2\dd x}{A_1\dd y-B_1\dd x}=\frac{C_2\dd y-D_2\dd x}{C_1\dd y-D_1\dd x}
\end{equation}

hence
\begin{equation}\label{eq:hyper_pol}
a(\dd y)^2-2b\ \dd x\dd y+c(\dd x)^2=0.
\end{equation}

Here,
\begin{align}
    a&=A_1C_2-A_2C_1,\\
    2b&=A_1D_2-A_2D_1,\\
    c&=B_1D_2-B_2D_1.
\end{align}

For the case of hyperbolic partial differential equations, two distinct roots of the quadratic equation~\eqref{eq:hyper_pol} exist. Therefor
\begin{equation}
b^2-ac>0.
\end{equation}

This excludes the exceptional case of all three coefficients vanishing. Moreover, it is assumed for convenience that 
\begin{equation}
a\neq0.
\end{equation}

The latter condition can always be satisfied, if necessary, by introducing new coordinates instead of \(x\) and \(y\). Consequently, \(\dd x\neq0\) for a characteristic direction \((\dd x,\dd y)\) as seen from~\eqref{eq:hyper_pol}; thus the slope
\begin{equation}
\zeta=\frac{\dd y}{\dd x}
\end{equation}

can be introduced, and \(\zeta\) satisfies the equation:
\begin{equation}
a\zeta^2-2b\zeta+c=0.
\end{equation}

This equation has two different real solutions \(\zeta_+\) and \(\zeta_-\),
\begin{equation}
\zeta_+\neq\zeta_-.
\end{equation}

Thus, at the point \((x,y)\), the two different characteristic directions are given by:
\begin{equation}\label{eq:char_roots1}
\frac{\dd y}{\dd x}=\zeta_+,\quad\frac{\dd y}{\dd x}=\zeta_-.
\end{equation}

Since \(a,b\) and \(c\) are in general functions of \(u,v,x\) and \(y\), \(\zeta_+\) and \(\zeta_-\) will also be functions of these quantities:
\begin{equation}\label{eq:char_roots}
\frac{\dd y}{\dd x}=\zeta_+(x,y,u,v),\quad\frac{\dd y}{\dd x}=\zeta_-(x,y,u,v).
\end{equation}

Once a solution \(u(x,y),v(x,y)\) of~\eqref{eq:L1} and~\eqref{eq:L2} has been obtained, equations\\\eqref{eq:char_roots} become two separate ordinary differential equations of the first order. These ODEs define two one-parameter families of \emph{characteristic curves} (often abbreviated to characteristics), in the \((x,y)\) plane, belonging to this solution \(u(x,y),v(x,y)\). These two families form a curvilinear coordinate net (as could be figure~\ref{fig:char_curvas}).

If \(\zeta_+\) and \(\zeta_-\) are functions of \(x\) and \(y\) only, then
\begin{equation}
\frac{\dd y}{\dd x}=\zeta_+(x,y),\quad\frac{\dd y}{\dd x}=\zeta_-(x,y),
\end{equation}

and it is not necessary to find a solution to~\eqref{eq:L1} and~\eqref{eq:L2} in order to find the equations of the characteristics; hence the problem is simplified.

Substituting the solutions~\eqref{eq:char_roots1} into the expressions for \(\lambda_1/\lambda_2\) given in~\eqref{eq:quot_lamb} yields
\begin{equation}\label{eq:quot_lamb2}
\frac{\lambda_1}{\lambda_2}=-\frac{A_2\zeta_+-B_2}{A_1\zeta_+-B_1},\quad\frac{\lambda_1}{\lambda_2}=-\frac{A_2\zeta_--B_2}{A_1\zeta_--B_1}.
\end{equation}

Finally, combining~\eqref{eq:quot_lamb2} and~\eqref{eq:dxL} gives
\begin{align}
    F\dd u+(a\zeta_+-G)\dd v+(K\zeta_+-H)\dd x&=0,\\
    F\dd u+(a\zeta_--G)\dd v+(K\zeta_--H)\dd x&=0,
\end{align}

where
\begin{align}
F&=A_1B_2-A_2B_1,&G&=B_1C_2-B_2C_1, \nonumber\\
&&&\\
K&=A_1E_2-A_2E_1,&H&=B_1E_2-B_2E_1.\nonumber
\end{align}

Thus, the following four characteristic equations have been obtained:
\begin{align}
    \dd y-\zeta_+\dd x&=0\label{eq:char1}\\
    F\dd u+(a\zeta_+-G)\dd v+(K\zeta_+-H)\dd x&=0\\
    \dd y-\zeta_-\dd x&=0\\
    F\dd u+(a\zeta_--G)\dd v+(K\zeta_--H)\dd x&=0\label{eq:char2}
\end{align}

Equations~\eqref{eq:char1}--\eqref{eq:char2} are of a particularly simple form, inasmuch as each equation contains only total derivatives of all the variables.

According to the derivation, every solution of the original system~\eqref{eq:L1} and~\eqref{eq:L2} satisfies the system~\eqref{eq:char1}--\eqref{eq:char2}. Courant and Friedrichs showed that the converse is also true in~\cite{courant48}.

At this point we could apply the previous reasoning to our case (previously checking the hyperbolicity) but we will avoid this path. The reasoning followed is more typical in physics' and engineering areas, but we can provide more rigour to this method. The method of characteristics can also be achieved by means of linear algebra and standard calculus: we will dedicate the following subsections to the derivation of the method in this way. This will also provide the notation and the tools for a posterior analysis of the problem.

\subsection{Setting of the Navier-Stokes system}\label{ssec:charac_NS}

First of all, let us manipulate the expression of the momentum equation~\eqref{eq:momentum-equation} for convenience (the left-hand side):
\begin{align}
\begin{split}
\frac{\partial Q}{\partial t}+\frac{\partial\alpha Qu}{\partial x}&=\frac{uA}{\partial t}+\frac{\partial\overbrace{\alpha u^2A}^{(\alpha u A)u}}{\partial x}=A\frac{\partial u}{\partial t}+u\frac{\partial A}{\partial t}+u\frac{\partial\alpha u A}{\partial x}+\alpha u A\frac{\partial u}{\partial x}\\
&=A\frac{\partial u}{\partial t}+u\frac{\partial A}{\partial t}+u\frac{\partial\alpha u A}{\partial x}+\alpha u A\frac{\partial u}{\partial x} + u\frac{uA}{\partial x}-u\frac{u A}{\partial x}\\
&=u\frac{\partial A}{\partial t}+u\frac{\partial uA}{\partial x}+u\frac{\partial\alpha uA}{\partial x}-u\frac{\partial\alpha uA}{\partial x}+\alpha u A\frac{\partial u}{\partial x} \\
&=u\left\{\frac{\partial A}{\partial t}+\frac{\partial uA}{\partial x}\right\}+u\frac{\partial(\alpha-1)uA}{\partial x}+A\left\{\frac{\partial u}{\partial t}+\alpha u\frac{\partial u}{\partial x}\right\}.
\end{split}
\end{align}

The first term is the mass conservation equation~\eqref{eq:mass-conservation} and is therefor zero. Following the simplifications made by~\cite{sherwin03b} we now assume inviscid flow with a \emph{flat profile}, which implies that \(f=0\) and \(\alpha=0\).

This will not be a big deal, since it was shown in~\cite{canic03} that the source term is one order of magnitude smaller than the effects of non-linear advection. Since the inviscid flow does not generate boundary layer it is physically reasonable to assume the aforementioned flat velocity profile.

Even though this might seem a crude assumption\footnote{The analytical solution of pulsatile flow in a straight cylindrical elastic tube is given in reference~\cite{womersley55}. Analytical solutions for an initially stressed, anisotropic elastic tube are presented in reference~\cite{tsangaris89}.}, comparison with experimental data~\cite{hunter72} has shown that blood velocity profiles are rather flat on average. Furthermore, this assumption simplifies the analysis. However, we should stress that the methods to be presented in the following sections may be readily extended to the case \(\alpha\neq1\). Moreover, due to its wide use in the literature we will also assume the model M3) for the pressure of subsection~\ref{subsec:tube_laws}.

Using this and the continuity equation~\eqref{eq:mass-conservation} we can write the governing equations in terms of the \((A,u)\) variables as
\begin{align}
    \frac{\partial A}{\partial t}+\frac{\partial uA}{\partial x}&=0,\label{eq:mc_sher}\\
    \frac{\partial uA}{\partial t}+\frac{\partial u^2/2}{\partial x}&=-\frac{1}{\rho}\frac{\partial p}{\partial x}.\label{eq:mom_sher}
\end{align}

Equivalently, we can write the system in \emph{conservative} form as
\begin{equation}\label{eq:system1D}
    \UU_t+\FF_x=0,
\end{equation}

with
\begin{equation}
\UU=\begin{bmatrix}A\\u\end{bmatrix},\qquad \FF=\begin{bmatrix}uA\\ \frac{u^2}{2}+\frac{p}{\rho}\end{bmatrix},
\end{equation}

denoting the derivatives with subscripts. With this rearrangement of the system, we will obtain now the characteristic curves and variables.

\subsection{Characteristic variables}
The first step is to use the chain rule in the equation~\eqref{eq:pressure}. Thus, we have
\begin{equation}
    \frac{\partial p}{\partial x}=\frac{\partial p}{\partial A}\frac{\partial A}{\partial x}+\frac{\partial p}{\partial \beta}\frac{\partial \beta}{\partial x}+\frac{\partial p}{\partial A_0}\frac{\partial A_0}{\partial x}
    =\frac{\beta}{2\sqrt{A}}\frac{\partial A}{\partial x}+\frac{\partial p}{\partial \beta}\frac{\partial \beta}{\partial x}+\frac{\partial p}{\partial A_0}\frac{\partial A_0}{\partial x}.
\end{equation}

Hence, we can rewrite system~\eqref{eq:system1D} in its \emph{quasi-linear} form as
\begin{equation}\label{eq:eigen}
\textbf{U}_t+\textbf{H}\textbf{U}_x=\begin{bmatrix}A_t\\u_t\end{bmatrix}+\begin{bmatrix}u&A\\c^2/A&u\end{bmatrix}\begin{bmatrix}A_x\\u_x\end{bmatrix}=\begin{bmatrix}0\\g\end{bmatrix}=\textbf{G},
\end{equation}

with
\begin{align}
\begin{split}
c^2&=\frac{\beta\sqrt{A}}{2\rho},\\ 
g&=\frac{1}{\rho}\left(-\frac{\partial p}{\beta}\frac{\partial\beta}{\partial x}+\frac{\partial p}{\partial A_0}\frac{\partial A_0}{\partial x}\right).
\end{split}
\end{align}

To see what kind of problem we are dealing with, we look at the eigenvalues of \(\textbf{H}\). We can diagonalise \(\textbf{H}\) as
\begin{equation}\label{eq:left_eigensystem}
\textbf{H}=\textbf{P}^{-1}\textbf{D}\textbf{P}
\end{equation}

with
\begin{equation}\label{eq:expression_eigenvalues}
\textbf{P}=\begin{bmatrix}\frac{c}{A}&1\\-\frac{c}{A}&1\end{bmatrix},\qquad \textbf{D}=\begin{bmatrix}u+c&0\\0&u-c\end{bmatrix}.
\end{equation}

Now, \(\textbf{D}_{11}>0\) and, since \(\beta\) is usually considerably bigger than the typical blood velocities, we have that \(\textbf{D}_{22}<0\). For the typical values of the parameters found in the literature we refer to appendix~\ref{appendix:bio}. Hence, we have two real and distinct eigenvalues for the quasi-linear problem. This means that our problem is strictly hyperbolic. 

\medskip

We recall that we are looking for the characteristic variables, the Riemann invariants. This is equivalent to find some variables that satisfy some equation in conservative form such as equation~\eqref{eq:system1D}. We do this in the following way. With the decomposition~\eqref{eq:left_eigensystem} (left eigensystem) we can reformulate the system~\eqref{eq:eigen} as
\begin{equation}
\textbf{U}_t+\textbf{H}\textbf{U}_x=\textbf{G}\equiv\textbf{U}_t+\textbf{P}^{-1}\textbf{D}\textbf{P}\textbf{U}_x=\textbf{G}\equiv\textbf{P}\textbf{U}_t+\textbf{D}\textbf{P}\textbf{U}_x=\textbf{P}\textbf{G}.
\end{equation}

If \(\textbf{P}\) were a constant matrix with respect to the solution variables, we could rename \(\textbf{P}\textbf{U}=:\textbf{W}\) so we could get a conservative system. Following this idea, we are looking for \(\textbf{W}\) such that
\begin{align}
    \textbf{W}_t&=\textbf{P}\textbf{U}_t,\\
    \textbf{W}_x&=\textbf{P}\textbf{U}_x,
\end{align}

where we recall that the matrix \(\textbf{P}\) depends on \(A\). On the other hand, by the chain rule we have
\begin{equation}
\frac{\partial\textbf{W}(A,u)}{\partial t}=\textbf{W}_AA_t+\textbf{W}_uu_t=\left(\nabla\textbf{W}\right)\textbf{U}_t,
\end{equation}

being \(\nabla\textbf{W}\) the gradient of \(\textbf{W}\) (the procedure with \(\textbf{W}_x\) is the same). Therefor 
\begin{align}
\textbf{P}\textbf{U}_x&=\textbf{W}_x=(\nabla\textbf{W})\textbf{U}_x,\\
\textbf{P}\textbf{U}_t&=\textbf{W}_t=(\nabla\textbf{W})\textbf{U}_t
\end{align}

which implies that
\begin{equation}
\textbf{P}=\nabla\textbf{W}
\end{equation}

if and only if
\begin{align}
    \frac{\partial W_1}{\partial A}&=\frac{c}{A},&
    \frac{\partial W_1}{\partial u}&=1,\\
    \frac{\partial W_2}{\partial A}&=-\frac{c}{A},&
    \frac{\partial W_2}{\partial u}&=1.
\end{align}

Integrating we arrive to:
\begin{align}\label{eq:W12}
\begin{split}
   W_{1,2}&=\int_{u_{\textrm{ref}}}^u\dd u\pm\int_{A_{\textrm{ref}}}^A\frac{c}{A}\dd A=u-u_{\textrm{ref}}\pm\int_{A_{\textrm{ref}}}^A\frac{c}{A}\dd A\\
   &=u\pm4\sqrt{\frac{\beta}{2\rho}}A^{1/4},
\end{split}
\end{align}

where \((u_{\textrm{ref}},A_{\textrm{ref}})=(0,0)\) is taken as the reference state. This assumption has been made previously in~\cite{sherwin03}. The characteristic variables given by equation~\eqref{eq:W12} are also Riemann invariants of the system~\eqref{eq:mc_sher} and~\eqref{eq:mom_sher} in terms of the \((A,u)\) variables. We have achieved an explicit expression for the characteristic curves. This will allow us to use them and to check some interesting properties for practical uses. Among others, in section~\ref{sec:PDE} we will be able to use these expressions to check the existence of a global smooth solution.

Finally, since in a coherent solution \(\beta(x,t)\) is always positive, we can write the variables \((A,u)\) in terms of \((W_1,W_2)\) as
\begin{gather}
    A=\left(\frac{W_1-W_2}{4}\right)^4\left(\frac{\rho}{2\beta}\right)^2,\label{eq:recA}\\
    u=\frac{W_1+W_2}{2}.\label{eq:recu}
\end{gather}

The above result has been previously obtained in~\cite{sherwin03,sherwin03b} and with a more general system in terms of the variables \((A,Q)\) in~\cite{formaggia99,formaggia01}.

%% file: Theory/PDE_Stuff.tex
In this section we will prove some useful theoretical results of hyperbolic systems. Theorems~\ref{theo:collapse} and~\ref{theo:smooth} first appeared in~\cite{canic03} and theorem~\ref{theo:shock} is an adaptation of the one stated by~\cite{canic03} but using our own problem and assumptions. The proofs are mainly based on the study of the behaviour of the solution and its derivative along the characteristics (see~\cite{lax57,lax64}).

Let
\begin{equation}\label{eq:hyperb}
    \UU_t+\FF\left(\UU\right)_x=0,\qquad x\in\RR,\ t>0
\end{equation}

be a \(2\times2\) system of conservation laws where \(\UU\left(x,t\right)\in\RR^2\) and \(\FF\colon\RR^2\to\RR^2\) is a smooth function of \(\UU\). We shall assume that the system is strictly hyperbolic, that is, there exist two real and distinct eigenvalues \(\lambda_1>\lambda_2\). Suppose we have the above system in its characteristic form
\begin{align}\label{eq:characteristic_problem}
    \frac{\partial W_1}{\partial t}+\lambda_1\left(W_1,W_2\right)\frac{\partial W_1}{\partial x} &=0,\\
    \frac{\partial W_2}{\partial t}+\lambda_2\left(W_1,W_2\right)\frac{\partial W_2}{\partial x} &=0
\end{align}

where \(W_1,W_2\), the \emph{characteristic variables} or \emph{Riemann invariants}, are the unknown functions and \(\lambda_1,\lambda_2\) are smooth functions of \(W_1\) and \(W_2\). We note that due to the hyperbolicity of the system~\eqref{eq:hyperb} we can always do this at least locally. We will also assume that the system~\eqref{eq:hyperb} is non-linear in the considered domain, that is
\[\frac{\partial\lambda_1}{\partial W_1}\neq0\qquad\textrm{and}\qquad\frac{\partial\lambda_2}{\partial W_2}\neq0.\]
Consider the domain
\begin{equation}
    D=\{\left(x,t\right)\colon t\geq0, x_1\left(t\right)\leq x<+\infty\}
\end{equation}

with \(x_1\left(t\right)\in\RR\). Here we have the initial boundary--value problem
\begin{gather}\label{eq:characteristic_conditions}
    W_1\left(x,0\right)=W_1^0\left(x\right),\qquad W_2\left(x,0\right)=W_2^0\left(x\right)\quad\forall x\in[x_1(0),+\infty[\\
    W_2\left(x_1\left(t\right),t\right) = g\left(W_1\left(x,t\right),t\right)
\end{gather}

where we can assume without loss of generality that \(x_1\left(0\right)=0\).

The first important result is if we can assure that the tube does not collapse spontaneously. If this did not happen, the arterial amplitude could shrink over time until the vessel gets blocked. Indeed, with the notation and conclusions from the previous subsection, taking \(\lambda_1=\textbf{D}_{11}\) and \(\lambda_2=\textbf{D}_{22}\) we have:
\begin{theo}\label{theo:collapse}
Suppose that the left boundary \(x_1(t)=0\) is non-characteristic (\textit{i.e.} \(\lambda_1<x_1'<\lambda_2\)). Let \(x_2(t)\) be the forward characteristic emanating from the origin. If \(A(x,0)>0\) and if \(A(x_1(t),t)>0\) on the left boundary, then \(A(x,t)>0\), \(\forall(x,t)\in D_2^T\), and so system~\eqref{eq:mc_sher} and~\eqref{eq:mom_sher} is strictly hyperbolic in \(D_2^T\) where
\begin{equation}
    D_2^T=\left\{(x,t)\colon0\leq t\leq T,\; x_1(t)\leq x\leq x_2(t)\right\},
\end{equation}

for every \(T>0\).  
\end{theo}
\begin{proof}
Let \(x=x(t)\) be a solution curve of the ODE
\begin{equation}
        \frac{\dd x}{\dd t}=u(x,t).
\end{equation}

Mass-conservation equation~\eqref{eq:mc_sher} implies that along \(x=x(t)\) the cross-sectional area satisfies
\begin{equation}
    \frac{\partial A}{\partial t}+u\frac{\partial A}{\partial x}=-A\frac{\partial u}{\partial x}.
\end{equation}

Suppose that \((x^*,t^*)\in D_2^T\) is such that \(A(x^*,t^*)=0\) with \(t^*>0\), \(x^*>x_1\). This implies, due to equation~\eqref{eq:expression_eigenvalues} \(\lambda_1=\lambda_2=u(x^*,t^*)\). We also see that up to \((x^*,t^*)\) the integral curve \(x=x(t)\) passing through \((x^*,t^*)\) lies between the characteristic curves through \((x^*,t^*)\). Therefor, it either intersects the \(t=0\) axis or it intersects the left boundary. Suppose that \(x=x(t)\) intersects the initial line \(t=0\); denote that point by \((x(t^*,x^*),0)=(x_0,0)\). The solution of the ODE satisfied by \(A\) along \(x=x(t)\) is given by 
\begin{equation}
    A(x^*,t^*)=A(x_0,0)e^{-\int_0^{t^*}(\partial u/\partial x)\dd\tau},
\end{equation}

and we see that \(A(x^*,t^*)=0\) if and only if \(A(x_0,0)=0\) which contradicts the assumption that initially \(A(x,0)>0\) for every \(x\). The same reasoning applies to the case when \(x=x(t)\) intersects the left boundary. Hence, by contradiction the conclusion holds. 
\end{proof}

Now that we have assure this, the natural question would be if there is a unique, smooth solution. For this, we have the following theorem\footnote{For compactness we will use the apostrophe for denoting the derivative when there is no possibility of confussion.}:
\begin{theo}\label{theo:smooth}
With the previous notation, let us suppose  that the following hypotheses hold:
\begin{enumerate}
    \item[H1] Initial values \(W_1^0,W_2^0\) and the function \(g\) are \(\Cf^1\) and the boundary \(x_1\in\Cf^2\).
    \item[H2] The boundary \(x_1\) satisfies
    \begin{equation}
        \lambda_2\left(W_1,W_2\right)<x_1'\left(t\right)<\lambda_1\left(W_1,W_2\right)
    \end{equation}

    on \(x=x_1\left(t\right)\) and
    \begin{gather}
        \lambda_1\left(W_1,W_2\right)-x_1'\left(t\right)\geq M\left(\overline{t},\overline{W_1},\overline{W_2}\right),\\
        \forall0\leq t\leq \overline{t},\;\forall|W_1|\leq \overline{W_1},\;\forall|W_2|\leq \overline{W_2}
    \end{gather}

    where \(M\left(\overline{t},\overline{W_1},\overline{W_2}\right)>0\).
    \item[H3]  \(\|\left(W_1^0,W_2^0\right)\|_{\Cf^0}\) is bounded and \(\left(W_2^0\right)'\left(x\right)\leq0,\left(W_1^0\right)'\left(x\right)\geq0\) for \(0\leq x<+\infty\).
    \item[H4] The dependence of \(g\) on \(W_2\) is such that \(\partial g/\partial W_2\geq0\).
    \item[H5] The eigenvalues satisfy \(\partial\lambda_2/\partial W_2<0\), \(\partial\lambda_1/\partial W_1>0\).
    \item[H6] The following compatibility conditions holds:
    \begin{align}
        &W_1^0\left(0\right)=g\left(W_2^0\left(0\right),0\right),\\
        &x_1'\left(0\right)-\lambda_1\left(W_1^0\left(0\right),W_2^0\left(0\right)\right)\left(W_1^0\right)'\left(0\right)\\
        &=\frac{\partial g}{\partial t}\left(W_2^0\left(0\right),0\right)+\frac{\partial g}{\partial W_2}\left(W_2^0\left(0\right),0\right)\left(x_1'\left(0\right)-\lambda_2\left(W_1^0\left(0\right),W_2^0\left(0\right)\right)\right)\left(W_2^0\right)'\left(0\right).
    \end{align}

\end{enumerate}
If \(\partial g/\partial t\leq0\) the initial boundary problem~\eqref{eq:characteristic_problem},~\eqref{eq:characteristic_conditions} admits a unique global \(\Cf^1\) solution \(\left(W_1\left(x,t\right),W_2\left(x,t\right)\right)\) on the domain \(D\).
\end{theo}
\begin{proof}
Let \(x_2\left(t\right)\) be the forward characteristic emanating from the origin \((x,t)=(0,0)\). Using the results presented in the book of Li~\cite{li94} we can assure due to hypotheses H3, H4 and H5 that there exists a unique, global, \(\Cf^1\) solution in the domain 
\begin{equation}
    D_1=\{\left(x,t\right)\colon x\geq x_2\left(t\right),t\geq0\}.
\end{equation}

Furthermore, the behaviour of \(W_2\) along the characteristic \(x_2\left(t\right)\) passing through the origin is such that
\begin{equation}
    \frac{\partial W_2}{\partial t}\left(x_2\left(t\right),t\right)\leq0.
\end{equation}

In the rest of the proof we will try to extend the previous solution to the remaining domain 
\begin{equation}
    D_2=\{\left(x,t\right)\colon x_1\left(t\right)\leq x\leq x_2\left(t\right),t\geq0\}.
\end{equation}

\begin{figure}[t]
    \centering
    \includegraphics[trim={0 18cm 0 0},clip,width=\linewidth]{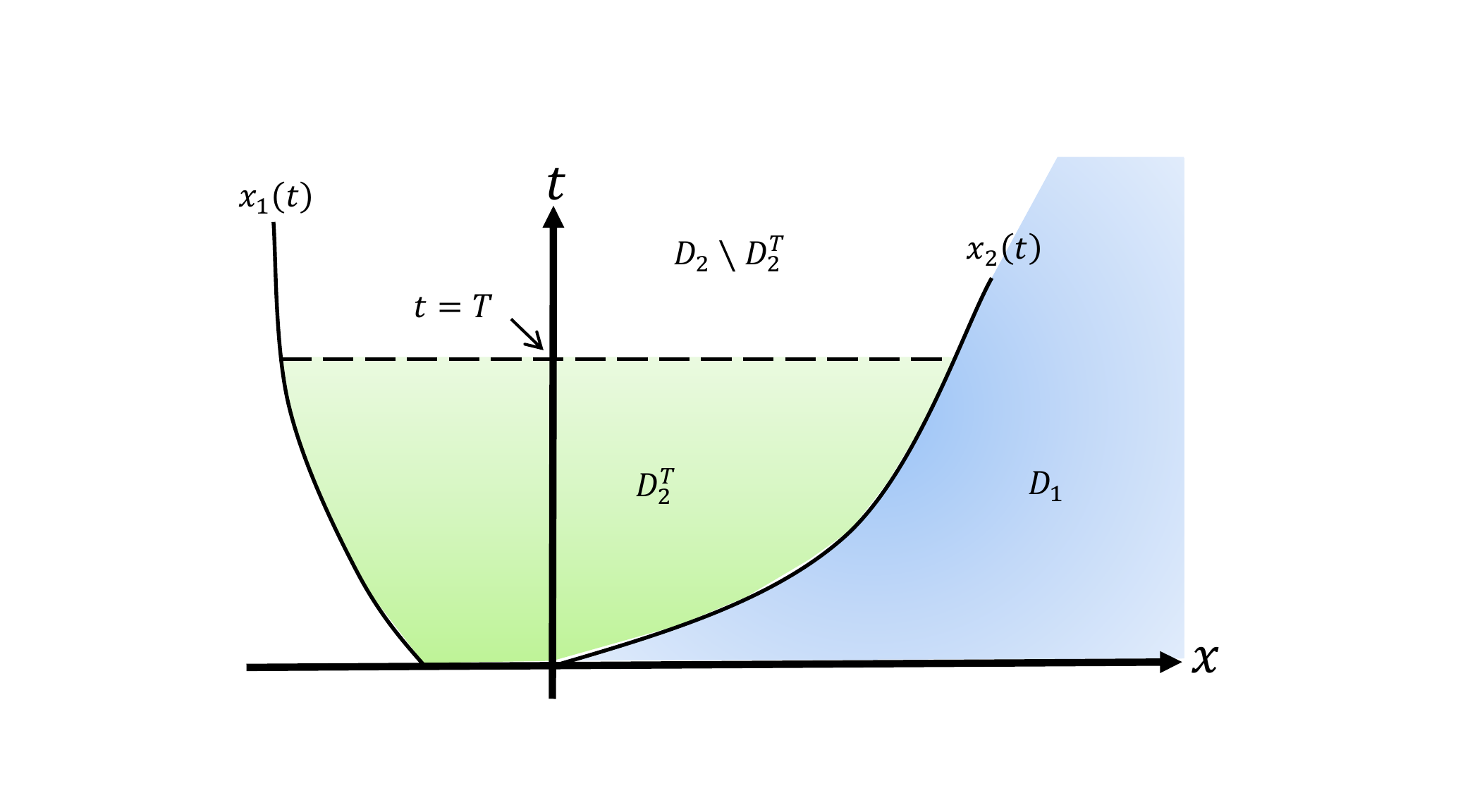}
    \caption{The subdomains \(D_1\), \(D_2\) and \(D_2^T\).}
    \label{fig:domain_char}
\end{figure}

 Let us see that for any fixed \(\overline{t}>0\) and \(0<T\leq\overline{t}\), the \(\Cf^1\) norm of the solution over the domain 
 \begin{equation}
 D_2^T=\{\left(x,t\right)\colon0\leq t\leq T,\ x_1\left(t\right)\leq x\leq x_2\left(t\right)\}    
 \end{equation}

 is bounded independently of \(T\), namely, \(\|\left(W_1,W_2\right)\|_{\Cf^1\left(D_2^T\right)}\leq C\left(\overline{t}\right)\), where \\\(C\left(\overline{t}\right)>0\) is independent of \(0<T\leq\overline{t}\). For a scheme of the situation see figure~\ref{fig:domain_char}. 
This will be done in two steps. First we will check the boundedness in \(\Cf^0\) norm. Afterwards, we will see the boundedness of the derivative (the spatial derivative).

Let \(\left(x,t\right)\in D_2^T\). In order to check the boundedness in \(\Cf^0\) norm, we will study the behaviour of the solution and its derivative along the forward and backward characteristics passing through this point. Using H2, any forward characteristic (it will have slope \(\lambda_1>0\)) passing through \(\left(x,t\right)\) must intersect the boundary \(x_1\) at one (only one) point: let this point be \(\left(\xi\left(x,t\right),\omega\left(x,t\right)\right)\). Analogously, due to the hiperbolicity, any backward characteristic (slope \(\lambda_2<0\)) must intersect the characteristic curve \(x_2\) at one (only one) point: let this point be \(\left(\eta\left(x,t\right),\beta\left(x,t\right)\right)\). Denote by \(W_2^2\left(t\right)\) the value of \(W_2\) along the characteristic boundary \(x_2\). Then
\begin{align}
    W_2\left(x,t\right)&=W_2^2\left(\beta\left(x,t\right)\right),\\
    W_1\left(x,t\right)&=g\left(W_2\left(\xi\left(x,t\right),\omega\left(x,t\right)\right),\omega\left(x,t\right)\right).\label{eq:22}
\end{align}

Since \(\beta\left(x,t\right)\leq t\) we have
\begin{equation}\label{eq:23}
    \|W_2\|_{\Cf^0\left(D_2^T\right)}\leq C\left(\overline{t}\right),\;\forall\left(x,t\right)\in D_2^T\quad\textrm{with }0<T<\overline{t}.
\end{equation}

Now, since \(\omega\left(x,t\right)\leq t\), equations~\eqref{eq:22} and estimation~\eqref{eq:23} imply
\begin{equation}
    \|W_1\|_{\Cf^0\left(D_2^T\right)}\leq C\left(\overline{t}\right),\;\forall\left(x,t\right)\in D_2^T\quad\textrm{with }0<T<\overline{t}.
\end{equation}

So we have achieved a uniform \(\Cf^0\) estimate of the solution.

Let us focus now on its derivatives, starting with the estimation of \(\partial W_2/\partial x\). When this is done, we repeat the process with \(\partial W_1/\partial x\) (although a different strategy will be needed). Let
\begin{equation}
    v=e^{k\left(W_1,W_2\right)}\frac{\partial W_2}{\partial x}
\end{equation}

where \(k\) is defined by means of
\begin{equation}
    \frac{\partial k}{\partial W_1}=-\frac{1}{\lambda_1-\lambda_2}\frac{\partial\lambda_2}{\partial W_1}.
\end{equation}

It is straight-forward to check that the following ordinary differential equation is satisfied by \(v\) along the backward characteristic \(x'\left(t\right)=\lambda\) where \(W_2\) is constant:
\begin{equation}\label{eq:27}
    \frac{\partial v}{\partial t}+\lambda_2\left(W_1,W_2\right)\frac{\partial v}{\partial x}=-e^{k\left(W_1,W_2\right)}\frac{\partial\lambda_2}{\partial W_2}v^2
\end{equation}

whose initial condition will be given on \(x_2\left(t\right)=\lambda_1\left(W_1^0,W_2^2\right)\). In order to see what form has it we note that, on \(x_2\), we have that 
\begin{equation}
    \left(W_2^2\right)'\left(t\right)=\frac{\partial W_2}{\partial t}+\lambda_1\frac{\partial W_2}{\partial x}.
\end{equation}

Since
\begin{equation}
    \frac{\partial W_2}{\partial t}=-\lambda_2\frac{\partial W_2}{\partial x}\implies \frac{\partial W_2}{\partial x}=\frac{\left(W_2^2\right)'\left(t\right)}{\lambda_1-\lambda_2}
\end{equation}

on \(x_2\). Therefor, the initial condition can be written as
\begin{equation}\label{eq:28}
    v|_{x_2}=\frac{e^{k\left(W_1^0,W_2^2\right)}}{\left(\lambda_1\left(W_1^0,W_2^2\right)-\lambda_2\left(W_1^0,W_2^2\right)\right)}\left(W_2^2\right)'\left(t\right).
\end{equation}

The initial-value problem~\eqref{eq:27} and~\eqref{eq:28} has a solution \(v\) which is given by
\begin{equation}
    v\left(x,t\right)=\frac{e^{k\left(W_1^0,W_2^2\left(\beta\right)\right)}\left(W_2^2\right)'\left(\beta\right)}{B\left(\beta\left(x,t\right),t\right)}
\end{equation}

where
\begin{align}
    B\left(\beta,t\right)=&\lambda_1\left(W_1^0,W_2^2\left(\beta\right)\right)-\lambda_2\left(W_1^0,W_2^2\left(\beta\right)\right)\\
    &+\left(W_2^2\right)'\left(\beta\right)e^{k\left(W_1^0,W_2^2\left(\beta\right)\right)}
    \\&\times\int_\beta^t\left[\frac{\partial\lambda_2}{\partial W_2}\left(W_1\left(\widetilde{x}\left(\beta,\tau\right),\tau\right),W_2^2\left(\beta\right)\right) e^{k\left(W_1\left(\widetilde{x}\left(\beta,\tau\right),\tau\right),W_2^2\left(\beta\right)\right)}\right]\dd\tau.
\end{align}

Here, \(x=\widetilde{x}\left(\beta,\tau\right)\) denotes the backward characteristic passing through the point \(\left(\eta,\beta\right)\). Using hypothesis H2 we have that \(\lambda_1-\lambda_2\) is bounded uniformly away from zero, and using hypotheses H3 and H5 we have that \(\left(W_2^2\right)'\partial\lambda_2/\partial W_2\geq0\). With this, we conclude that \(B\) is non-zero and hence we have arrived to a uniform bound of \(v\), that is, a uniform bound of \(\partial W_2/\partial x\):
\begin{equation}
    \left|\frac{\partial W_2}{\partial x}\left(x,t\right)\right|\leq C\left(\overline{t}\right),\quad\forall\left(x,t\right)\in D_2^T,\quad0<T\leq\overline{t}.
\end{equation}

Furthermore, we get that the sign of \(\partial W_2/\partial x\) is the same that \(\left(W_2^2\right)'\left(t\right)\) which is negative.

Finally, we estimate \(\partial W_1/\partial x\), in turn, in three steps:
\begin{enumerate}
    \item[Step 1] Let us see the sign of \(\partial W_1/\partial x\) on \(x=x_1\left(t\right)\). Differentiating \(W_1=g\left(W_2,t\right)\) along \(x=x_1\left(t\right)\) we have
    \begin{equation}\label{eq:W_1'}
        \frac{\partial W_1}{\partial x}=\frac{1}{x_1'\left(t\right)-\lambda_1}\left[\frac{\partial g}{\partial W_2}\cdot\left(x_1'\left(t\right)-\lambda_2\right)\cdot\frac{\partial W_2}{\partial x}\right].
    \end{equation}

Now, since
\begin{itemize}
    \item \(\lambda_2<x_1'\left(t\right)<\lambda_1\) by H2;
    \item \(\partial W_2/\partial x<0\) as we said previously;
    \item \(\partial g/\partial W_2\geq0\) by H4;
    \item \(\partial g/\partial t\leq 0\) by assumption
\end{itemize}
we obtain that
\begin{equation}
        \frac{\partial W_1}{\partial x}=\underbrace{\frac{1}{x_1'\left(t\right)-\lambda_1}}_{<0}\left[\underbrace{\frac{\partial g}{\partial W_2}}_{\geq0}\cdot\left(\underbrace{x_1'\left(t\right)-\lambda_2}_{>0}\right)\cdot\underbrace{\frac{\partial W_2}{\partial x}}_{<0}\right].
\end{equation}

and hence \(\partial W_1/\partial x\geq0\) on \(x_1\left(t\right)\).

\item[Step 2] Now we check the sign of \(\partial W_1/\partial x\) but now in \(D_2^T\). Since \(W_1\) is constant along the forward characteristic (is a Riemann invariant) we have
\begin{equation}
        W_1\left(x,t\right)=W_1\left(\xi\left(x,t\right),\omega\left(x,t\right)\right),\quad\forall\left(x,t\right)\in D_2^T
\end{equation}
where \(\xi\left(x,t\right)=x_1\left(\omega\left(x,t\right)\right)\). Therefor,
\begin{align}
        \frac{\partial W_1}{\partial x}\left(x,t\right)&=\frac{\partial W_1}{\partial t}\left(\xi,\omega\right)\frac{\partial\omega}{\partial x}\left(x,t\right)+\frac{\partial W_1}{\partial x}\left(\xi,\omega\right)\frac{\partial\xi}{\partial x}\left(x,t\right)\\
        &=\left(x_1'\left(t\right)-\lambda_1\right)\frac{\partial W_1}{\partial x}\left(\xi,\omega\right)\frac{\partial\omega}{\partial x}\left(x,t\right).
\end{align}

To determine the sign of \(\partial W_1/\partial x\) in \(D_2^T\) we first notice that \(x_1'-\lambda_1<0\) and \(\partial W_1/\partial x\left(\xi,\omega\right)\geq0\). From the definition of \(\omega\), we see that it decreases as \(x\) increases, so \(\partial\omega/\partial x\leq0\). With this, we conclude that \(\partial W_1/\partial x\geq0\) in \(D_2^T\).
\item[Step 3] Finally, let us see the \(\Cf^0\)-estimate of \(\partial W_1/\partial x\) in \(D_2^T\). Let
\begin{equation}
        u=e^{h\left(W_2,W_1\right)}\frac{\partial W_2}{\partial x}
\end{equation}

where
\begin{equation}
        \frac{\partial h}{\partial W_2}=\frac{1}{\lambda_1-\lambda_2}\frac{\partial\lambda_1}{\partial W_2}.
\end{equation}

This function \(u\) satisfies the ODE
\begin{equation}
        \frac{\partial u}{\partial t}=-e^{-h\left(W_2,W_1\right)}\frac{\partial\lambda_1}{\partial W_1}u^2
\end{equation}

along the characteristic \(x'=\lambda_1\). The initial condition is given on\\ \(x_1\left(t\right)=\left(\xi\left(x,t\right),\omega\left(x,t\right)\right)\) and has the form
\begin{equation}
        u\left(\xi,\omega\right)=e^{h\left(W_1\left(\xi,\omega\right),W_2\left(\xi,\omega\right)\right)}\frac{\partial W_1}{\partial x}\left(\xi,\omega\right).
\end{equation}

By integration we get
\begin{equation}\label{eq:Aexplodes}
        u\left(x,t\right)=\frac{e^{h\left(W_1\left(\xi,\omega\right),W_2\left(\xi,\omega\right)\right)}}{A\left(\omega,t\right)}\frac{\partial W_1}{\partial x}\left(\xi,\omega\right)
\end{equation}

with
\begin{align}
\begin{split}
        A\left(\omega,t\right)&=1+\frac{\partial W_1}{\partial x}\left(\xi,\omega\right)e^{h\left(W_1\left(\xi,\omega\right),W_2\left(\xi,\omega\right)\right)}\\
        &\times\int_\omega^t\frac{\partial\lambda_1}{\partial W_1}\left(W_1\left(\widetilde{x}\left(\omega,\tau\right),\tau\right),W_2\left(\xi,\omega\right)\right)e^{-h\left(W_1\left(\widetilde{x}\left(\omega,\tau\right),\tau\right),W_2\left(\xi,\omega\right)\right)}\dd\tau
\end{split}
\end{align}

where \(x=\widetilde{x}\left(\omega,\tau\right)\) is the forward characteristic passing through the point \(\left(\xi,\omega\right)\). Now, due to hypothesis H5 we have that \(\partial\lambda_1/\partial W_1\geq0\) and due to Step 1 we have that \(\partial W_1/\partial x\geq0\) on \(x_1\left(t\right)\) so, we see that \(A\left(\omega,t\right)\) is never zero. With this, using hypothesis H2 and the uniform estimates of \(W_2,W_1\) and \(\partial W_2/\partial x\) in \(D_2^T\) we get the uniform estimate for \(\partial W_1/\partial x\) on \(D_2^T\).

Concluding, since this estimate is independent of \(T\) and \(D=D_1\cup D_2\) we complete the proof.
\end{enumerate}
\end{proof}

Let us check if these hypotheses are satisfied for our problem. Writing all together we have:
\begin{align}
        W_{1,2}&=u\pm4\sqrt{\frac{\beta}{2\rho}}A^{1/4},\\
        A&=\left(\frac{W_1-W_2}{4}\right)^4\left(\frac{\rho}{2\beta}\right)^2,\label{eq:Aw}\\
        u&=\frac{W_1+W_2}{2},\label{eq:uW}\\
        x_1(t)&=0\;\forall t,\\
        W_2(x_1(t),t)&=W_2(0,t)=\uinf-4\sqrt{\frac{\beta}{2\rho}}\Ainf^{1/4},\label{eq:g}\\
        \lambda_{1,2}&=u\pm4\sqrt{\frac{\beta}{2\rho}}A^{1/4},\\
        W_1^0(x)&=W_2^0(x)=u_0(x)\pm4\sqrt{\frac{\beta}{2\rho}}A_0^{1/4},
\end{align}

where \(\uinf=u(0,t)\) and \(\Ainf=A(0,t)\). Now, regarding the hypotheses:
\begin{enumerate}
    \item[H1] We need the initial velocity and the boundary conditions to be \(\Cf^1\), which is physically correct (without abrupt accelerations).
    \item[H2] Since \(\lambda_1>0\) and \(\lambda_2<0\) as we said before, this hypothesis is satisfied.
    \item[H3] The boundedness is achieved due to biological reasons, and we need to choose \(u_0\) to be constant.
    \item[H4] Replacing equations~\eqref{eq:Aw} and~\eqref{eq:uW} into the expression of \(g\) \eqref{eq:g}, we see that this hypothesis is also satisfied.
\end{enumerate}
But the problem comes with the fifth and sixth hypotheses, that are false in our case. Indeed, 
\begin{equation}
\lambda_2=W_2\implies \frac{\partial\lambda_2}{\partial W_2}=1>0.
\end{equation}

Nevertheless, since it is a sufficient condition, we could get smoothness of the solution in this case too. In fact, this is checked via numerical simulations in chapter~\ref{chap:results}.

With a slightly modification of the equations and other tube law, it is shown in~\cite{canic03} that under some plausible conditions such as the no-singularity of the cross-sectional area and a pulsating boundary conditions (such as heart beats) the thesis of the previous theorem is achieved.

\subsection{Shock wave formation in compliant arteries}\label{subsec:poster}

We have just seen that we can not assure smooth solutions. This makes also clinical sense since if the heart beats are too abrupt, for example, the blood could be faster than the usual blood wave: this is precisely the condition for a \emph{shock wave} to form.

Motivated for this fact, we study when and where the first shock wave could be formed and which factors are relevant to it.

Next we present an original result where we obtain an explicit formulation for when the first shock wave appears. Furthermore, we will see the direct influence of the physical parameters to characterise from clinical data if this pathology happens. We have followed the steps of Čanić and Kim~\cite{canic03}, although they used a slight different formulation. Keener and Sneyd also obtained a similar (but less general) result in~\cite{keener98}. This result was presented in \(2017\) in the \textit{4º Congreso de jóvenes investigadores} (IV Conference for young researchers) (see~\cite{rodero17})

\begin{theo}\label{theo:shock}
Assuming constant initial data
\begin{equation}
    A(x,0)=A_0,\qquad u(x,0)=0,
\end{equation}

the time \(t_s\) of the first shock formation is given by
\begin{equation}
    t_s=\omega+\frac{\lambda_1}{\uinf'(t)}=\omega+\frac{\uinf+4\sqrt{\beta/(2\rho)}\Ainf^{1/4}}{\uinf'(t)}.
\end{equation}

where \(\omega\) is the first time when the forward characteristic intersects the left spatial boundary.
\end{theo}
\begin{proof}

In terms of the Riemann invariants, the initial data read \(W_1^0(x)=W_2^0(x)=u_0(x)\pm4\sqrt{\frac{\beta}{2\rho}}A_0^{1/4}\). For this set of initial data, \(W_1\) is constant everywhere in the region of smooth flow; the characteristics are straight lines in
\begin{equation}
    D_2=\{(x,t)\colon x_2(t)\leq x<+\infty,\ t\geq0\},
\end{equation}

where \(x_2(t)\) is the forward characteristic \(x_2'=\lambda_1\) emanating from \((0,0)\). The solution in region
\begin{equation}
    D_1=\{(x,t)\colon 0\leq x\leq x_2(t),\ t\geq0\},
\end{equation}

bounded by the left boundary \(x_1=0\) and the forward characteristic \(x_2\), is driven by \(u(\cdot,t)\) on \(x_1\) and will develop shock waves due to the fact that \(u'(\cdot,t)\) changes sign.

To estimate the time \(t_s\) we note that at the point \((t_s,x_s)\) the partial derivative \(\partial W_1/\partial x\) blows up, as we can see in figure~\ref{fig:char_shock}. This occurs at the point where the denominator \(A(\omega,t)\) in~\eqref{eq:Aexplodes} becomes equal to zero. We recall that its expression was
\begin{align}\label{eq:Aalpha}
\begin{split}
        A\left(\omega,t\right)&=1+\frac{\partial W_1}{\partial x}\left(\xi,\omega\right)e^{h\left(W_1\left(\xi,\omega\right),W_2\left(\xi,\omega\right)\right)}\\
        &\times\int_\omega^t\frac{\partial\lambda_1}{\partial W_1}\left(W_1\left(\widetilde{x}\left(\omega,\tau\right),\tau\right),W_2\left(\xi,\omega\right)\right)e^{-h\left(W_1\left(\widetilde{x}\left(\omega,\tau\right),\tau\right),W_2\left(\xi,\omega\right)\right)}\dd\tau.
\end{split}
\end{align}

Hence, it can be calculated by recalling that \(\lambda_1=W_1\), what means that \(\partial\lambda_1/\partial W_1=1\) and that \(W_1=W_1^0\) everywhere. This implies that in~\eqref{eq:Aalpha}
\begin{equation}
    e^{h\left(W_1\left(\xi,\omega\right),W_2\left(\xi,\omega\right)\right)-h\left(W_1\left(\widetilde{x}\left(\omega,\tau\right),\tau\right),W_2\left(\xi,\omega\right)\right)}=1
\end{equation}

and so
\begin{equation}
    A(\omega,t)=1+\frac{\partial W_1}{\partial x}(\xi,\omega)(t-\omega). 
\end{equation}

From equation~\eqref{eq:W_1'} we see that
\begin{equation}
    \left.\frac{\partial W_1}{\partial x}\right|_{x_1=0}=-2\frac{u'(t)}{\lambda_1}+\frac{\lambda_1}{\lambda_2}\frac{\partial W_2}{\partial x}.
\end{equation}

Since \(\partial W_2/\partial x=0\), we obtain
\begin{equation}
    A(\omega,t)=1-\frac{u'(t)}{\lambda_1}(t-\omega).
\end{equation}

Therefor, isolating, the first time the shock forms is equal to
\begin{equation}
    t_s=\omega+\frac{\lambda_1}{\uinf'(t)}=\omega+\frac{\uinf+4\sqrt{\beta/(2\rho)}\Ainf^{1/4}}{\uinf'(t)}.
\end{equation}
\end{proof}

Two remarkable conclusions become deduced from this result: the shock will be produced earlier if the inflow accelerates and if the walls of the vessel are less rigid (due to the \(\beta\) factor).

A case that could be interesting from a medical point of view is to see, in the case of the aorta since is the one where we can find more literature, if this model can explain the pistol-shot heard in aortic insufficiency. Taking as an approximation the measures done in~\cite{erbel06,mao08,cozijnsen11} we assume \(A_0=\Ainf\approx4\times10^{-2}m\). Regarding blood flow velocity we have assumed \(1\) m/s, following the measures found in~\cite{mowat83}. Recalling that
\begin{equation}
    \beta=\frac{\sqrt{\pi}Eh}{(1-\nu^2)A_0},
\end{equation}

we will use the parameters of appendix~\ref{appendix:bio}.

Now, for a healthy human being we have taken \(\uinf'=7\)m/s\textsuperscript{2}, following the correlations of~\cite{mowat83} and a Young's modulus of \(E=10^5\ Pa\). 
Hence, the value for the first time and place when shock forms is
\begin{equation}
    t_s\approx \frac{1+4\sqrt{29633.5/(2\times1050)}(4\times10^{-2})^{1/4}}{7}\approx1.1\text{s}
\end{equation}

and
\begin{equation}
    x_s=t_s\lambda_1\approx8.5\textrm{m},
\end{equation}

which, according to~\cite{dotter50} is far from the mean length value of the aorta, \(33.2\) cm.

Now, for a sick human being, say by aortic insufficiency, the heart increases its volume and since the aortic valve does not close properly, the muscle must do a greater contraction, so a greater blood flow acceleration happens in each heart beat (see~\cite{hall15}). Using the same bibliography as we have previously consulted, we could have that \(\uinf'=15\) m/s\textsuperscript{2}. Moreover, the arteries could be not rigid enough (what may cause an aneurysm), say \(E=2\cdot10^{3}Pa\). With this variation we would have:
\begin{equation}
    t_s\approx \frac{1+4\sqrt{592.67/(2\times1050)}(4\times10^{-2})^{1/4}}{15}\approx0.13\text{s}
\end{equation}

and
\begin{equation}
    x_s=t_s\lambda_1\approx0.25\textrm{m}.
\end{equation}

As it should be, this case is worse for the patient and a shock wave is formed inside the typical length of the aorta. For the interested reader, some research has been done in this line, although with different and sometimes less general models. See~\cite{elgarayhi13,painter08,shoucri07}. Henceforth, we move on to the numerical simulations. Although shock waves have not been simulated due to complexity of controlling the inflow data, healthy cases have been simulated. For these simulations the method used has been the Galerkin Discontinuous Finite Element Method, presented in the next chapter.

%% file: Numeric/Why_Galerkin.tex
In computational fluid dynamics, specially in medical applications, accuracy is preferred rather than velocity or simplicity in numerical methods. We must require a number of key properties such as flexibility in geometry, robustness, efficiency, high-order or variable order accuracy --- since long time integration is needed --- and, if possible, possibility of high performance computing.

During the last decades, a number of numerical techniques for the solution of nonlinear conservation laws, nonlinear convection-diffusion problems and compressible flow have been developed. We first review briefly the three most popular: finite differences, finite volumes and finite elements.
\begin{itemize}
    \item Finite-difference methods (FDM) are discretisation methods for solving differential equations by approximating them with difference equations, in which finite differences approximate the derivatives. This is usually done using Taylor series expansion and truncating at the desired order. Due to this, we can achieve any order desired for the method, which is an advantage.

    The two sources of error in finite difference methods are round-off error --- loss of precision due to computer machine --- and discretisation error, the difference between the exact solution and the exact approximation assuming perfect arithmetic (that is, assuming no round-off). 
    
    This method relies on discretising a function on a grid, so to approximate the solution to a problem, one must first discretise the problem's domain. This is usually done by dividing the domain into a uniform grid. Therefor, the solution is recovered in a pointwise way. It can be shown that the truncation error is proportional to the step sizes (time and space steps). If we reduce the step size or increase the truncation of the Taylor's expansion the accuracy of the approximate solution increases, but also the simulation's duration. Therefor, a reasonable balance between data quality and simulation duration is necessary for practical usage. Large time steps are useful for increasing simulation speed in practice. However, time steps which are too large may create instabilities and affect the data quality~\cite{hoffman01,majumdar05}. There are quite a lot of variations of the FDM, both explicit and implicit cases and quite a lot of theory behind these methods, since FDMs are the dominant approach to numerical solutions of partial differential equations~\cite{grossmann07}. For a first approach to these methods we refer the interested reader to~\cite{iserles09,arandiga08,arandiga00}. The main drawback of these methods is that complex geometries are not allowed in a simple way.
    
    \item Similar to the FDM, in the finite volume method (FVM), values are calculated at discrete places on a meshed geometry~\cite{leveque02,toro13,kroener97}. ``Finite volume'' refers to the small volume surrounding each node point on a mesh (\emph{control volumes}, or \emph{cells}). In the finite volume method, volume integrals in a partial differential equation that contain a divergence term are converted to surface integrals, using the divergence theorem. These terms are then evaluated as fluxes at the surfaces of each finite volume. Thus, FVMs use piecewise constant approximations. Because the flux entering a given volume is identical to that leaving the adjacent volume, these methods are conservative. Another advantage of the finite volume method is that it is easily formulated to allow for unstructured meshes and, since it uses cell averages, it allows discontinuities. The problem comes when we want to achieve high order on general grids. Also, there are requirements of grid smoothness non trivial at all~\cite{godlewski13,barth90,harten91}. For a survey of various techniques and results from the FVM, we refer the reader to the monograph~\cite{eymard00}.
    
    \item The next step is the family of finite element methods (FEM). It is also referred to as finite element analysis (FEA). The main basic steps of the FEM are~\cite{logan11}:
    \begin{enumerate}
        \item Divide the whole domain of the problem into subdomains (can be a regular or unstructured mesh) called finite elements. 
        \item Convert the differential equations we are dealing with into its weak form, multiplying by an arbitrary function and integrating.
        \item Choose appropriate test functions (they are usually polynomials) to arrive to an algebraic system, and solve it.
    \end{enumerate}
    Regarding the history, as it is often the case with original developments, it is rather difficult to quote an exact date of invention, but in \cite{bathe06}, the roots of the FEM are traced back to three separate research groups: applied mathematicians \cite{courant43}, physicists \cite{synge57} and engineers \cite{hinton68}, although the FEM obtained its real impetus from the development of engineers.  
    These methods can achieve high order and can deal without many problems complex geometries. The main issue is that they are implicit in time and when we are dealing problems with direction (such a diffusion) are not really well suited.
\end{itemize}
To sum up the observations, we have:
\begin{table}[H]
\centering
\begin{tabular}{c|ccc}
    & Complex geometries & High-order accuracy & Explicit semi-discrete form \\ \hline
FDM & \(\times\)         & \checkmark          & \checkmark                  \\
FVM & \checkmark         & \(\times\)          & \checkmark                  \\
FEM & \checkmark         & \checkmark          & \(\times\)                 
\end{tabular}
\caption{Advantages and disadvantages of the most used methods for solving differential equations numerically.}
\end{table}
The ideal case would be to have a scheme with local high-order and flexible elements as in the FEM; the nice handling of discontinuities as in the FVM; and the explicit and relatively simple (semi-)discrete form of the FDM. We can find indeed a method with these components, inside the FEMs, called \emph{Discontinuous Galerkin Finite Element Method} (DG-FEM, DGM or simply DG).

As happened with FEM, DG methods can be considered as numerical schemes for the weak formulation of the equations. They were first applied to first-order equations by Reed and Hill in~\cite{reed73}, but their widespread use followed from the application to hyperbolic problems by Cockburn and collaborators in a series of articles~\cite{cockburn89,cockburn90,cockburn98}. In the DG-FEM framework, the solution is recovered in a more continuous way, with polynomials, without the need of using a reconstruction operator (such as interpolation). This feature of DG schemes is in common with the classical FEM. But, unlike classical finite elements, the numerical solution given by a DG scheme is discontinuous at element interfaces and this discontinuity is resolved by the use of a so-called numerical flux function, which is a common feature with FVM.

Hence, in the next sections we will give more details about this method, starting from definitions and some mathematical notation up to the computational implementation. As we said in the previous chapters, we are treating with a one-dimensional problem, so both the theory as the implementation are a little simpler. For a more general theoretical treatment (dimensions \(2\) and \(3\)) we refer the interested reader to the lecture notes~\cite{dolejvsi16}. 

%% file: Numeric/Notation_Galerkin.tex
\begin{figure}[htb]
    \centering
    \includegraphics[trim={0 5cm 7cm 5cm},clip,width=.6\linewidth]{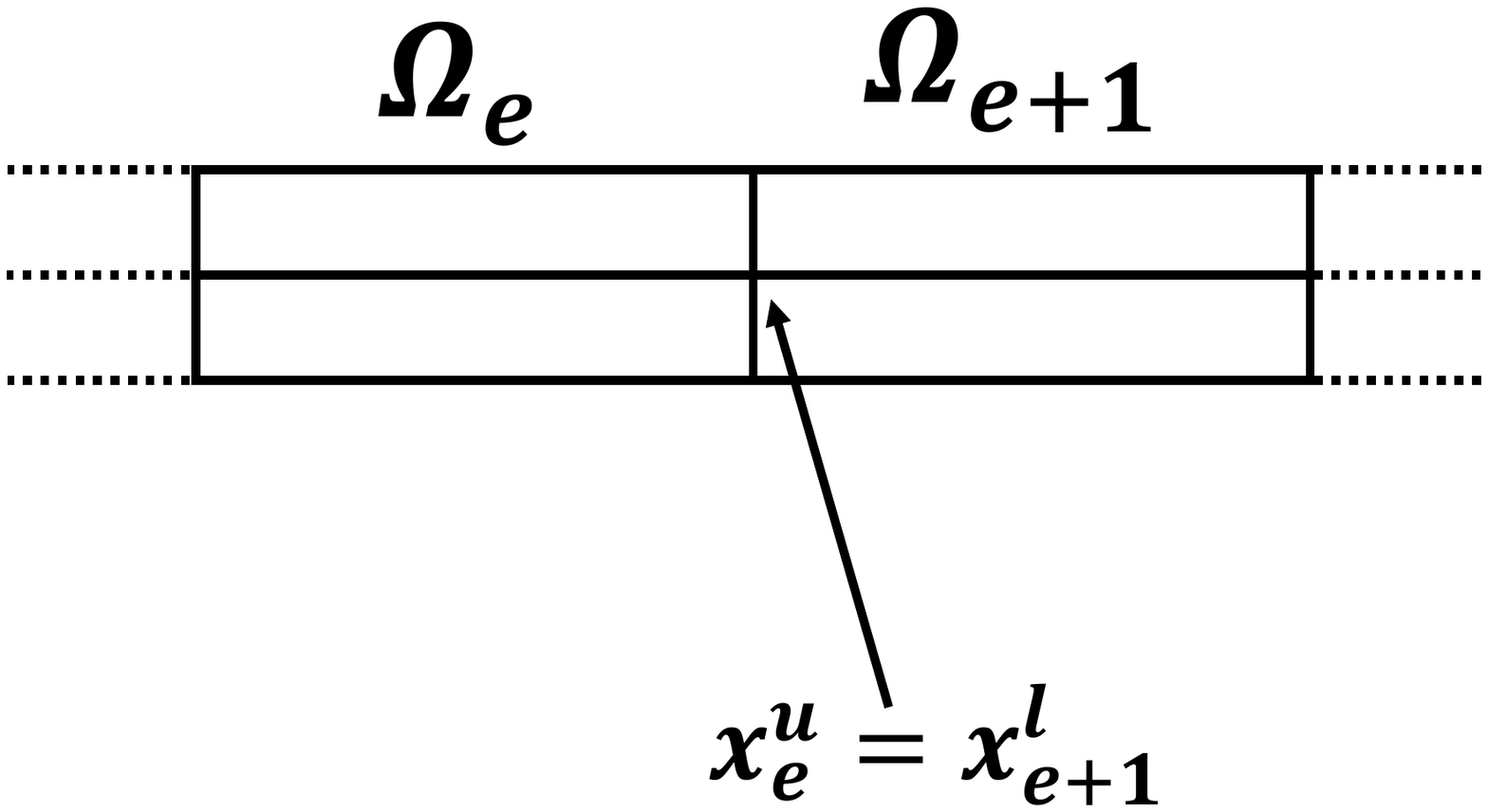}
    \caption{One dimensional finite element discretisation.}
    \label{fig:fem}
\end{figure}

The idea of this family of methods is to split the domain we are dealing with, namely \(\Omega=[a,b]\), into a set of the so-called \emph{elements}, in our case they are subintervals as shown in figure~\ref{fig:fem}. Formally we have
\begin{gather}
    \Omega=\bigcup_{e=1}^{N_{\textrm{el}}}\Omega_e,\\
    \Omega_e=[x^u_e,x^l_e]
\end{gather}

where 
\begin{equation}
    x_e^u=x_{e+1}^l.
\end{equation}

This subdivision of the general domain \(\Omega\) will be denoted by \(\Th\) and we call it \emph{the triangulation} of \(\Omega\) (for analogy to the two and three-dimensional cases).

By \(\Fh\) we denote the set of boundaries of the elements, namely
\begin{equation}
    \Fh=\left\{x_e^l\right\}_{e=1}^{\Nel}\cup\left\{x_{\Nel}^u\right\}=\left\{x_e^u\right\}_{e=1}^{N_el}\cup\left\{x_{1}^l\right\}.
\end{equation}

Moreover, we will sometimes treat in a separated way the inlet and outlet boundaries, so
\begin{gather}
    \Fh^{io}=\left\{x_1^l\right\}\cup\left\{x_{\Nel}^u\right\},\\
    \Fh^{W}=\Fh\setminus\Fh^{io}.
\end{gather}

First we recall some standard notation from the measure and integral Lebesgue theory. 

\subsection{Measure and integral Lebesgue notation}
Let \(M\subset\RR^n\), \(n=1,2,\dots\) be a Lebesgue measurable set. We recall that two measurable functions are \emph{equivalent} if they differ at most on a set of zero Lebesgue measure. The \emph{Lebesgue space} \(L^p(M)\), with \(1\leq p<\infty\) is the linear space of all functions measurable on \(M\) (more precisely, of classes of equivalent measurable functions) such that
\begin{equation}
    \int_M|u|^p\dd x<+\infty.
\end{equation}

With this, let \(k\geq0\) be an artbitrary integer and \(1\leq p<\infty\). We define the \emph{Sobolev space} \(W^{k,p}(M)\) as the space of all functions from the space \(L^p(M)\) whose distributional derivatives \(D^\alpha u\), up to order \(k\), also belong to \(L^p(M)\), \textit{i.e.},
\begin{equation}
   W^{k,p}(M)=\{u\in L^p(M)\colon D^\alpha u\in L^p(M)\,\forall\alpha,\ |\alpha|\leq k\}.
\end{equation}

For \(p=2\), \(W^{k,2}\) is a Hilbert space and we denote it by \(H^k(M)\). Now, the DGM is based on the use of discontinuous approximations. This is the reason why over a triangulation \(\Th\), we define the \emph{broken Sobolev space}
\begin{equation}
    H^k(\Omega,\Th)=\{u\in L^2(\Omega)\colon u|_{\Omega_e}\in H^k(\Omega_e)\,\forall\Omega_e\in\Th\}.
\end{equation}

Since we will be working with two dimensional vector-valued functions we will denote
\begin{equation}
    \textbf{H}^k(\Omega,\Th)=H^k(\Omega,\Th)\times H^k(\Omega,\Th).
\end{equation}

The DGM can be characterised as a finite element technique using piecewise polynomial approximations, in general discontinuous on interfaces between neighbouring elements. Therefor, we introduce a finite-dimensional subspace of \(H^k(\Omega,\Th)\), where the approximate solution will be sought.

Let \(\Th\) be a triangulation of \(\Omega\) and let \(p\geq0\) be an integer. We define the space of discontinuous piecewise polynomial functions
\begin{equation}
    S_{hP}=\{v\in L^2(\Omega)\colon v|_{\Omega_e}\in \mathcal{P}_P(\Omega_e)\,\forall\Omega_e\in\Th\},
\end{equation}

where \(\mathcal{P}_P(\Omega_e)\) denotes the space of all polynomials of degree \(\leq P\) on \(\Omega_e\). Obviously, \(S_{hP}\subset H^k(\Omega,\Th)\) for any \(k\geq1\) and its dimension is \(\dim S_{hP}=P+1\).

Similarly as we have done before, we will denote
\begin{equation}
    \textbf{S}_{hP}=S_{hP}\times S_{hP}.
\end{equation}

Finally, since we are going to deal with vectorial functions, we will understand operations such as integration of derivation of the vector as the operation component-wise.

%% file: Numeric/Galerkin.tex
In order to give a more general situation, this time we will consider the source term, namely the friction term. Since we are assuming that the blood is a Newtonian fluid, the friction term has the form
\begin{equation}
    f=-K_Ru,
\end{equation}

where \(K_R\) is a strictly positive quantity which represents the viscous resistance of flow per unit length of tube. We still assume that \(\alpha=1\) (a flat profile) for the reasons we discussed in subsection~\ref{ssec:charac_NS}. Hence, we can write the system as
\begin{equation}\label{eq:NS_source}
    \UU_t+\FF_x=\TT 
\end{equation}

with
\begin{equation}
    \UU=\begin{bmatrix}A\\u\end{bmatrix},\quad\FF=\begin{bmatrix}uA\\\frac{u^2}{2}+\frac{p}{\rho}\end{bmatrix}\quad\text{and}\quad\text\TT=\begin{bmatrix}0\\\frac{K_R}{\rho}u\end{bmatrix}.
\end{equation}

In order to derive the discrete problem, we assume that there exists an exact solution \(\UU\in\Cf^1(\textbf{H}^1(\Omega,\Th);[0,T])\), where \(T>0\) represents the final time, of the Navier-Stokes equations~\eqref{eq:NS_source}. Then we multiply~\eqref{eq:NS_source} by a test function \(\pmb{\psi}\in\textbf{H}^1(\Omega,\Th)\) and integrate over the domain. Specifying the dependency on \(\UU\) we obtain:

\begin{equation}\label{eq:DG1}
    \int_\Omega\textbf{U}_t\pmb{\psi}\dd x+\int_\Omega\textbf{F}_x(\UU)\pmb{\psi}\dd x=\int_\Omega\TT(\UU)\pmb{\psi}.
\end{equation}

Let us denote, for a measurable set \(M\), the inner product of \(L^2(M)\) with
\begin{equation}
    \int_{M}uv\dd x=\pesc{u}{v}{M}.
\end{equation}

At this point is where the main difference with respect to the FEM appears. We do not demand continuity between elements, and since the elements of the triangulation \(\Th\) are disjoint, we can decouple~\eqref{eq:DG1} and work separately with each \(\Omega_e\in\Th\) as:
\begin{equation}
    \pesc{\UU_t}{\pmb{\psi}}{\Omega_e}+\pesc{\FF_x(\UU)}{\pmb{\psi}}{\Omega_e}=\pesc{\TT(\UU)}{\pmb{\psi}}{\Omega_e}.
\end{equation}

We integrate by parts the flux term (in higher dimensions this would be equivalent to apply Green's theorem):
\begin{equation}\label{eq:after_parts}
    \pesc{\UU_t}{\pmb{\psi}}{\Omega_e}+\left[\FF(\UU)\cdot\pmb{\psi}\right]_{x_e^l}^{x_e^u}-\pesc{\FF(\UU)}{\pmb{\psi}_x}{\Omega_e}=\pesc{\TT(\UU)}{\pmb{\psi}}{\Omega_e}.
\end{equation}

We remark that the term evaluated in the boundary is interpreted as component-wise. It is not the usual dot product of vectors.

As we said before, we discretise the solution \(\UU\) with an approximation \(\UU_h\in\textbf{S}_{hP}\). Due to the fact that \(\dim S_{hP}=P+1\), we can express this discrete solution as a linear combination
\begin{equation}\label{eq:linear_expansion}
    \left.\UU_h\right|_{\Omega_e}=\sum_{p=0}^P\hat{\UU}_p^e\varphi_p
\end{equation}

where \(\{\varphi_p\}_{p=0}^P\) is a basis of the space \(S_{hP}\). In the previous expression, the product denotes a component-wise product.

\begin{figure}[thb]
    \centering
    \includegraphics[trim={3cm 18cm 2cm 2cm},clip,width=.7\linewidth]{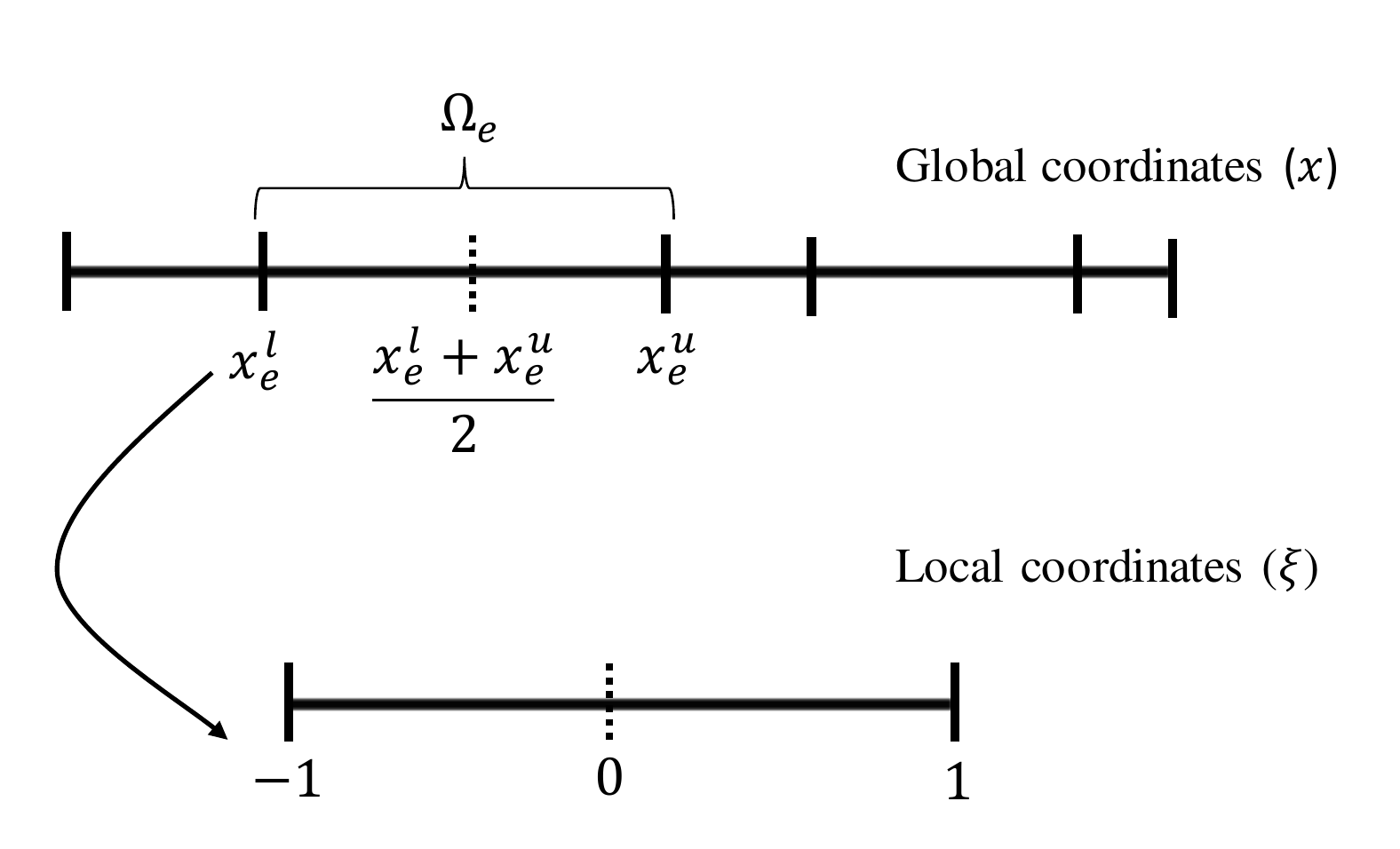}
    \caption{Change of coordinates from global to local at an element.}
    \label{fig:glob2loc}
\end{figure}

Now, this expansion is carried out at each element \(\Omega_e\in\Th\), so we need to do a readjustment of the coordinates. More precisely, the global semidiscrete solution \(\UU_h\) is defined in all the length of the artery, namely \(0\leq x\leq l\). But at each element we need to resize to \(-1\leq\xi\leq1\), see figure~\ref{fig:glob2loc}. To carry out this conversion, affine mappings (common in FEM) are used. Their expression, from global to local and local to global coordinates are, respectively,
\begin{align}
    \xi(x_e)&=\frac{2x-x_e^l-x_e^u}{x_e^u-x_e^l},\\
    x_e(\xi)&=x_e^l\frac{1-\xi}{2}+x_e^u\frac{1+\xi}{2},
\end{align}

(obviously one mapping is the inverse function of the other one). More precisely equation~\eqref{eq:linear_expansion} has the form
\begin{equation}\label{eq:DoF}
    \left.\UU_h(x_e,t)\right|_{\Omega_e}=\sum_{p=0}^P\hat{\UU}_p^e(t)\varphi_p(\xi).
\end{equation}

We remark that the coefficients of the previous expressions only depend on time. This is the characteristic of the so-called \emph{modal form} of a finite element method. The other possibility is the \emph{nodal form}, where the solution is calculated in some points (nodes) of each element. These coefficients are usually called~\emph{degrees of freedom}. Here is where the main distinctive of the Galerkin methods come out: the test function \(\pmb{\psi}\) also belongs to the finite-dimensional space \(\textbf{S}_{hP}\), indeed it will be the same on each component, so we will change the notation to \(\psi\). With this, the first term of~\eqref{eq:after_parts} is approximated by
 \begin{equation}
    \pesc{\UU_t(x_e(\xi),t)}{\psi(\xi)}{I}=\pesc{\frac{\partial\UU(x_e(\xi),t)}{\partial t}}{\psi(\xi)}{I}\approx\sum_{p=0}^P\frac{\partial\hat{\UU}_p^e(x_e(\xi),t)}{\partial t}\pesc{\varphi_p(\xi)}{\psi(\xi)}{I},
\end{equation}

where \(I=[-1,1]\).
Moreover, by substitution of variables, the third term of equation~\eqref{eq:after_parts} becomes
\begin{align}
\begin{split}
    \pesc{\FF(\UU)}{\pmb{\psi}_x}{\Omega_e}=&\int_{\Omega_e}\FF(\UU(x,t))\psi_x(\xi(x_e))\ \dd x\\
    =\frac{2}{x_e^u-x_e^l}&\int_{-1}^1\FF(\UU(\xi,t))\psi_x(\xi)\ \dd\xi=\frac{2}{x_e^u-x_e^l}\pesc{\FF(\UU)}{\psi_x}{I}.
\end{split}
\end{align}

Up to now there are some unclear points in the reasoning previously done. First of all, the election of the test and basis functions. In the next subsection we will choose them motivating the election.

\subsection{Election of test and basis functions}\label{subsec:test}
In a benchmark more common in calculus, test functions are chosen as \(\Cf^\infty\) with compact support. Nevertheless, for practical uses we consider both the test function \(\psi\) and the basis functions \(\{\varphi_p\}_p\) belonging to the piecewise polynomial space \(S_{hP}\). More precisely, in order to avoid the appearance of more coefficients we choose the test functions as one of the basis functions. The idea of this is that in the expression
\begin{equation}
    \sum_{p=0}^P\frac{\partial\hat{\UU}^e_p(x_e(\xi),t)}{\partial t}\pesc{\varphi_p(\xi)}{\varphi_q(\xi)}{I}
\end{equation}

by varying the test function all over the basis functions, we decouple the system, so we can isolate the derivative of the degrees of freedom. The previous equation can be seen as a product matrix-vector
\begin{equation}
    \mathcal{M}\hat{\UU}^e
\end{equation}

where
\begin{equation}\label{eq:defM}
    \mathcal{M}_{ij}=\pesc{\varphi_i}{\varphi_j}{I},\qquad\hat{\UU}^e=\begin{bmatrix}\hat{\UU}_1^e\\\vdots\\\hat{\UU}_P^e\end{bmatrix}.
\end{equation}

The most straightforward election would be to choose the canonical basis
\begin{equation}
    \varphi_p(\xi)=\xi^{p-1}.
\end{equation}

\begin{figure}[thb]
    \centering
    \includegraphics[width=\linewidth]{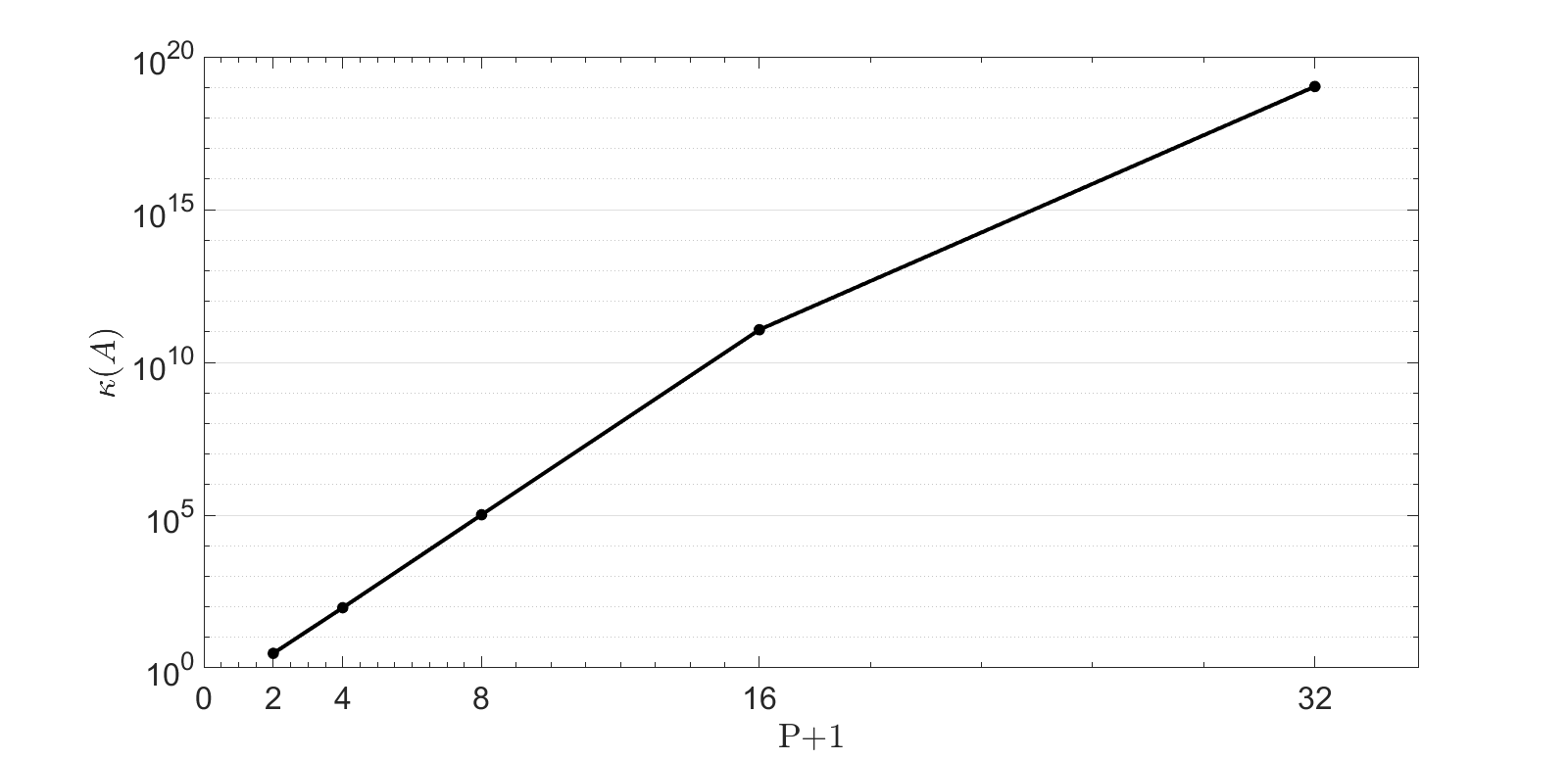}
    \caption{Condition number (in norm \(1\)) of the matrix \(\mathcal{M}\). As we can see, it grows almost exponentially.}
    \label{fig:cond_canonical}
\end{figure}

As a quality measure, we can look to the condition number of \(\mathcal{M}\). In the norm \(1\), for example, we can see that it grows very quickly, as we present in figure~\ref{fig:cond_canonical}. Therefor, this election does not look like a robust choice for high \(P\).

With this understanding, we look for a better basis, specifically one which makes \(\mathcal{M}\) diagonal. \textit{I.e.}, we want an orthogonal basis. The basis usually chosen is the one formed by Legendre's polynomials. Instead of enumerating the properties, in this subsection we will reason in a straightforward manner its appearance. Some useful properties of these polynomials are proved using basic algebra and calculus.

First of all, we apply the \emph{Gram-Schmidt process} to the canonical basis and we get a complete orthogonal basis \(\{Q_n\}_{n=0}^P\), where \(Q_0(x)=1\) and
\begin{equation}
    Q_n(x)=x^n-\sum_{i=0}^{n-1}\frac{\pesc{x^n}{Q_i}{I}}{\pesc{Q_i}{Q_i}{I}}Q_i(x).
\end{equation}

Now we need a result about orthogonal bases extracted from~\cite{kumar13}:
\begin{prop}\label{prop:2.3}
Let \(\{p_n\}_{n=0}^\infty\) be an orthogonal polynomial sequence (OPS) in \(L^2([a,b])\) with \(\deg (p_n)=n\). Then
\begin{enumerate}[label=(\roman*)]
    \item \(\pesc{x^k}{p_n}{[a,b]}=0\) for \(k=0,1,\dots,n-1\).\label{prop:pol1}
    \item The polynomial \(p_n(x)\) of degree \(n\) has exactly \(n\) simple zeros in the open interval \(]a,b[\).\label{prop:pol2}
    \item\label{prop:pol3} There is a recurrence relation of the form
    \begin{equation}
        p_{n+1}(x)=(\alpha_nx+\beta_n)p_n(x)+\gamma_np_{n-1}(x),\;\forall n\geq1
    \end{equation}

    where \(\alpha_n,\beta_n\gamma_n\) are real constants depending on \(n\).
\end{enumerate}
\end{prop}

\begin{proof}
\begin{enumerate}[label=(\roman*)]
\item Let \(q\) be a polynomial of degree \(k\). Since \(\deg(p_n)=n\), it follows that \(q\) lies in the span of \(\{p_0,p_1,\dots,p_k\}\), \textit{i.e.}, \(q=\sum_{0\leq i\leq k}a_ip_i\), \(a_i\in\RR\). If \(k<n\), then by orthogonality
\begin{equation}
    \pesc{q}{p_n}{I}=\sum_{i=0}^ka_i\pesc{p_i}{p_n}{I}=0.
\end{equation}

In particular, \(\pesc{x^k}{p_n}{I}=0\) for \(k<n\). This proves~\ref{prop:pol1}.
   
\item Let \(x_1,x_2,\dots,x_m\) be distinct real zeros of \(p_n\) in \(]a,b[\). Then we can factorise
    \begin{equation}
        p_n(x)=(x-x_1)^{\epsilon_1}(x-x_2)^{\epsilon_2}\dots(x-x_m)^{\epsilon_m}r(x),
    \end{equation}

    where \(\epsilon_i\geq1\) and the polynomial \(r(x)\) has no zero on \(]a,b[\). Thus \(r(x)>0\) (or else \(r(x)<0\)), \(\forall x\in]a,b[\).
    
    Let 
    \begin{equation}
        \phi(x)=(x-x_1)^{\delta_1}(x-x_2)^{\delta_2}\dots(x-x_m)^{\delta_m},
    \end{equation}

    where \(\delta_i=0\) or \(1\) according as \(\epsilon_i\) is even or odd. Then \(\deg(\phi)\leq m\) and \(\phi(x)p_n(x)\geq0\) (or else \(\phi(x)p_n(x)\leq0\)) \(\forall x\in]a,b[\). This shows that \(\pesc{\phi}{p_n}{[a,b]}\neq0\). Thus, in view of the part~\ref{prop:pol1}, \(\deg(\phi)\geq n\). Since \(\deg(\phi)\leq m\leq n\), we must have \(m=n\). This proves~\ref{prop:pol2}.
    
\item Let \(\alpha_n\) be the ratio of leading coefficients of \(p_{n+1}\) and \(p_n\). Then \(p_{n+1}-\alpha_nxp_n(x)\) is a polynomial of degree at most \(n\). Let
    \begin{equation}
        p_{n+1}(x)-\alpha_nxp_n(x)=\sum_{i=0}^nb_ip_i(x).
    \end{equation}

    Then
    \begin{align}
    \begin{split}
    b_k&=\frac{\pesc{p_{n+1}-\alpha_nxp_n}{p_k}{[a,b]}}{\pesc{p_k}{p_k}{[a,b]}}=\frac{\pesc{p_{n+1}}{p_k}{[a,b]}}{\pesc{p_k}{p_k}{[a,b]}}-\alpha_n\frac{\pesc{xp_n}{p_k}{[a,b]}}{\pesc{p_k}{p_k}{[a,b]}}\\
    &=\frac{\pesc{p_{n+1}}{p_k}{[a,b]}}{\pesc{p_k}{p_k}{[a,b]}}-\alpha_n\frac{\pesc{p_n}{xp_k}{[a,b]}}{\pesc{p_k}{p_k}{[a,b]}}=0
    \end{split}
    \end{align}

    if \(k<n-1\) (in view of~\ref{prop:pol1}). Thus
    \begin{equation}
        p_{n+1}(x)-\alpha_nxp_n(x)=b_np_n(x)+b_{n-1}p_{n-1}(x).
    \end{equation}

    This proves~\ref{prop:pol3}.
\end{enumerate}
\end{proof}

With this, since all the \(n\) distinct zeros of the polynomial \(Q_n(x)\) of degree \(n\) lie in \(]-1,1[\), we have \(Q_n(1)\neq0\) for all \(n\geq0\). Thus, we may define a new set of polynomials \(\widetilde{Q}_n(x)=Q_n(1)\) so that \(\widetilde{Q}_n(1)=1\). We will name these polynomials again \(Q_n\). Clearly \(Q_0(1)=1\) and \(Q_1(x)=x\).

The following lemma will be useful to arrive to the useful form of the polynomials \(Q_n\):
\begin{lemma}\label{lemma:3.1}
\(\pesc{x^m}{Q_n}{I}=0\) if \(m\) and \(n\) have different parity, \textit{i.e.}, \(m+n\) is odd.
\end{lemma}
\begin{proof}
The proof is by induction on \(n\). Since \(Q_0(x)=1\), \(Q_1(x)=x\) and \(\int_If(x)\dd x=0\) for an odd continuous function \(f\), the lemma holds for \(n=0\) and \(n=1\). By induction hypothesis, we assume the lemma holds for \(Q_i(x)\) for \(i<n\) and we shall show that it also holds for \(Q_n(x)\). We have \(\pesc{x^n}{Q_i}{I}=0\) for \(i=n-1,n-3,\dots\), as \(n\) and \(i\) have different parity. Thus
\begin{equation}
    Q_n(x)=x^n-\sum_{k=1}^{[\frac{n}{2}]}\frac{\pesc{x^n}{Q_{n-2k}}{I}}{\pesc{Q_{n-2k}}{Q_{n-2k}}{I}}Q_{n-2k}(x),
\end{equation}

and it is an odd or even function according as \(n\) is odd or even. Thus \(\int_Ix^mQ_n(x)\dd x=0\) if \(m\) and \(n\) have different parity.
\end{proof}
In consequence, we have the following proposition:
\begin{prop}
The OPS \(\{Q_n\}_n\) satisfies the recurrence relation
\begin{equation}
    Q_{n+1}(x)=\left(\frac{2n+1}{n+1}\right)xQ_n(x)-\left(\frac{n}{n+1}\right)Q_{n-1}(x)\;\forall n\geq1.
\end{equation}

\end{prop}
\begin{proof}
Since \(\{Q_n\}_n\) is a complete OPS, by proposition~\ref{prop:2.3}, there is a recurrence relation of the form
\begin{equation}
    Q_{n+1}(x)=(\alpha_nx+\beta_n)Q_n(x)+\gamma_nQ_{n-1}(x),\ \forall n\geq1
\end{equation}

where \(\alpha_n,\beta_n,\gamma_n\) are real constants depending on \(n\). Now, by lemma~\ref{lemma:3.1}, \(Q_n(x)\) is an odd or even function of \(x\) according as \(n\) is odd or even. Thus \(Q_n(1)=1\) implies that \(Q_n(-1)=(-1)^n\). On substituting \(x=1\) and \(x=-1\) in the recurrence relation, we get
\begin{equation}
    \alpha_n+\beta_n+\gamma_n=1\quad\text{and}\quad\alpha_n-\beta_n+\gamma_n=1.
\end{equation}

Thus, \(\beta_n=0\) and \(\alpha_n+\gamma_n=1\). Also, \(\alpha_n\) is the ratio of the leading coefficients of \(Q_{n+1}(x)\) and \(Q_n(x)\). The leading coefficients can be obtained from lemma~\ref{lemma:3.1} substituting in the points we know and solving a linear system by Cramer's rule. The reason why this method is used is because some recurrences on columns appear and we are able to obtain the determinants via known productories (further details in~\cite{kumar13}). This leading coefficient has the form
\begin{equation}
    \frac{(2n)!}{2^n(n!)^2}
\end{equation}

and therefor,
\begin{equation}
    \alpha_n=\frac{\frac{(2(n+1))!}{2^{n+1}((n+1)!)^2}}{\frac{(2n)!}{2^n(n!)^2}}=\frac{2n+1}{n+1}.
\end{equation}

Thus \(\gamma_n=1-\alpha_n=-n/(n+1)\). This gives the recurrence formula.
\end{proof}
Actually, these polynomials are widely known as \emph{Legendre's polynomials} and we will denote them by \(L_n(x)\). The previous recurrence formula is also known as Bonnet's formula.
\begin{figure}[htb]
    \centering
    \includegraphics[width=\linewidth]{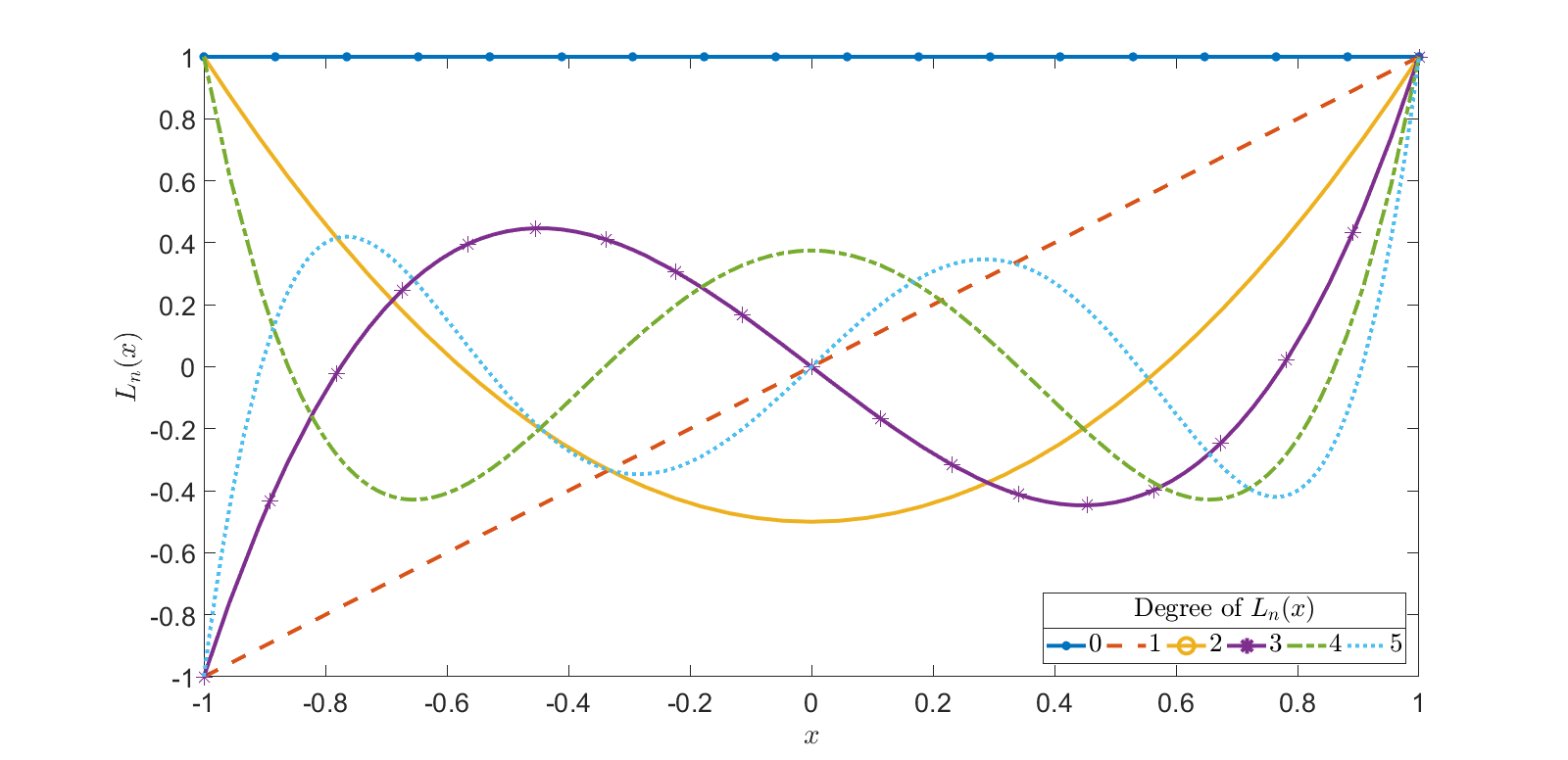}
    \caption{Legendre's polynomials up to degree \(5\).}
    \label{fig:plotleg}
\end{figure}

In figure~\ref{fig:plotleg} we can see the first \(6\) Legendre's polynomials. The norm of Legendre's polynomials is given by:
\begin{equation}
    \|L_n(x)\|^2=\frac{2}{2n+1},
\end{equation}

and an explicit formula for the derivative is
\begin{equation}\label{eq:derleg}
    \frac{\dd}{\dd x}L_{n+1}(x)=(2n+1)L_n(x)+(2(n-2)+1)L_{n-2}(x)+(2(n-4)+1)L_{n-4}(x)+\dots
\end{equation}

The proofs of these facts need more results from differential equations that would get us far from our purpose. We refer the interested reader to~\cite{kumar13} and~\cite{simmons16}.

Returning to the original discussion, the new definition of the matrix in equation~\eqref{eq:defM} is
\begin{equation}
    \mathcal{M}_{i,j}=\pesc{L_i}{L_j}{I}.
\end{equation}

\begin{figure}[htb]
    \centering
    \includegraphics[width=\linewidth]{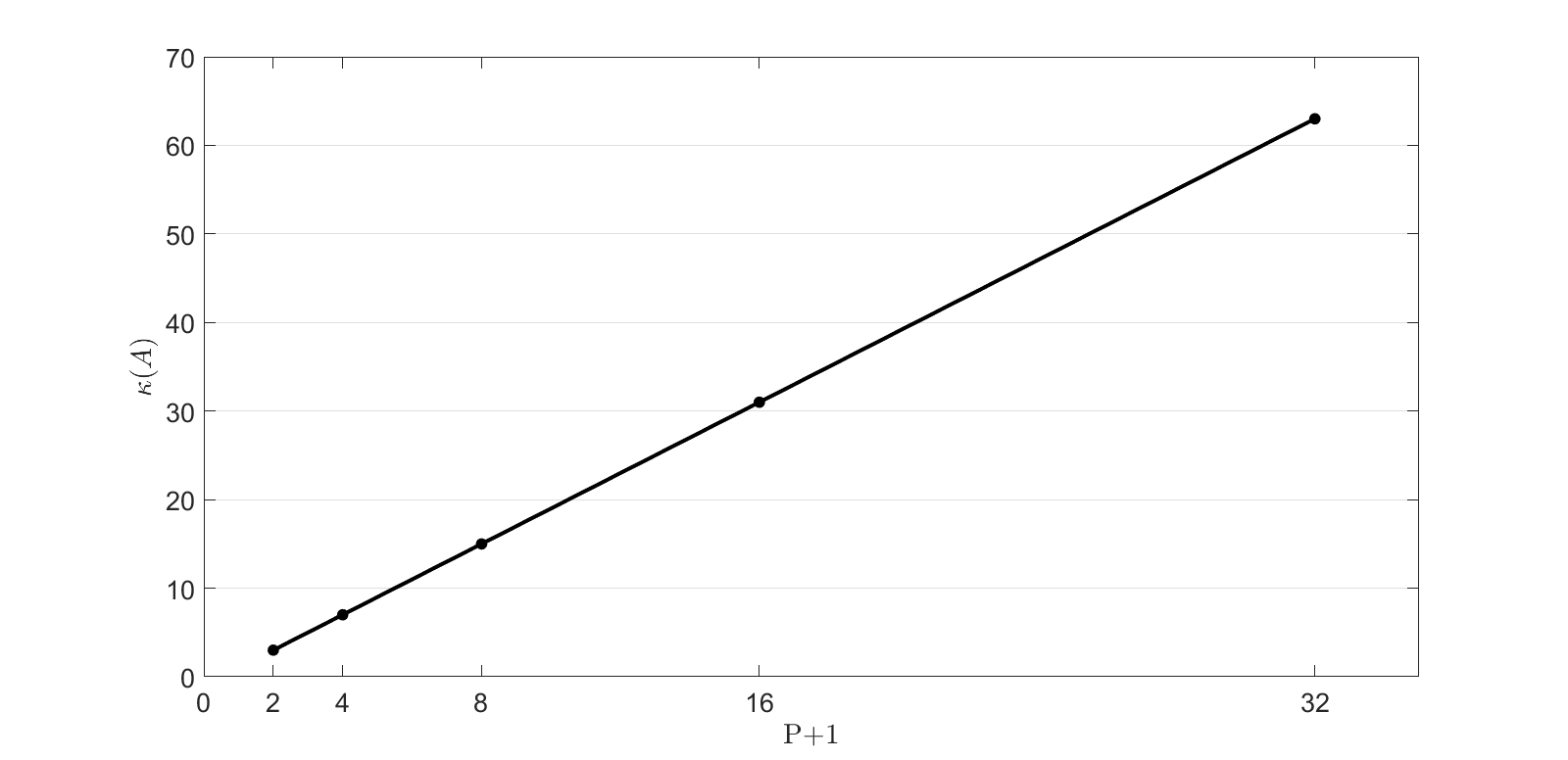}
    \caption{Condition number (in norm \(1\)) of the matrix \(\mathcal{M}\). As we can see, it grows linearly.}
    \label{fig:condleg}
\end{figure}

In this case, as we would expect, the condition in norm \(1\) of \(\mathcal{M}\) is much smaller now, as we can appreciate in figure~\ref{fig:condleg}. We observe now a linear behaviour with respect to the matrix size, in opposition with the exponential when the canonical basis was used (figure~\ref{fig:cond_canonical}). With this, we have solved what to do with the first term of equation~\eqref{eq:after_parts}. In the next subsection we will be dealing with the second term, the evaluation of the functions in the boundary of elements.

\subsection{Numerical flux}\label{subsec:flux}
We recall that the second term of equation~\eqref{eq:after_parts} had the form
\begin{equation}
    \left[\FF(\UU)\cdot\pmb{\psi}\right]_{x_e^l}^{x_e^u}.
\end{equation}

Now we know that the test function will be a Legendre's polynomial so, after resizing to local coordinates, in the boundary (that now will be either \(-1\) or \(1\)) will have the values
\begin{equation}
    L_p(\xi(x_e^l))=L_p(-1)=(-1)^p,\qquad L_p(\xi(x_e^u))=L_p(1)=1.
\end{equation}

The problem comes with \(\FF(\UU)\). This term is evaluated at the boundary of elements and hence, it carries the information between elements. But these values at the extremes of the elements may be not well defined, since although the solution is continuous at each element, it may not be continuous in \(\Fh^W\). So, we have the problem (actually is a \emph{Riemann problem}) of obtaining the value of the physical flux \(\FF\) at one point for which has two different values
\begin{equation}
\FF(\UU_h(x_e^u,t))\neq \FF(\UU_h(x_{e+1}^l,t)).
\end{equation}

For convenience in the notation we will denote by
\begin{equation}
\UU_L=\UU_h(x_e^u,t),\qquad \UU_R=\UU_h(x_{e+1}^l,t)
\end{equation}

the left and right solution at one interelement point, where \(x_e\in\Fh^W\).

The way to solve this is approximating it with the aid of the so-called \emph{numerical flux} \(\HH(\UU_L,\UU_R)\colon\RR^2\times\RR^2\to\RR^2\):
\begin{equation}
    \FF(\UU_h(x_e^u,t))\approx\HH(\UU_L,\UU_R)
\end{equation}

In words, what we are doing is approximating a not well-defined function in the problematic point by a proper function\footnote{We could use two different numerical fluxes, one per component, but for simplicity we will use just one.} that depends on the ``left'' and ``right'' values of the solution. Some desirable properties of the numerical flux are the following:
\begin{enumerate}
    \item \emph{Continuity}: \(\HH(\ell,r)\) is \emph{Lipschitz-continuous} with respect to \(\ell,r\), \textit{i.e.} there exists a constant \(L_H>0\) such that
    \begin{equation}
        \|\HH(\ell,r)-\HH(\ell^*,r^*)\|\leq L_H(\|\ell-\ell^*\|+\|r-r^*\|),
    \end{equation}

    where \(\ell,r,\ell^*,r^*\in\RR^2\), for some norm \(\|\cdot\|\).
    
    \item \emph{Consistency}: it works properly if we have usual continuity:
    \begin{equation}
        \HH(c,c)=\FF(c)
    \end{equation}

    for \(c\in\RR^2\).
    
    \item \emph{Conservativity}: in higher dimensions, this condition means some kind of mass conservation. In one dimension it reduces to commutativity of the arguments:
    \begin{equation}
        \HH(\ell,r)=\HH(r,\ell).
    \end{equation}

\end{enumerate}

Apart from the aforementioned properties, there is no scientific agreement of which particular numerical flux to choose. In the next part we present some reasonable options.

\subsubsection{Some choices of the numerical flux}
One of the most intuitive options is the \emph{central numerical flux} given by
\begin{equation}\label{eq:central_flux}
    \HH(\UU_L,\UU_R)=\frac{\FF(\UU_L)+\FF(\UU_R)}{2},
\end{equation}

that is, an arithmetic mean of the left and right physical fluxes. But this numerical flux is \emph{unconditionaly unstable} and, therefor, worthless for practical purposes (see~\cite{feistauer01}). In the most of applications, in which Navier-Stokes is included, it is suitable to use \emph{upwinding} numerical fluxes. The concept of upwinding is based on the idea that the information on properties of a quantity is propagated in the flow direction. There is an abundance of schemes for this upwinding term --- Vijagasundaram scheme, Steger-Warming scheme, VanLeer scheme or Osher-Solomon scheme are some examples that can be found in~\cite{feistauer01}. Further Riemann solvers can be found also in~\cite{toro13,feistauer03,foerste94,kroener97,wesseling09}. We have chosen and implemented two upwinding schemes mainly due to their simplicity and physical meaning. In the rest of the section we will give more details about these two schemes. 

\begin{itemize}

\item The first one is a slight modification of the central numerical flux given in equation~\eqref{eq:central_flux}.
\begin{equation}
     \HH(\UU_L,\UU_R)=\frac{1}{2}\left[\FF(\UU_L)+\FF(\UU_R)-a(\UU_R-\UU_L)\right],
\end{equation}

where \(a\) is the \emph{local propagation speed} and the direction is reflected in the difference \((\UU_R-\UU_L)\) (since the wave propagates from left to right). The main idea is to use a centred flux to which just enough dissipation is added to ensure stability in all cases. In the scalar case the needed viscosity is given by the largest local wave speed (see~\cite{sonnendruecker13}). With more components, there might be a superposition of the waves, each one with its corresponding eigenvalue. So taking the viscosity coefficient in the flux as the maximum over all eigenvalues should work. With this we have that
\begin{equation}
    a=\max_{\UU\in[\UU_L,\UU_R]}\left|\lambda\left(\frac{\partial\FF(\UU)}{\partial\UU}\right)\right|,
\end{equation}

that is, the maximum modulus of the eigenvalues of the jacobian matrix \(\FF'(\UU)\).

We will refer to the above flux as the \emph{Nessyahu and Tadmor} (NT) central scheme~\cite{nessyahu90}. It is most frequently called Lax-Friedrichs flux (although it is worth mentioning that such flux expression does not appear in Lax~\cite{lax54} but rather on Rusanov~\cite{rusanov62}).

\item The second numeric flux is based on using the information provided by the method of characteristics. Since the characteristics are Riemann invariants, the solution should remain constant along these curves. We recall that the characteristics had the form

\begin{equation}
   W_{1,2}=u\pm4\sqrt{\frac{\beta}{2\rho}}A^{1/4}
\end{equation}

and therefor
\begin{gather}
    A=\left(\frac{W_1+W_2}{4}\right)^4\left(\frac{\rho}{2\beta}\right)^2,\\
    u=\frac{W_1+W_2}{2}.
\end{gather}

Since \(W_1\) is the forward characteristic, it will need to get the information from the left, so if we denote
\begin{equation}
    \UU_L=\begin{bmatrix}A_L\\u_L\end{bmatrix},\qquad\UU_R=\begin{bmatrix}A_R\\u_R\end{bmatrix}
\end{equation}

we have that the upwinding forward characteristic is
\begin{equation}
    W_{1}^{\text{u}}=u_L+4\sqrt{\frac{\beta}{2\rho}}A_L^{1/4}.
\end{equation}

Analogously we have
\begin{equation}
    W_{2}^{\text{u}}=u_R-4\sqrt{\frac{\beta}{2\rho}}A_R^{1/4}.
\end{equation}

From here we build the upwinding variables
\begin{gather}
    A^{\text{u}}=\left(\frac{W_1^{\text{u}}+W_2^{\text{u}}}{4}\right)^4\left(\frac{\rho}{2\beta}\right)^2,\\
    u^{\text{u}}=\frac{W_1^{\text{u}}+W_2^{\text{u}}}{2}.
\end{gather}

So, the numerical flux will be the physical flux applied to the upwinding variables, \textit{i.e.}
\begin{equation}
    \HH(\UU_L,\UU_R)=\FF\left(\begin{bmatrix}A^{\text{u}}\\u^{\text{u}}\end{bmatrix}\right)=\begin{bmatrix}u^{\text{u}}A^{\text{u}}\\ \frac{(u^{\text{u}})^2}{2}+\frac{p}{\rho}\end{bmatrix}.
\end{equation}

We can find this \emph{characteristic flux} in the paper of Sherwin \textit{et al.}~\cite{sherwin03}.
\end{itemize}

To the best of our knowledge there is no theoretical analysis of this last flux (notice that it is done specifically for this problem). Regarding the NT flux there is some work on the stability of the method. We refer the interested reader to the aforecited references. 

\medskip

In this subsection we have dealt with points \(x_e\in\Fh^W\), but the treatment of the inlet and outlet points \(x_e\in\Fh^{io}\) is the same one. The only difference in this case would be, as one could expect
\begin{equation}
    \UU_L=\UU_{\text{inflow}}
\end{equation}

for the first element and
\begin{equation}
    \UU_R=\UU_{\text{outflow}}
\end{equation}

for the last element. Since we are in the case of nonlinear hyperbolic systems, as we can read in~\cite{dolejvsi16}, the theory for the boundary conditions is missing. 

\medskip

In this chapter we have achieved the semi-discrete scheme:
\begin{gather}
\begin{split}
\overbrace{\frac{2}{2p+1}\frac{\partial\hat{\UU}_p(t)}{\partial t}}^{\textrm{Subsection }\ref{subsec:test}}=\frac{2}{x_e^u-x_e^l}\int_I\left[\FF(\UU(\xi,t))\cdot\overbrace{L_p'(\xi)}^{\textrm{Equation }\eqref{eq:derleg}}+\TT(\UU(\xi)L_p(\xi)\right]\dd\xi\\
\underbrace{+(-1)^p\HH(\UU_h(x_e^l,t),\UU_h(x_{e-1}^u,t))-\HH(\UU_h(x_e^u,t),\UU_h(x_{e+1}^l,t))}_{\textrm{Subsection }\ref{subsec:flux}},
\end{split}
\end{gather}

for the Legendre degree \(p=0,\dots,P\) and for every element \(\Omega_e\), \(e=1,\dots,\Nel\). In order to obtain a discrete algorithm we need a quadrature rule for the integral and a scheme for evolve in time. In the next chapter, the chosen rules will be explained, along with some implementation details. Some numerical results are also presented in the following chapter.

%% file: Computational/Integrations.tex
In this chapter we fully discretise the semi-discrete method previously presented. In order to achieve this, we present the quadrature rules used for the spatial and temporal integration. In the second part of this chapter we present some results of applying this method. We have studied its stability, convergence and sensitivity to parameters carrying out a total of nearly \(2000\) simulations.

\section{Spatial integration}\label{subsec:int}
With all we have done in this section, we have a semi-discrete method given by equation~\eqref{eq:after_parts}. But for the implementation of the method, we need to evaluate the integrals
\begin{align}
\pesc{\FF(\UU)}{\pmb{\psi}_x}{\Omega_e}&=\frac{2}{x_e^u-x_e^l}\pesc{\FF(\UU)}{\psi_x}{I}=\frac{2}{x_e^u-x_e^l}\int_{-1}^1\FF(\UU(\xi,t))\cdot L_p'(\xi)\dd\xi,\\
\pesc{\TT(\UU)}{\pmb{\psi}}{\Omega_e}&=\frac{2}{x_e^u-x_e^l}\int_{-1}^1\TT(\UU(\xi,t))\cdot L_p(\xi)\dd\xi.
\end{align}

The quadrature rule chosen, following the steps of~\cite{sherwin03} has been the Legendre-Gauss-Lobatto (LGL) quadrature. First we need to define the Gaussian integration:
\begin{defi}[Gaussian integration]
Let \(x_0,\dots,x_N\) be the roots of the \(N+1\)-th orthogonal polynomial \(p_{N+1}\), and let \(w_0,\dots,w_N\) be the solution of the linear system
\begin{equation}
    \sum_{j=0}^N(x_j)^kw_j=\int_{-1}^1x^kw(x)\dd x,\qquad0\leq k\leq N,
\end{equation}

where \(w(x)\colon A\to\RR^+\) is some weight function, where \(A\subset[-1,1]\) is a discrete set. The positive numbers \(w_j\) are called \emph{weights}.
\end{defi}

Then,
\begin{prop}
The following properties are satisfied:
    \begin{enumerate}
        \item 
        \begin{equation}\label{eq:2.2.10}
            \sum_{j=0}^Np(x_j)w_j=\int_{-1}^1p(x)w(x)\dd x
        \end{equation}

        for all \(p\in\mathcal{P}_{2N+1}(I)\) where \(I=[-1,1]\).
        
        \item It is not possible to find \(x_j\), \(w_j\), \(j=0,\dots,N\) such that~\eqref{eq:2.2.10} holds for all polynomials \(p\in\mathcal{P}_{2N+2}(I)\).
    \end{enumerate}
\end{prop}

The proofs can be found in~\cite{davis07}. This version of Gauss integration is quite well known. However, the roots, which correspond to the collocation points, are all in the open interval \(]-1,1[\). The requirement of imposing boundary conditions at one or both end points creates the need for the generalised Gauss integration formulas which include these points. This lead us to the Gauss-Lobatto integration:
\begin{defi}[Gauss-Lobatto integration]
Let \(-1=x_0,x_1,\dots,x_n=1\) be \(N+1\) roots of the polynomial
\begin{equation}
    q(x)=p_{N+1}(x)+ap_N(x)+bp_{N-1}(x),
\end{equation}

where \(a\) and \(b\) are chosen so that \(q(-1)=q(1)=0\). Let \(w_0,\dots,w_N\) be the solution of the linear system
\begin{equation}
    \sum_{j=0}^N(x_j)^kw_j=\int_{-1}^1x^kw(x)\dd x\qquad0\leq k\leq N.
\end{equation}

Then
\begin{equation}
    \sum_{j=0}^Np(x_j)w_j=\int_{-1}^1p(x)w(x)\dd x,
\end{equation}

for all \(p\in\mathcal{P}_{2N-1}(I)\).
\end{defi}

In the special case of Jacobi weight, \textit{i.e.} \(w(x)=(1-x)^\alpha(1+x)^\beta\) with \(\alpha,\beta\in[-1/2,1/2]\) there is an alternative characterisation of the Gauss-Lobatto points, namely they are the points \(-1\), \(+1\) and the roots of the polynomial
\begin{equation}
    q(x)=p'_N(x).
\end{equation}

For the proofs of the aforementioned results see~\cite{canuto12}. With this, if we choose \(p_N=L_N\) the \(n\)-th Legendre we obtain the \emph{Legendre-Gauss-Lobato} quadrature. Since explicit formulas for the quadrature nodes are not known to the best of our knowledge, such points have to be computed numerically as zeroes of appropriate polynomials. The quadrature weights can be expressed in closed form in term of the nodes, as indicated in the following formulas (see, \textit{e.g.}, Davis and Rabinowitz~\cite{davis07}):
\begin{gather}
    x_0=-1,x_N=1,\{x_j\}_{j=1}^{N-1}\text{ zeroes of }L_N';\\
    w_j=\frac{2}{N(N+1)}\frac{1}{[L_N(x_j)]^2}\qquad j=0,\dots,N.
\end{gather}

Some points and weights are presented in table~\ref{table:LGL}.

\begin{table}[thb]
\centering
\begin{tabular}{|c|c|c|}
\hline
\textbf{Number of nodes}        & \textbf{Nodes}                                                          & \textbf{Weights}                                    \\ \hline
\(1\)                  & \(0\)                                                          & \(2\)                                      \\ \hline
\(2\)                  & \(\pm\sqrt{\frac{1}{3}}\approx\pm0.58\)                        & \(1\)                                      \\ \hline
\multirow{2}{*}{\(3\)} & \(0\)                                                          & \(\frac{8}{9}\approx0.89\)                 \\ \cline{2-3} 
                       & \(\pm\sqrt{\frac{3}{5}}\approx\pm0.77\)                        & \(\frac{5}{9}\approx0.56\)                 \\ \hline
\multirow{2}{*}{\(4\)} & \(\pm\sqrt{\frac{3}{7}-\frac{2}{7}\sqrt{65}}\approx\pm0.34\) & \(\frac{18+\sqrt{30}}{36}\approx0.65\)     \\ \cline{2-3} 
                       & \(\pm\sqrt{\frac{3}{7}+\frac{2}{7}\sqrt{65}}\approx\pm0.86\) & \(\frac{18-\sqrt{30}}{36}\approx0.35\)     \\ \hline
\multirow{3}{*}{\(5\)} & \(0\)                                                          & \(\frac{128}{225}\approx0.57\)             \\ \cline{2-3} 
                       & \(\pm\frac{1}{3}\sqrt{5-2\sqrt{\frac{10}{7}}}\approx\pm0.54\)  & \(\frac{322+13\sqrt{70}}{900}\approx0.48\) \\ \cline{2-3} 
                       & \(\pm\frac{1}{3}\sqrt{5+2\sqrt{\frac{10}{7}}}\approx\pm0.90\)  & \(\frac{322+13\sqrt{70}}{900}\approx0.24\) \\ \hline
\end{tabular}%
\caption{Points and weights for the Legendre-Gauss-Lobatto quadrature in \([-1,1]\).}
\label{table:LGL}
\end{table}

\bigskip

With this, we have the final discrete algorithm

\begin{gather}
\begin{split}
\overbrace{\frac{2}{2p+1}\frac{\partial\hat{\UU}_p(t)}{\partial t}}^{\textrm{Subsection }therefor{subsec:test}}=\overbrace{\frac{2}{x_e^u-x_e^l}\int_I\left[\FF(\UU(\xi,t))\cdot\underbrace{L_p'(\xi)}_{\textrm{Equation }\eqref{eq:derleg}}+\TT(\UU(\xi)L_p(\xi)\right]\dd\xi}^{\textrm{Subsection }\ref{subsec:int}}\\
\underbrace{+(-1)^p\HH(\UU_h(x_e^l,t),\UU_h(x_{e-1}^u,t))-\HH(\UU_h(x_e^u,t),\UU_h(x_{e+1}^l,t))}_{\textrm{Subsection }\ref{subsec:flux}},
\end{split}
\end{gather}

for the Legendre degree \(p=0,\dots,P\) and for every element \(\Omega_e\), \(e=1,\dots,\Nel\). We have not mentioned it, but regarding the initial condition we need to transform it to its degrees of freedom. Following the reasoning of the previous chapter, if \(\UU_0(x)=\UU(x,0)\) we have
\begin{equation}
    \frac{2}{2p+1}\left(\hat{\UU}_{0}\right)_p=\int_I\UU_0(\xi)L_p(\xi)\dd\xi
\end{equation}

for \(p=0,\dots,P\). This result was also obtained by~\cite{tassi03}.

\section{Temporal integration}
We have obtained an algorithm for getting the derivatives of the degrees of freedom. Hence, we need a scheme to evolve in the temporal dimension. We will only recover the physical solution when the simulation is done (in order to plot the solution) using equation~\eqref{eq:DoF}. In most of the applications a Runge-Kutta of order \(2\) or \(3\) is enough, but following~\cite{sherwin03b} we have chosen the \emph{Adams-Bashforth} scheme. These methods were designed by Adams to solve a differential equation modelling capillary action due to Bashforth  in~\cite{bashforth83}. The main reasons of this election are, on one hand, the implementation-friendly expression; and on the other hand, the fact that the Adams-Bashforth method with \(s\) steps has order \(s\). This will allow us to rise the order if it is desired.

This scheme is encompassed in the so-called linear multistep methods. Conceptually, multistep methods attempt to gain efficiency by keeping and using the information from previous steps rather than discarding it (as Euler explicit, for example). Moreover, in the case of linear multistep methods, a linear combination of the previous points and derivative values is used. Using our notation, denoting the iteration by superscripts, a linear multistep method has the form
\begin{gather}
    \hat{\UU}^{n+s}+a_{s-1}\cdot\hat{\UU}^{n+s-1}+a_{s-2}\cdot\hat{\UU}^{n+s-2}+\dots+a_{0}\cdot\hat{\UU}^{n}\\
    =\Delta t\cdot\left(b_s\cdot\hat{\UU}_t^{n+s}+b_{s-1}\cdot\hat{\UU}_t^{n+s-1}+\dots+b_0\cdot\hat{\UU}_t^{n}\right)
\end{gather}

where the coefficients \(\{a_i\}_{i=0}^{s-1}\) and \(\{b_i\}_{i=0}^s\) determine the method, \(\hat{\UU}^0=\hat{\UU}_0\) and \(\Delta t\) is the time step. In the case of the Adams-Bashforth methods, they are explicit methods and hence \(a_{s-1}=-1\) and \(a_{s-2}=\dots=a_0=0\). Regarding the other coefficients, the main idea is to interpolate the derivatives using the Lagrange formula for polynomial interpolation. This yields the expression for the coefficients
\begin{equation}
    b_{s-j-1}=\frac{(-1)^j}{j!(s-j-1)!}\int_0^1\prod^{s-1}_{\substack{i=0\\i\neq j}}(u+i)\dd u,
\end{equation}

for \(j=0,\dots,s-1\). It can be shown that with this construction the \(s\)-step Adams-Bashforth method has order \(s\) (see~\cite{iserles09}). 

\medskip

In the following sections some numerical results have been performed using this discrete method presented along this chapter and the previous one. All the simulations have been done in Matlab R2017a.

%% file: Computational/Experiments.tex
\section{Test case}\label{sec:results}
As a test case we have replicated one of the numerical experiments done by~\cite{sherwin03}. We consider a normalised vessel of unit area, \(A_0=1\) and normalise the mean velocity so that it has a unit value too (\(u_0=1\)). Physiologically we expect the wave speed to be an order of magnitude higher than the mean velocity and so we prescribe a mean wave speed of \(c_0=\sqrt{\beta/(2\rho)}A_0^{1/4}=10\). This can be achieved by selecting \(\beta=100\) and \(\rho=0.5\). As inflow velocity we use an analytic function simulating the heart beat of the form
\begin{equation}
    \uinf(t)=1-0.4\sin(wt)-0.4\sin(2wt)-0.2\cos(2wt)
\end{equation}

where \(w=2\pi/T\) and \(T\) is the time period (see figure~\ref{fig:inflow}).
\begin{figure}[htb]
    \centering
    \includegraphics[width=\linewidth]{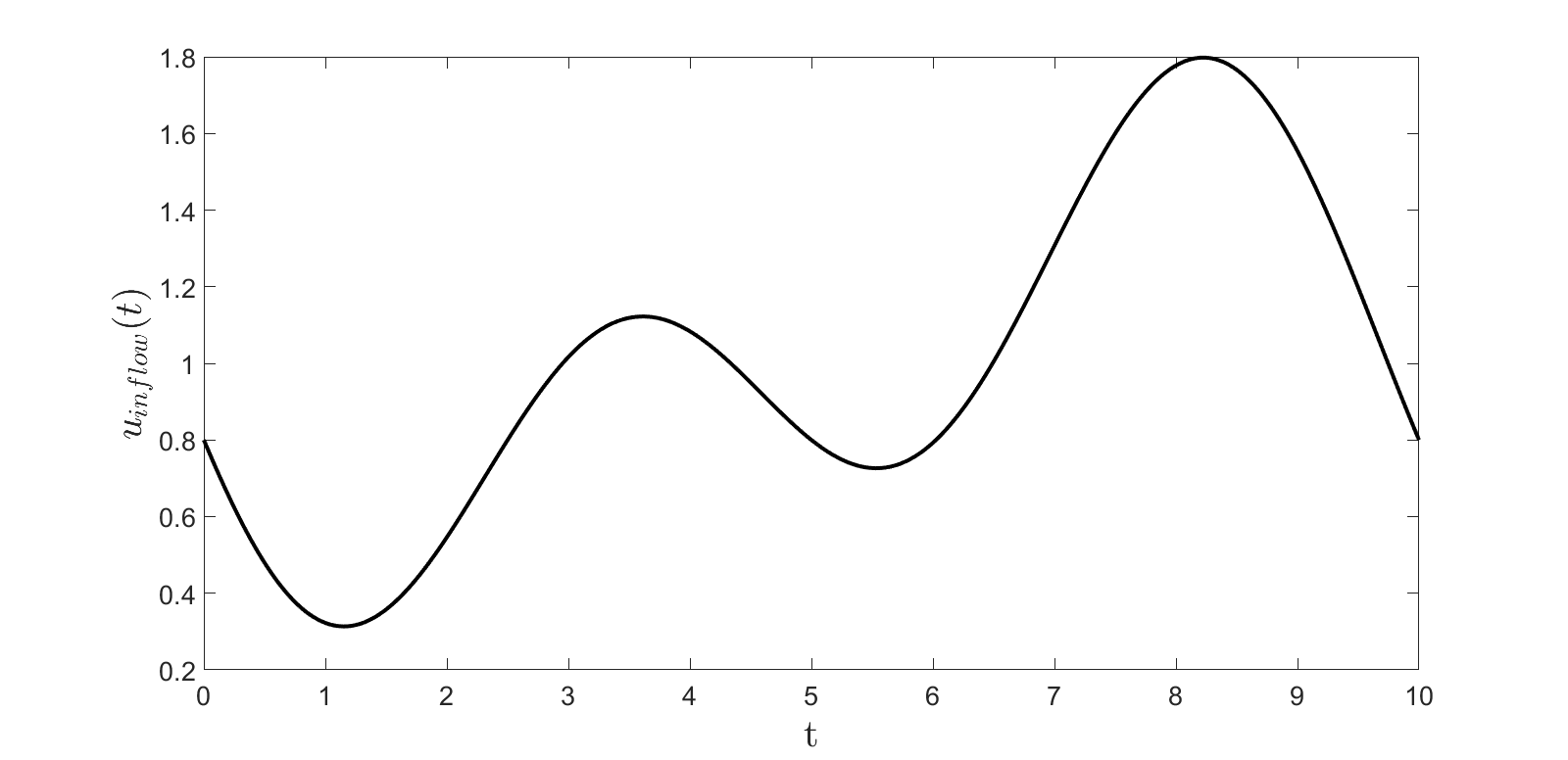}
    \caption{Inflow velocity in the same units as the initial velocity.}
    \label{fig:inflow}
\end{figure}

Making the assumption that the spatial wavelength \(\lambda\) is approximately \(100\) times larger than the vessel diameter, we choose a time span of \(10\) since for the linear case \(\lambda\approx100\).

Since we are considering a problem with a wavelength of \(\lambda=100\), in order to observe the wave as a function of the artery centreline we will consider a computational interval \([-100,100]\). The domain is subdivided into \(\Nel=10\) elements of equal length and a polynomial order of \(P=7\) is applied within each element. We impose the boundary conditions of
\begin{align}
u_l(-100,t)&=\uinf(t),&u_r(100,t)&=1,\\
A_l(-100,t)&=1,&A_r(100,t)&=1.
\end{align}

A second-order time stepping scheme was applied with a time step of \(\Delta t=5\cdot10^{-3}\). 
\begin{figure}[htb]
    \centering
    \includegraphics[width=\linewidth]{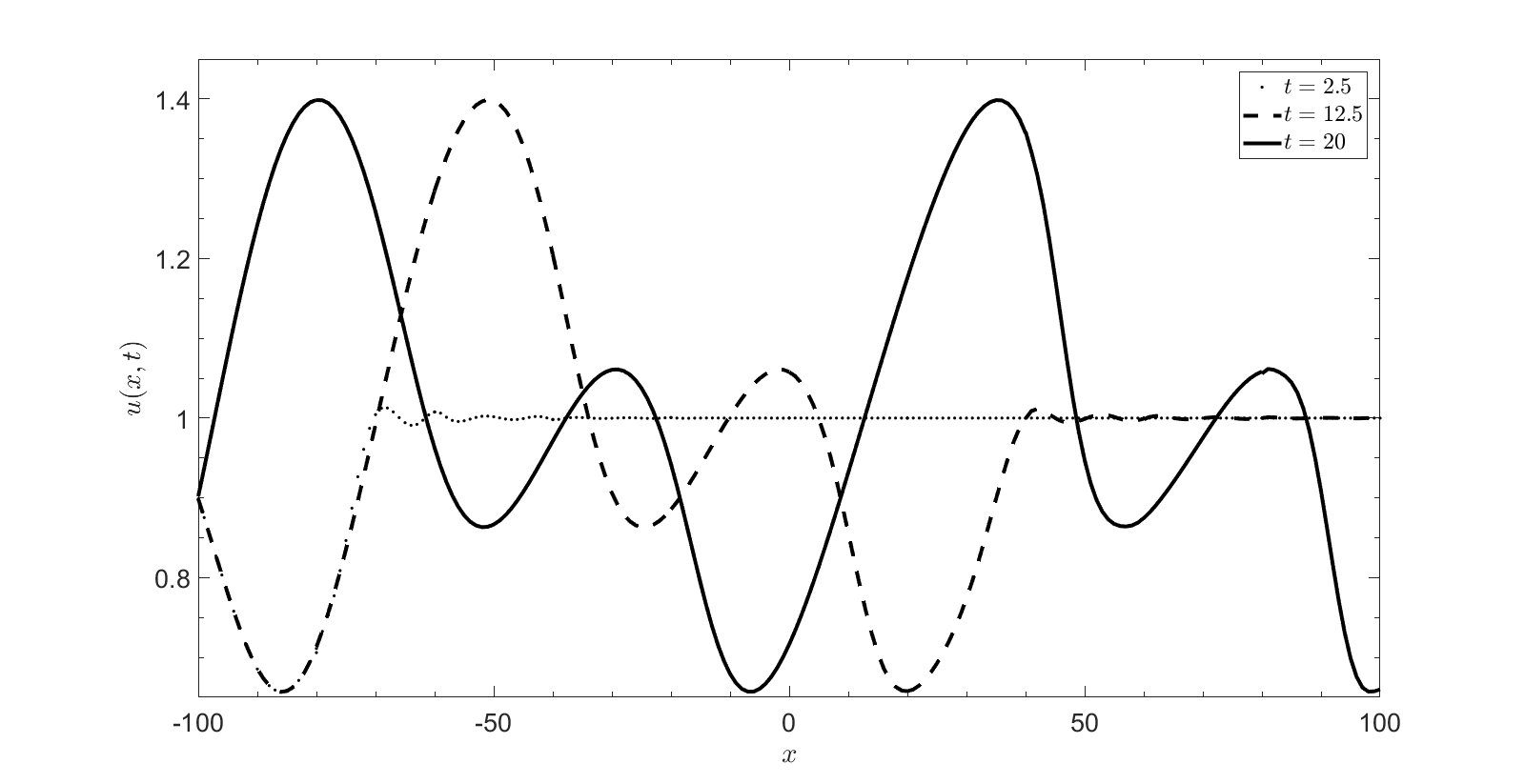}\\
    \includegraphics[width=\linewidth]{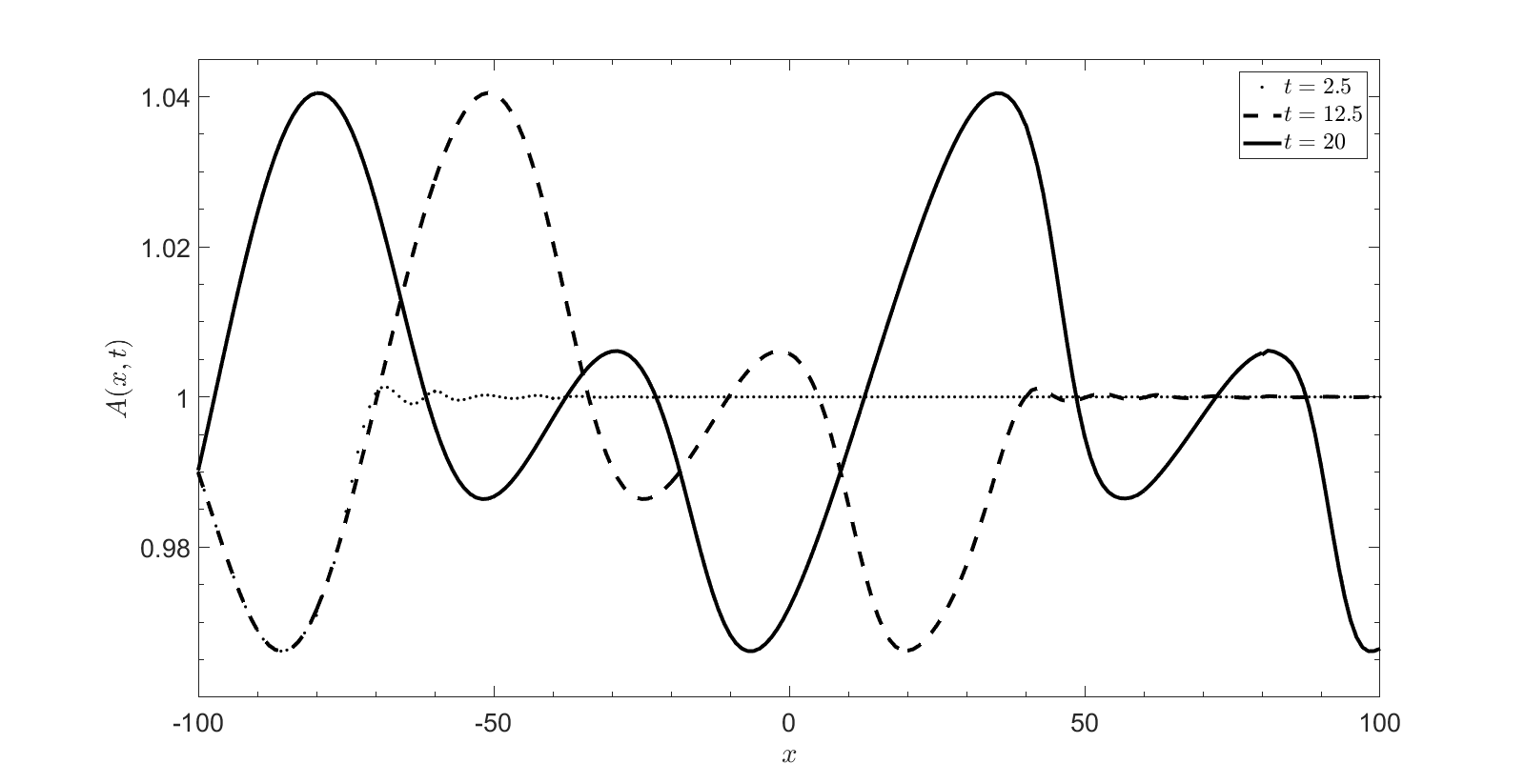}
    \caption{Advection of travelling waves at times \(t=2.5\), \(t=12.5\), \(t=20\).}
    \label{fig:sherwin}
\end{figure}

In figure~\ref{fig:sherwin} we can see how, indeed, we obtain the same results and a travelling wave appears.

Both in these simulations and henceforth, the numerical flux used was the characteristic flux presented in~\ref{subsec:flux}. The plots have been repeated for the NT flux and no significant differences have been observed. 

\section{Convergence and stability}
Regarding the stability and convergence of the method, there is not so much in the literature. We can find some partial results about stability (specially for \emph{nodal} DG) in~\cite{krivodonova13,radice11,gottlieb91}. In these papers some CFL conditions are given, but for specific problems. Hence, in general, their conclusions can not be extrapolated. In order to obtain a rough idea of the stability of this method for our problem, we have performed some simulations. 
\begin{table}[htb]
\centering
\begin{tabular}{|l|c|l|c|}
\hline
\(\Nel\) & \(1,\dots,8\)                     & \(\beta\)  & \(2.2\cdot10^4\) \\
\(P\)    & \(0,\dots,9\)                     & \(BPM\)    & \(80\)           \\
\(\Delta t\)   & \(10^{-4}:5\cdot10^{-4}:10^{-2}\) & Final time & \(5\) seconds    \\ \hline
\end{tabular}
\caption{Parameters' values for the stability experiment.}
\label{table:stability}
\end{table}

In table~\ref{table:stability} we can see the parameters used for the simulation. For the time step we have used Matlab notation (\texttt{ini:step:end}). Following the measures used in subsection~\ref{subsec:poster} and appendix~\ref{appendix:bio} we have estimated \(A_0=A_{\text{outflow}}=A_{\text{inflow}}=4\cdot10^{-2}\) and \(u_0=u_{\text{outflow}}=1\). The other parameters have the same values as in the previous simulations. Therefor, the conclusions obtained are based on \(1600\) simulations. What we have observed is that, in every case, no matter the polynomial degree, number of elements (and hence space grid refinement), or time step the solution obtained is completely stable, \textit{i.e.}, we have not observed neither oscillations nor divergence. Hence we cannot state any kind of CFL condition relating \(\Delta t\) with \(\Nel\). We remark that these are great results, since from a clinical point of view (due to the reasonable choice of the parameters) the method is stable. 

\bigskip

With the convergence of the method we have a similar situation. Some works have been made in this line, see~\cite{lasaint74,johnson86}, but for specific (and scalar) equations. Nevertheless, it is well known that FEMs are chosen for their accuracy (see~\cite{logan11}), and this feature is inherited by DG-FEM. Moreover, in our case we have an additional problem: the boundary conditions. For reproducibility we have considered the same inlet data as~\cite{sherwin03}, a periodic function where we can control the beats per minute. But, because of the form how it is constructed, it is not immediate its manipulation. Due to this, we can not compare it with real data to measure the real error.

What we can do is to measure how discontinuous is the \emph{discontinuous} Galerkin method. For fixed parameters (the same ones as in the previous simulations, unless otherwise indicated) we have refined the mesh, \textit{i.e.}, we have augmented the number of elements \(\Nel\) from two to eight. As a measure of discontinuity we have considered the norm infinity of the differences between the left and right values of blood flow velocity at each element boundary. We have considered a \(5\) seconds simulation, measuring approximately at first, at the middle and at the end of the simulation. Since depending on the time step, each simulation has a variable number of iterations we have specified this value in the variable \textit{It}.
\begin{figure}[htb]
    \centering
    \begin{subfigure}{\textwidth}
    \includegraphics[width=.5\textwidth]{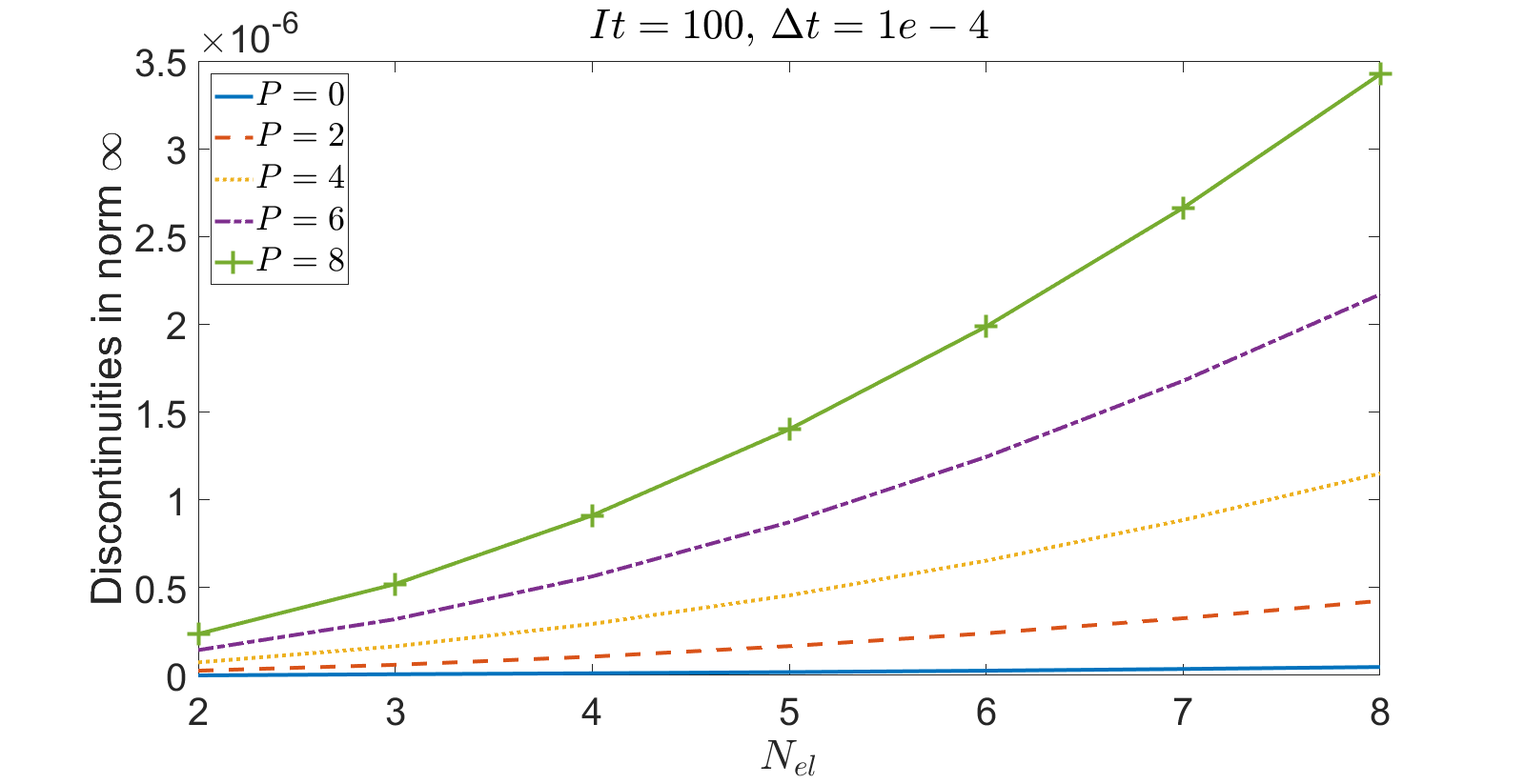}
    \includegraphics[width=.5\textwidth]{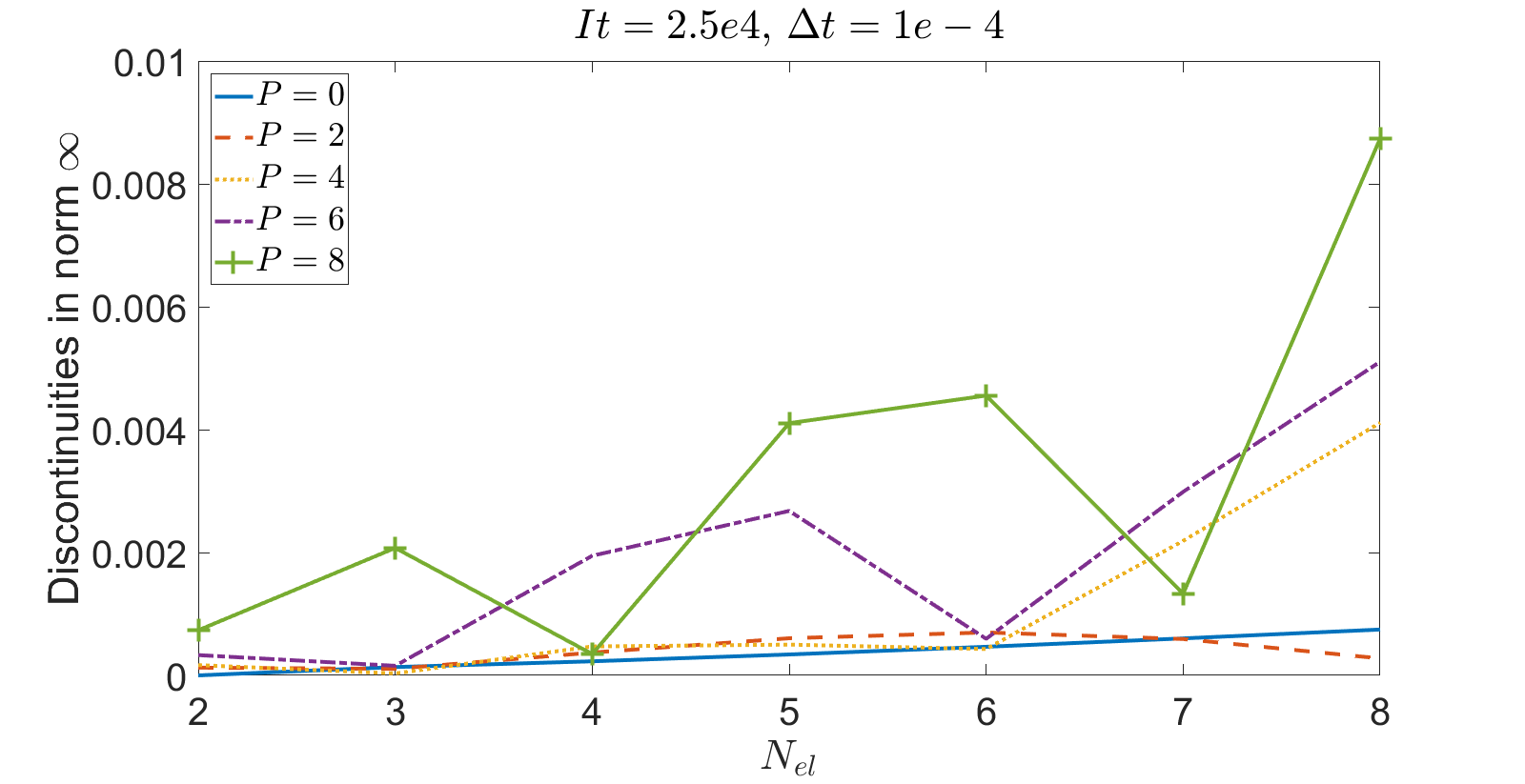}
    \end{subfigure}\\
    \begin{subfigure}{\textwidth}
    \includegraphics[width=.5\textwidth]{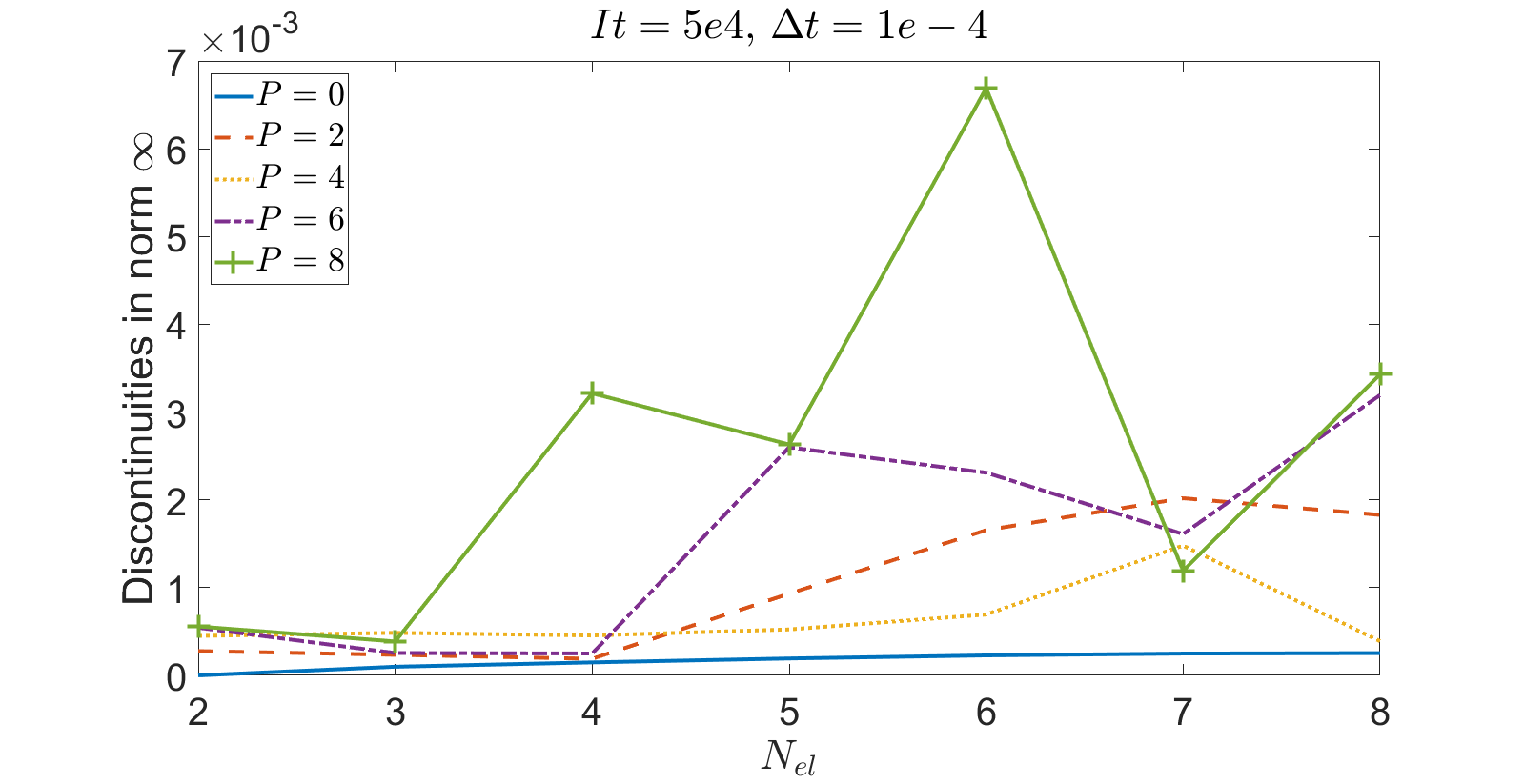}
    \includegraphics[width=.5\textwidth]{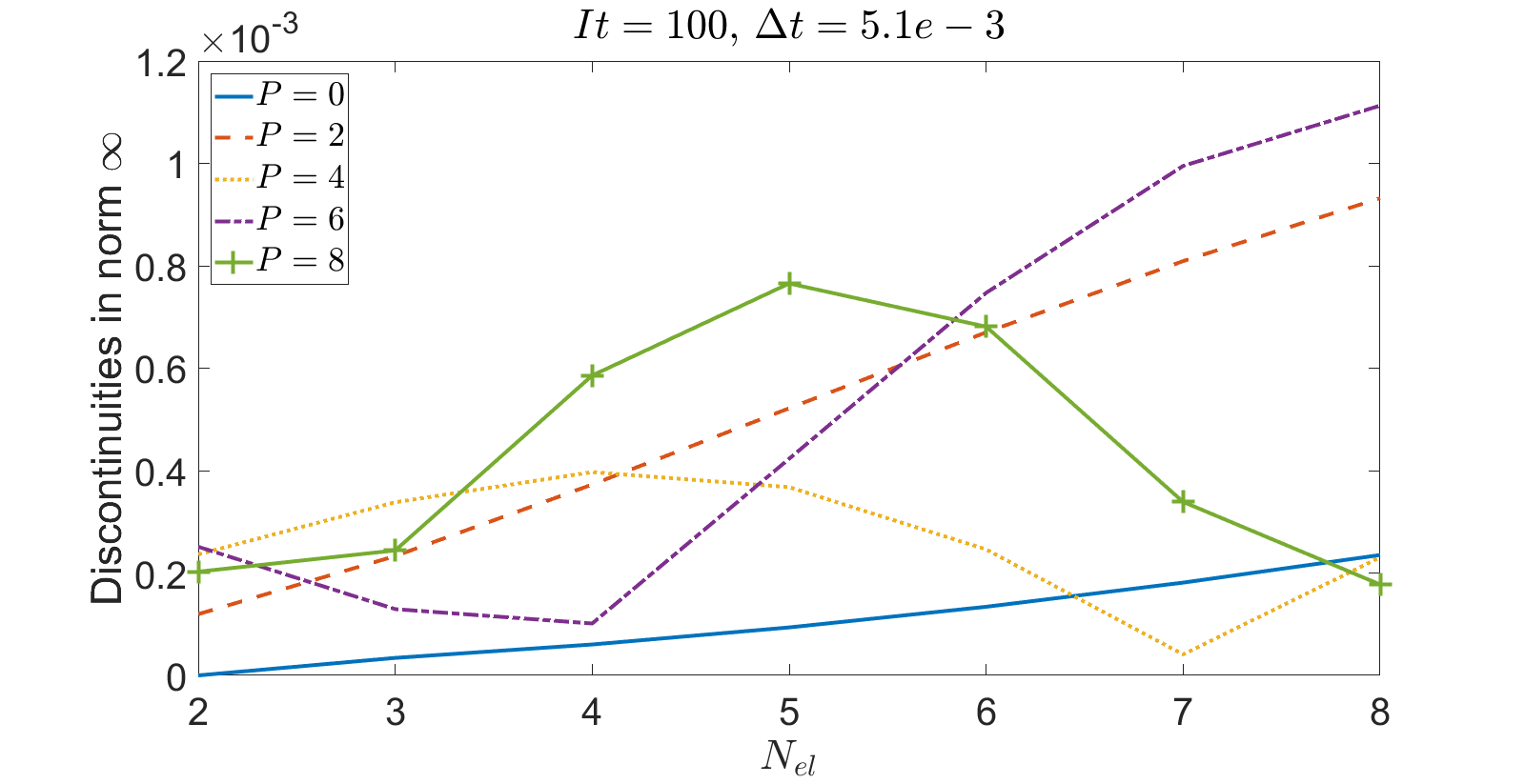}
    \end{subfigure}\\
    \begin{subfigure}{\textwidth}
    \includegraphics[width=.5\textwidth]{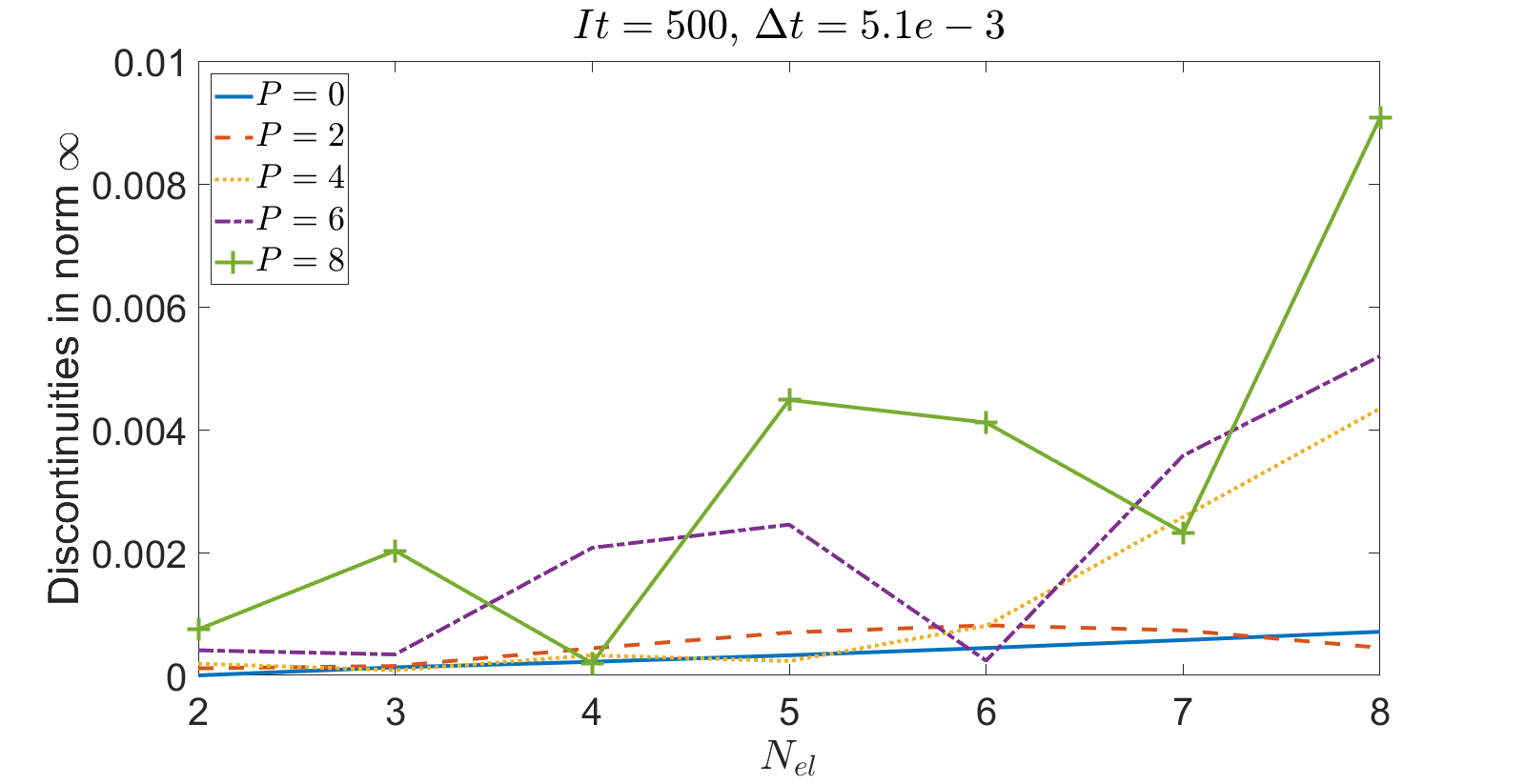}
    \includegraphics[width=.5\textwidth]{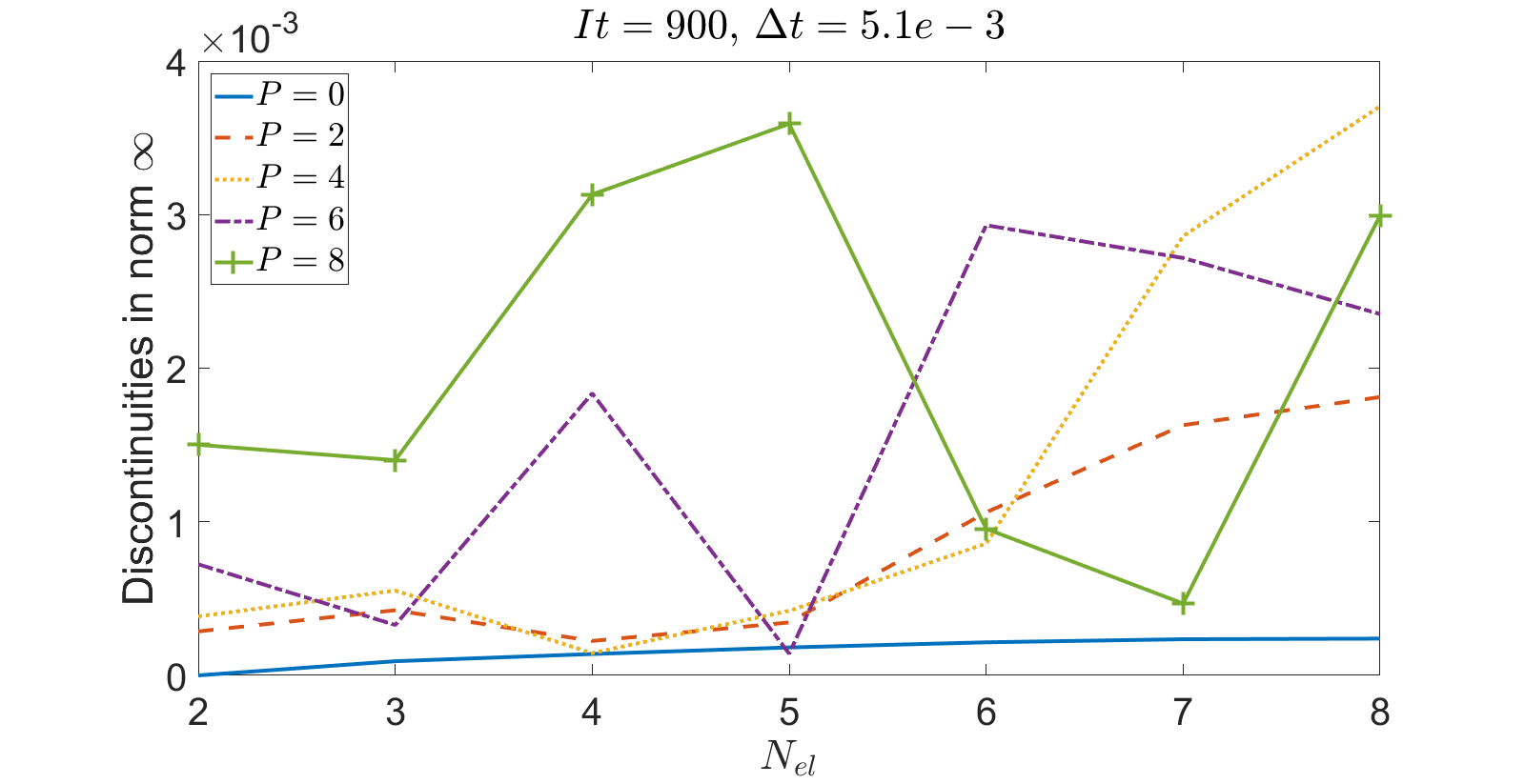}
    \end{subfigure}\\
    \begin{subfigure}{\textwidth}
    \includegraphics[width=.5\textwidth]{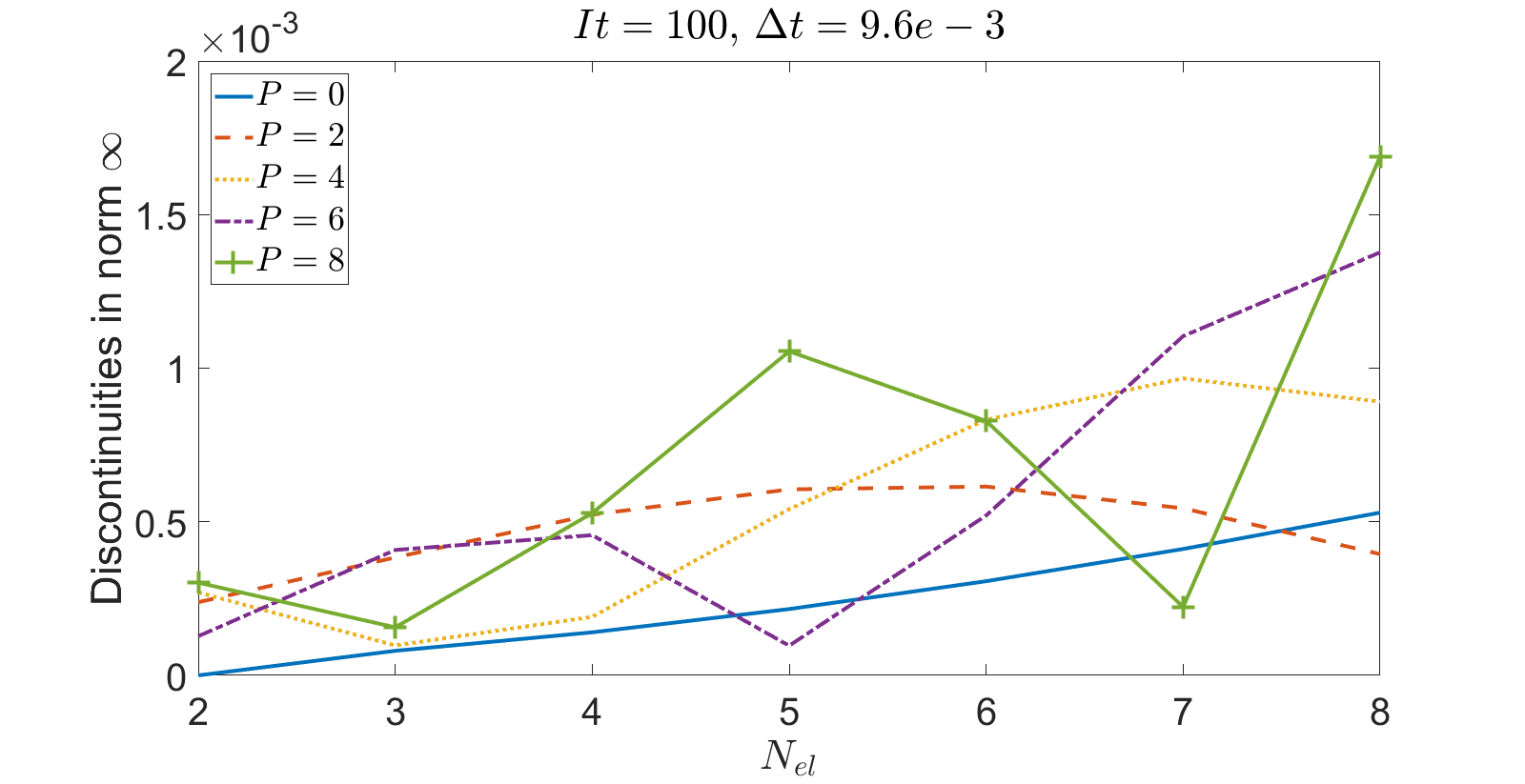}
    \includegraphics[width=.5\textwidth]{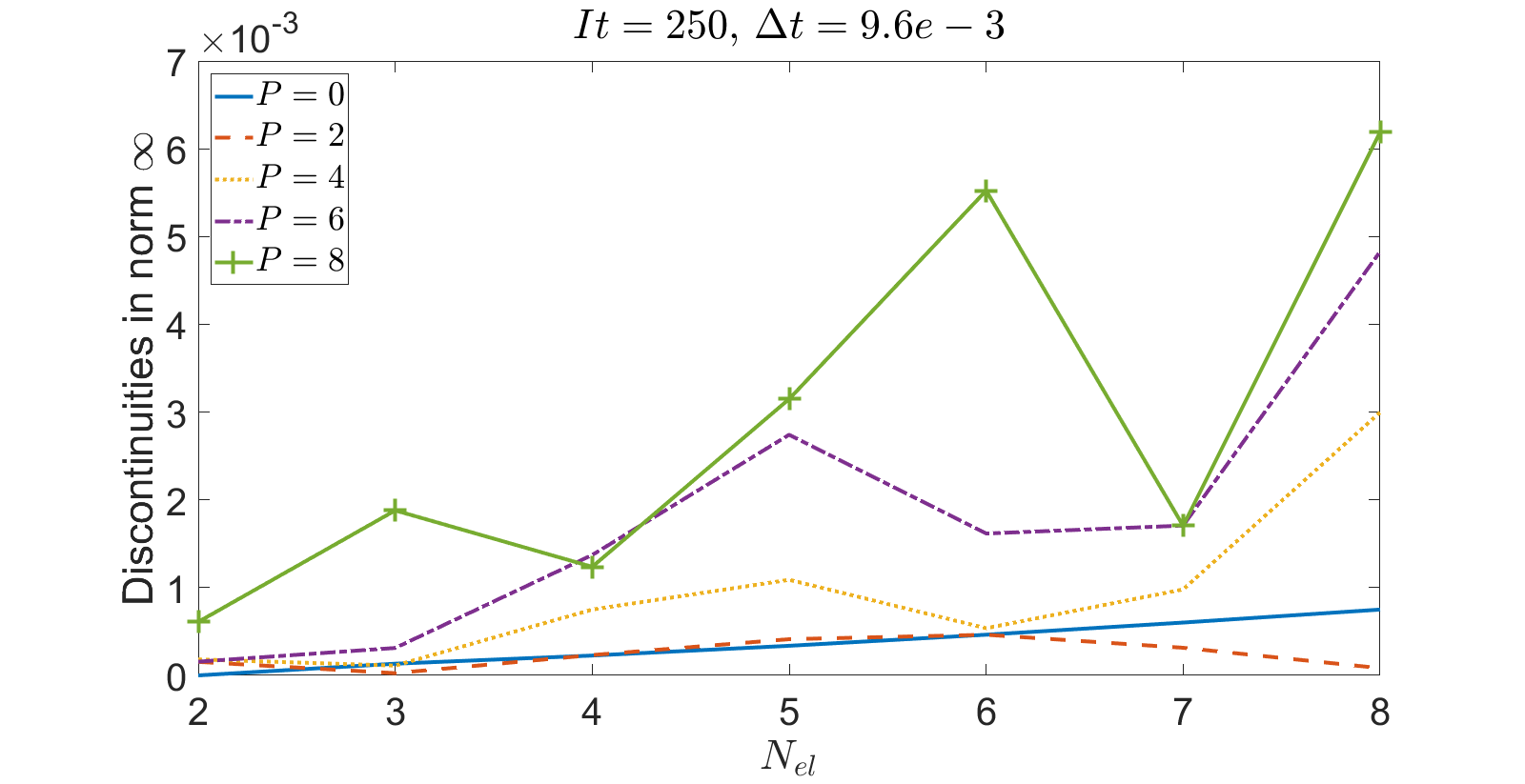}\\
    \includegraphics[width=.5\textwidth]{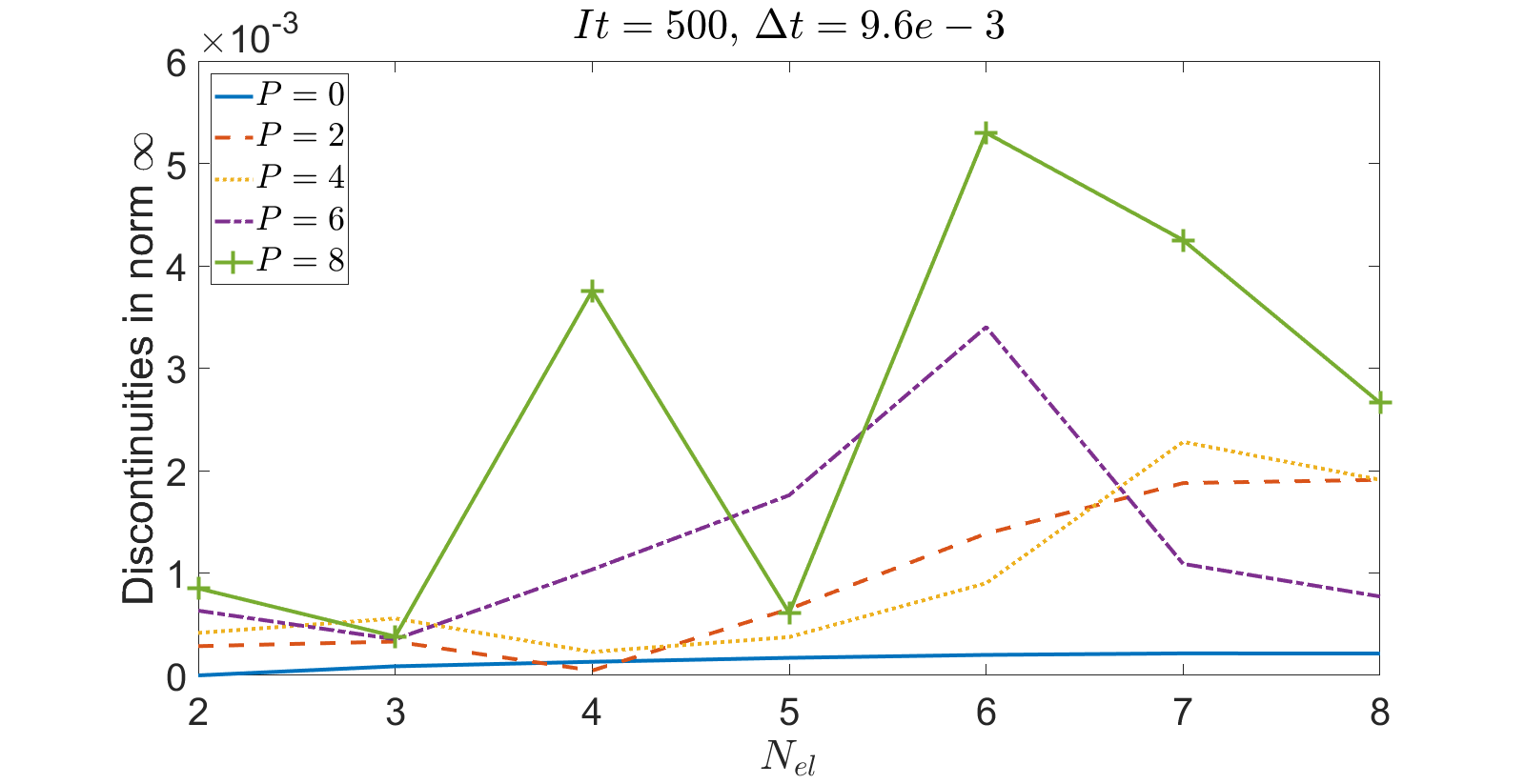}
    \end{subfigure}
    \caption{Discontinuity measure varying \(\Nel\), \(\Delta t\), \(P\) and the time when it is measured.}
    \label{fig:convergence}
\end{figure}

We present the results in figure~\ref{fig:convergence}. As it can be seen, there is not a clear pattern from which some assertion can be done. It seems that the higher is the number of elements, the higher are the discontinuities. This may be caused due to the freedom of the polynomial to oscillate when it has a higher degree. Nevertheless, we notice that the magnitudes are insignificant. We recall that the rest blood velocity was \(1\) m/s.

\section{Biomedical simulations dependent on the parameters}
The next result we show is how the variables change depending on the physical parameters. As independent variables we have chosen two parameters of our model. The first parameter is \(\beta\), which embraces the physical properties of the vessel's wall. Following the measures used in subsection~\ref{subsec:poster} and appendix~\ref{appendix:bio} we have estimated a reasonable range of \(\beta\in[15\cdot10^3,30\cdot10^3]\).  The second variable we have controlled is the inflow data. We have used the aforementioned inflow function for velocity but now controlling the period \(T\). Since it is the periodicity of the wave, biologically it means the heart beats per second. We have done simulations for \(T\in[0.8,2]\) which correspond to values of beats per minute (in what follows, BPM) between \(48\) and \(120\). 

As dependent variables we have considered the maximum value of the flow velocity and vessel amplitude. These simulations are done in a time span of \(5\) seconds with a time step of \(0.005\). The spatial discretisation has been of \(10\) elements with polynomials of degree \(5\).
\begin{figure}[htb]
    \centering
    \includegraphics[width=\linewidth]{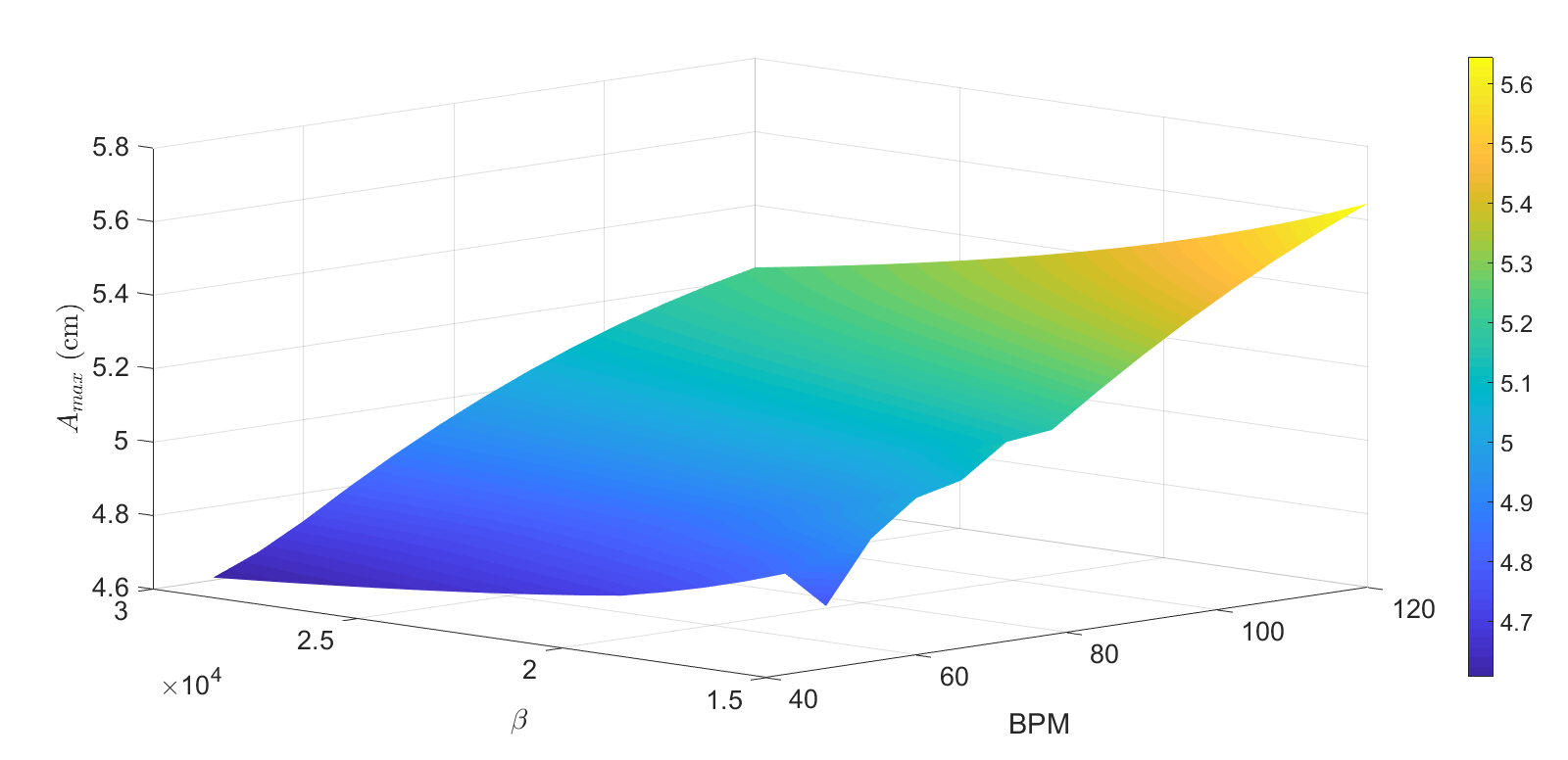}
    \caption{Maximum arterial amplitude (in centimetres) as a function of the BPM and \(\beta\).}
    \label{fig:Amax}
\end{figure}

\begin{figure}[htb]
    \centering
    \includegraphics[width=\linewidth]{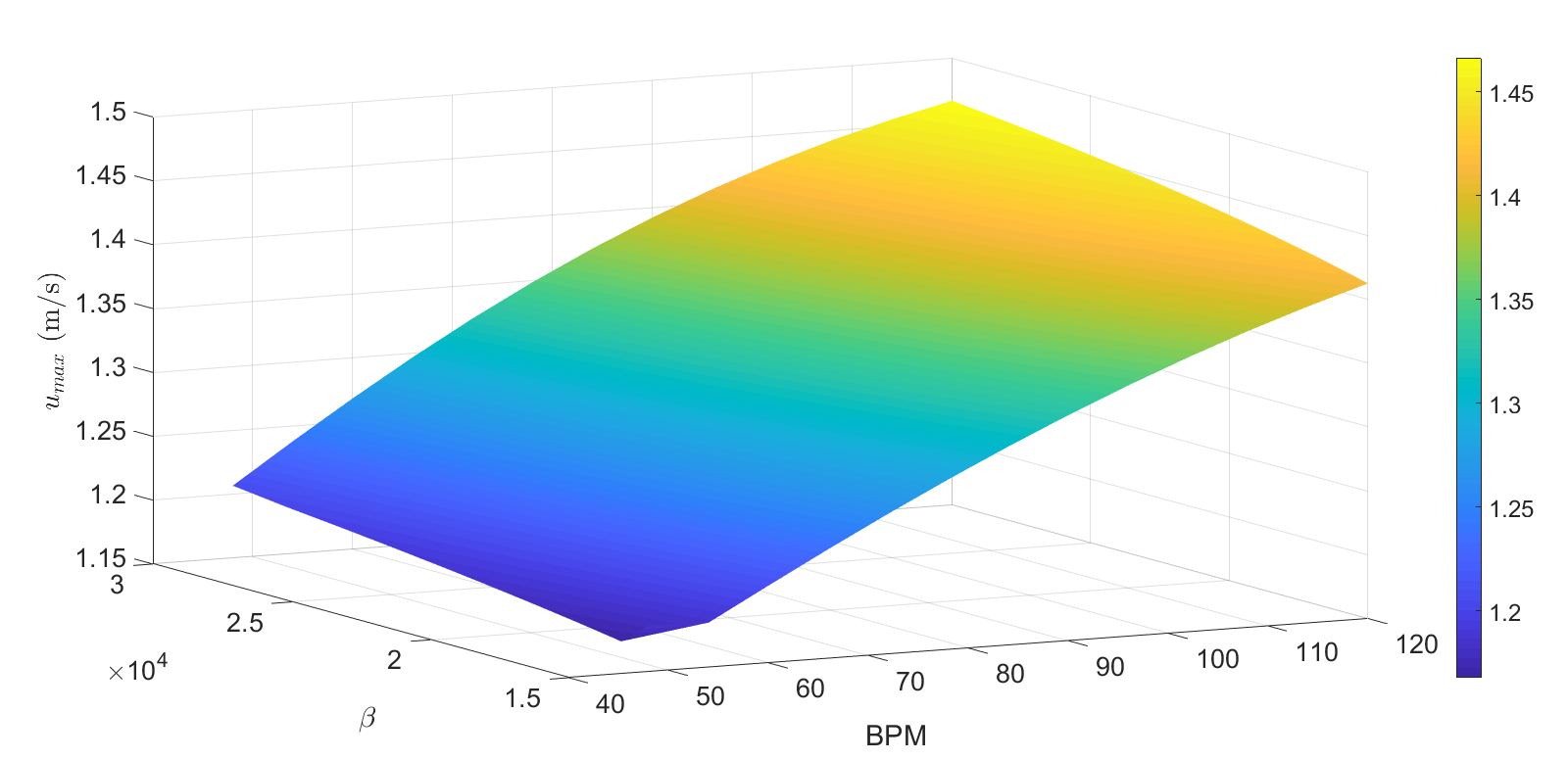}
    \caption{Maximum blood flow velocity (in m/s) as a function of the BPM and \(\beta\).}
    \label{fig:umax}
\end{figure}

Several conclusions are obtained from these plots. The most notorious is maybe the clear growth of both quantities as the heart rate rises. This is what we could expect since the greater is the heart rate, the faster is the blood velocity. Also, if we augment the BPM, the inflow blood acceleration \(\uinf'\) also rises. This is the reason why the walls need to expand. 

Regarding the dependence on \(\beta\), the inverse behaviour is observed between the two quantities. Nevertheless, this also makes sense because the greater is \(\beta\), the more rigid is the wall. This can be achieved by rising  Young's modulus or making more compressible the material, shrinking Poisson's ratio. In both cases, more effort will have to be done to displace the wall, as we observe in figure~\ref{fig:Amax}. In order to conserve the flux, if the amplitude gets smaller, the blood velocity must be greater. Indeed, this is what can be observed in figure~\ref{fig:umax}.

The magnitudes are also coherent with the physical meaning. In the case of the amplitude we see how the greatest difference is about one centimetre (we recall that the rest amplitude was of \(4\) cm). With the velocity, the difference is similar, since the initial velocity was of \(1\) m/s and the values are between \(1.15\) and \(1.45\) m/s.

%% file: Boundary/Conclusions.tex
In this work we have carried out an analysis of a blood flow model in elastic arteries, based on the Navier-Stokes equations. Starting from a short historical survey, we have stated the equations for our model. We have treated them from a more abstract mathematical point of view to extract some results about the arterial behaviour. In the second part of this work we have focused on the numerical method for the simulations, and its deduction. We have ended with a short analysis of this method and with some numerical tests. Next we review the main aspects that have been treated during the text.

 As the reader has been able to see, four paradigms have merge along the whole document:
\begin{itemize}
    \item A small survey point of view. Chapter~\ref{chap:intro} has been presented as a historical background on arterial mechanics, up to \(20\)\textsuperscript{th} century. In subsection~\ref{subsec:tube_laws} the main models for one dimensional artery simulation have been enumerated, situating them chronologically. At the rest of the document, a distinguished amount of references has been cited in order to either justify assumptions or avoid digressions.
    \item Common reasonings in physics. Starting from physical principles such as Newton's second law, we have derived more complex mathematical relations in terms of partial derivative equations, as have been done in chapter~\ref{chap:model}. Once the equations have been established, conclusions have been extracted directly from these equations, such us the conservation of physical quantities.
    \item Mathematical, academic analysis. In chapter~\ref{chap:theo} we have presented more abstract results, suitable in calculus areas. Although, as we specified in the introduction, the objective of the work is the blood simulation, we have considered necessary this part. This analysis has made possible to ensure some necessary conditions for the feasibility of the model, such as the non-collapse of the artery. We recall that in subsection~\ref{subsec:poster} a explicit time and place for the appearance of a shock wave has been obtained. This theorem is an original result of this work based on a slightly different model presented by~\cite{canic03}. This result has been published in \(2017\) in the \textit{4º Congreso de jóvenes investigadores} (IV Conference for young researchers) (see~\cite{rodero17})
    \item Biomedical based numerical performance. Using the data presented in appendix~\ref{appendix:bio} and the numerical scheme of Discontinuous Galerkin (chapter~\ref{chap:galerkin}) some simulations have been performed. Although some wave simulations had been done in~\cite{sherwin03} among others, this kind of simulations have not been performed, to the best of our knowledge.
\end{itemize}

Nevertheless, some lines of future work have been opened and this field of knowledge continues proliferating. The step to three dimensional models has already been made, as in~\cite{formaggia01}. Even though we have not explain it for avoiding complexity, bifurcations are also possible as is explained in~\cite{sherwin03}. Maybe one of the most promising branches is coupling the arterial simulation with fractional calculus, as was illustrated in subsection~\ref{subsec:tube_laws}. This is useful where more viscoelastic behaviours appear, such as in capillaries, aneurysms or simulation of arterial valves. See for example~\cite{doehring05,craiem08,yu16,perdikaris14}. More recent are the works of Perdikaris \textit{et al.} simulating large arterial network using blood flow models together with fractal-tree closures~\cite{perdikaris15}.

Some possible lines of future work are the following:
\begin{enumerate}
    \item Using realistic inflow data. We have used analytic, explicit functions for the inflow data, but maybe real measures are more appropriate for validating the model. Moreover, instead of using averaged data, personalised measures would be interesting for the usefulness of the model. Thus we could predict some diseases or check the consequences of some medical procedures.
    \item Stochastic analysis. In this work we have performed some basic sensitivity analysis of the main parameters of the model. But, as a quick observation in the literature offers, there is a great variability in the values of biophysical parameters. Because of this uncertainty, some stochastic analysis would be interesting, either with the point of view of statistics, or from the random differential equations' point of view.
    \item Characterisation of different diseases. From aneurysm to blood thickening or aorta insufficiency, there are some conditions that could be modelled with blood flow simulations. This could provide some characterisation of such diseases and improving their understanding.
    \item In a more mathematical field, necessary and sufficient conditions for smooth flow. Moreover, if these conditions had biophysical meaning, they would be of a great usefulness.
\end{enumerate}

%% file: apendices/biological_parameters.tex
\chapter{Biological parameters}\label{appendix:bio}
In this appendix, following the suggestions of~\cite{sochi13b} we enumerate some biologically-realistic values for the 1D flow model parameters in the context of simulationg blood flow in large vessels.
\begin{enumerate}
    \item Blood mass density (\(\rho\)): \(1050\text{kg}\cdot\text{m}^{-3}\). \cite{formaggia01,smith02,lee08,badia09,avolio80,canic03,smith04,koshiba07,ashikaga08}.
    \item Blood dynamic viscosity (\(\mu\)): \(0.0035\text{N}/(\text{s}\cdot\text{m}^2)\).\cite{formaggia01,smith02,lee08,badia09,avolio80,koshiba07,janela10,alastruey08,formaggia03,antiga02,westerhof06,moura07}
    \item Young's elastic modulus (\(E\)): \(10^5\text{N}/\text{m}^2\). \cite{formaggia01,lee08,badia09,avolio80,janela10,alastruey08,formaggia03,moura07,formaggia06,zhang02,zhang06,levental07,calo08,hunter10,deng94}.
    \item Vessel wall thickness (\(h_0\)): this, preferably, is vessel dependent, \textit{i.e.} a fraction of the vessel radius according to some experimentally-established mathematical relation. For arteries, the typical ratio of wall thickness to inner radius is about \(0.1-0.15\), and this ratio seems to go down in the capillaries and arterioles. Therefor a typical value of \(0.1\) seems reasonable.\cite{formaggia01,lee08,formaggia03,zhang02,podesser98,blondel03,waite07,badia09,broek11}.
    \item Momentum correction factor (\(\alpha\)): assuming Newtonian flow, about \(4/3\) would give a parabolic profile, while with \(1\) we would get a flat profile. An intermediate value, \textit{e.g.} \(1.2\), may be used to account for non-Newtonian shear-thinning effects \cite{smith02,lee08,canic03,formaggia03,sherwin03b,formaggia06,passerini09,formaggia99}.
    \item Poisson's ratio (\(\nu\)): \(0.45\) \cite{formaggia01,badia09,avolio80,janela10,moura07,formaggia06,zhang02,calo08,deng94,sherwin03b,quarteroni01}.
\end{enumerate}